\documentclass{amsart}
\usepackage{amsmath,amssymb,amsfonts}
\usepackage{graphicx}
\usepackage[alphabetic,y2k,initials]{amsrefs}
\usepackage{xcolor}
\usepackage{pst-all}
\usepackage{pst-solides3d}
\usepackage[all]{xy}

\colorlet{crvena}{black!80}
\colorlet{plava}{black!50}
\colorlet{zelena}{black!20}

\def\rank{\mathrm{rank}}
\def\arctan{\mathrm{arc\,tan\,}}
\def\tan{\mathrm{tan\,}}
\def\dist{\mathrm{dist}}
\def\diag{\mathrm{diag}}

\def\Gr {\mathrm{Gr}}

\numberwithin{equation}{section}

\newtheorem{theorem}{Theorem}[section]
\newtheorem{proposition}[theorem]{Proposition}
\newtheorem{corollary}[theorem]{Corollary}
\newtheorem{lemma}[theorem]{Lemma}
\newtheorem{remark}[theorem]{Remark}
\newtheorem{example}[theorem]{Example}
\newtheorem{definition}[theorem]{Definition}

\definecolor{lightgray}{gray}{0.9}

\definecolor{Lavender}{cmyk}{0,0.28,0,0}
\definecolor{SkyBlue}{cmyk}{0.50,0,0.08,0}
\definecolor{SpringGreen}{cmyk}{0.26,0,0.76,0}

\definecolor{FireBrick}{cmyk}{0,0.809,0.809,0.302}
\definecolor{Choco}{cmyk}{0.0,0.5,0.857,0.176}
\definecolor{DarkOlive}{cmyk}{0.206,0,0.561,0.58}
\definecolor{OliveGrab}{cmyk}{0.246,0,0.754,0.443}
\definecolor{Orange}{cmyk}{0,0.3529,1,0}
\definecolor{Magenta}{cmyk}{0,1,0,0}

\date{}

\begin{document}

\hyphenation{Abe-li-an}

\hyphenation{boun-da-ry}

\hyphenation{Cha-o-tic cur-ves}

\hyphenation{Dy-na-mi-cal dy-na-mics}

\hyphenation{El-lip-ti-cal en-coun-ter}

\hyphenation{Lo-ba-chev-sky}

\hyphenation{Mar-den Min-kow-ski}

\hyphenation{pa-ra-met-ri-zes pa-ra-met-ri-za-tion
Pon-ce-let-Dar-boux}

\hyphenation{quad-ric quad-rics qua-dri-ques}

\hyphenation{sin-gu-la-ri-ties spa-ces}

\hyphenation{tra-jec-to-ry trans-ver-sal}

\author{Vladimir Dragovi\'c}
\address{
Department of Mathematical Sciences, University of Texas at Dallas, EC 35, 800 West Campbell Road, TX 75080 USA
\newline\indent Mathematical Institute SANU, Kneza Mihaila 36, Belgrade,
Serbia} \email{vladimir.dragovic@utdallas.edu}

\author{Milena Radnovi\'c}
\address{Mathematical Institute SANU, Kneza Mihaila 36, Belgrade, Serbia}
\email{milena@mi.sanu.ac.rs}

\title[Bicentennial of the Great Poncelet Theorem]
{Bicentennial of the Great Poncelet Theorem (1813-2013): Current Advances}

\thanks{The research which led to this paper was partially
supported by the Serbian Ministry of Education and Science (Project
no.~174020: \emph{Geometry and Topology of Manifolds and Integrable
Dynamical Systems}).
M.~R.~ is grateful the associateship scheme of the \emph{The Abdus Salam} ICTP (Trieste, Italy) for the support and to Vered Rom-Kedar from The Weizmann Institute of Science (Rehovot, Israel) for hospitality and support.
V.~D.~ is grateful to Marcelo Viana and IMPA (Rio de Janeiro, Brazil) for hospitality and support.
}


\begin{abstract}
The paper gives a review of very recent results related to the Poncelet Theorem, on the occasion of its bicentennial.
We are telling the story of one of the most beautiful theorems of Geometry, recalling for the general mathematical audience the dramatic historic circumstances which led to its discovery, a glimpse of its intrinsic appeal, and importance of its relationship to the dynamics of billiards within confocal conics.
We focus on the three main issues:
A) The case of Pseudo-Euclidean spaces, presenting a recent notion of relativistic quadrics, and applying it to the  description of periodic trajectories of billiards within quadrics.
B) The relationship between so-called billiard algebra and foundations of modern discrete differential geometry which leads to the Double-reflection nets.
C) We introduce a new class of dynamical systems -- pseudo-integrable billiards generated by the boundary composed of several arcs of confocal conics having nonconvex angles.
The dynamics of such billiards has several extraordinary properties.
They are related to the interval exchange transformations and generate families of flows which are minimal but not uniquely ergodic.
This type of dynamics provides a novel type of the Poncelet porisms -- the local ones.  
\end{abstract}

\maketitle

\tableofcontents

\setlength{\parskip}{2pt}

\section{Introduction and history}
\label{sec:intro}
On November 18th, 1812, near Smolensk, a young French officer serving as a battery commander was wounded and his horse killed under him.
This was the last day of the battle of Krasnoi, when the French III Corps of Marshal Ney clashed with three corps commanded by general Count Miloradovich of the Russian Army.
Napoleon's army was heavily defeated and many thousands of his men were left during the withdrawal and imprisoned.
A day after the battle, the officer was found by Russian soldiers.

During his subsequent imprisonment in Saratov, which lasted from April 1813 until June 1814, while recovering from an illness, he recalled the fundamental principles of geometry to which he had been introduced during his studies at \'Ecole Polytechnique.
Without literature, he not only recollected what he had learned from his professors Monge, Carnot, and Brianchon, but also went on developing projective geometry, in particular properties of conics.
The young officer's name was Jean Victor Poncelet.
The notes he made in prison, called \emph{Cahiers de Saratov}\footnote{Saratov notebooks} contained one of deepest, most beautiful, and most important theorems of projective geometry -- the Great Poncelet Theorem, see \cite{PonceletBio}.  

The first published proof of the Great Poncelet Theorem appeared in Poncelet's famous work \emph{Trait\'e des propri\'et\'es projectives des figures} in 1822 \cite{Poncelet1822}.
The complete \emph{Cahiers de Saratov} were published many decades later, in 1862 as \cite{Poncelet1862}.
Meanwhile, Poncelet became a professor of Mechanics at Sorbonne and at the Coll\`ege de France, general of a brigade, the governor of the \'Ecole Polytechnique, the commander of the National Guard of the Department of the Seine, elected member of the Constitutional Assembly, President of the Scientific Commission for the English exhibition of 1851.
Poncelet was also a Grand Officer of the Legion of Honour, Chevalier of the Prussian order, corresponding member of the academies of Sankt Petersburg, Turin, Berlin, and a Foreign member of the Royal Society of London.

It is interesting to mention that, in Chebyshev's report on his business trip to France \cite{Tcheb1852}, Poncelet is described ``as a well-known scientist in practical mechanics''.

Let us present now one of the formulations of the Poncelet Theorem, with an example shown on Figure \ref{fig:sedmouglovi}.
\begin{figure}[h]
\centering
\input{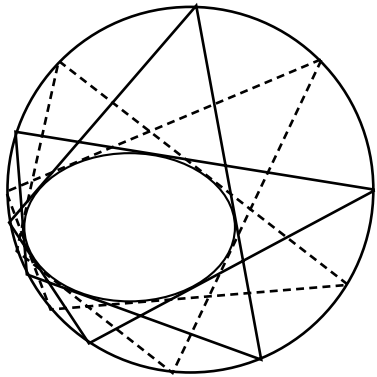}
\caption{Two heptagons inscribed in a conic and circumscribed about another one.}\label{fig:sedmouglovi}
\end{figure}

\begin{theorem}[The Poncelet Theorem]\label{th:poncelet0}
Let $\mathcal C$ and $\mathcal D$ be two conics in the plane.
Suppose that there is a polygon inscribed in $\mathcal C$ and circumscribed about $\mathcal D$.
Then there is infinitely many such polygons and all of them have the same number of sides.
Moreover, each point of $\mathcal C$ is a vertex of such a polygonal line.
\end{theorem}

Later, using the addition theorem for elliptic functions, Jacobi gave another proof of the theorem in 1828 (see \cite{JacobiGW}).
Essentially, the Poncelet Theorem is equivalent to the addition theorems for elliptic curves and Poncelet's proof represents a synthetic way of deriving the group structure on an elliptic curve.

An important question is to find an analytical condition determining, for two given conics, if an $n$-polygon inscribed in one conic and circumscribed about the second exists.
In 1853, such a condition was derived by Cayley, who used the theory of Abelian integrals \cite{Cayley1854}.
He was dealing with the Poncelet Porism in a number of other papers \cites{Cayley1853, Cayley1855,Cayley1857, Cayley1858, Cayley1861}.
Cayley's work served as an inspiration for another great mathematician, Lebesgue, who translated Cayley's proof to the geometric language.
He derived the proof of Cayley's condition using methods of projective geometry and algebra, see the remarkable book \emph{Les coniques} \cite{LebCONIQUES}.
In modern settings, Griffiths and Harris derived Cayley theorem by finding an analytical condition for points of finite order on an elliptic curve \cite{GrifHar1978}.

We have to emphasize that Poncelet originally proved a statement that is much more general than the theorem formulated above, see \cites{BergerGeometry, Poncelet1822}. He derived the latter as a corollary.
Namely, he considered $n+1$ conics of a pencil in the projective plane.
If there exists an $n$-polygon with vertices lying on the first of these conics and each side touching one of the other
$n$ conics, then infinitely many such polygons exist.
We shall refer to this statement as \emph{the Full Poncelet Theorem} and call such polygons \emph{the Poncelet polygons}.

A nice historical overview of the Poncelet Theorem, together with modern proofs and remarks is given in \cite{BKOR1987}.
Various classical theorems of Poncelet type with short modern proofs are reviewed in \cite{BarthBauer1996}, while the algebro-geometrical approach to families of the Poncelet polygons via modular curves is given in \cites{BaMi1993,Jak1993}.
There are also two recent books on the subject \cites{FlattoBOOK,DragRadn2011book}.


The Poncelet Theorem has an important mechanical interpretation.
\emph{The elliptical billiard} \cites{KozTrBIL,Kozlov2003,Tab2005book} is a dynamical system where a material point of the unit mass is moving under inertia, or in other words, with a constant velocity inside an ellipse and obeying the reflection law at the boundary, i.e.\ having congruent impact and reflection angles with the tangent line to the ellipse at any bouncing point, see Figure \ref{fig:reflection}.
\begin{figure}[h]
\centering
\input{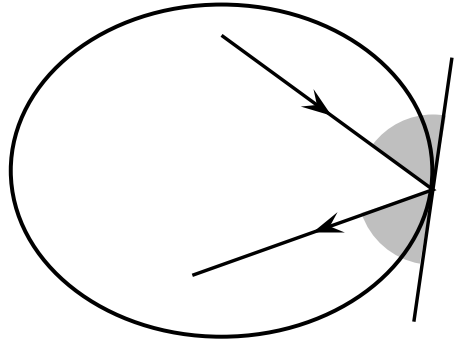}
\caption{Billiard reflection.}\label{fig:reflection}
\end{figure}
It is also assumed that the reflection is absolutely elastic and friction is neglected.

It is well known that any segment of a given elliptical billiard trajectory is tangent to the same conic, confocal with the boundary \cite{CCS1993}, see Figure \ref{fig:caustic3}.
\begin{figure}[h]
\psset{unit=0.9}
\input{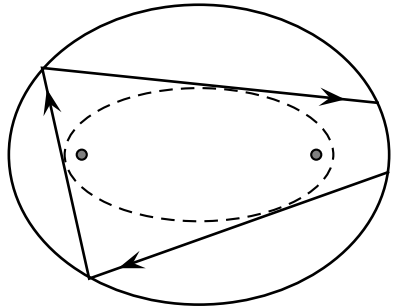}
\caption{Caustic of the billiard trajectory.}\label{fig:caustic3}
\end{figure}
If a trajectory becomes closed after $n$ reflections, then the Poncelet Theorem implies that any trajectory of the billiard system, which shares the same caustic curve, is also periodic with the period $n$.

Moreover, for any given pair of conics, there is a projective transformation of coordinates, such that the conics become confocal in the new coordinates.
Then, polygonal lines inscribed in one of the conics and circumscribed about the other conic will become billiard trajectories.

The Full Poncelet Theorem also has a mechanical meaning.
The configuration dual to a pencil of conics in the plane is a family of confocal second order curves \cite{ArnoldMMM}.
Let us consider the following, slightly unusual billiard.
Suppose $n$ confocal conics are given.
A particle is bouncing on each of these $n$ conics respectively.
Any segment of such a trajectory is tangent to the same conic confocal with the given $n$ curves.
If the motion becomes closed after $n$ reflections, then, by the Full Poncelet Theorem, any such a trajectory with the
same caustic is also closed.

The statement dual to the Full Poncelet Theorem can be generalized to the $d$-di\-men\-sio\-nal space \cite{CCS1993} (see also \cites{Previato1999,Previato2002}).
Suppose vertices of the polygon $x_1x_2\dots x_n$ are respectively placed on confocal quadric hyper-surfaces $\mathcal Q_1$, $\mathcal Q_2$, \dots, $\mathcal Q_n$ in the $d$-di\-men\-sio\-nal Euclidean space, with consecutive sides obeying the reflection law at the corresponding hyper-surface.
Then all sides are tangent to some quadrics $\mathcal Q^1$, \dots, $\mathcal Q^{d-1}$ confocal with $\{\mathcal Q_i\}$; for the hyper-surfaces $\{\mathcal Q_i,\mathcal Q^j\}$, an infinite family of polygons with the same properties exist.
Systematic exposition of this higher-dimensional theory has been presented in a survey paper \cite{DragRadn2010} and in the book \cite{DragRadn2011book}.

But, more than one century before these quite recent results, in 1870, Darboux proved the generalization of the Poncelet Theorem for a billiard within an ellipsoid in the three-dimensional space \cite{Darboux1870}.
It seems that his work on this topic was completely forgotten until very recently, see \cite{DragRadn2011book}.

An interesting generalization of the Poncelet theorem concerning polyhedra that are simultaneously inscribed and circumscribed about two given quadrics in the three-dimensional space was obtained in \cite{GrifHar1977}.
Further generalizations in that direction can be found in \cite{DragRadn2008,DragRadn2011book}.

It is natural to search for a Cayley-type condition related to generalizations of the Poncelet Theorem.
Such conditions for the billiard system inside an ellipsoid in the Euclidean space of arbitrary finite dimension were derived in \cites{DragRadn1998a,DragRadn1998b} using algebro-geometric approach from \cites{Ves1988,MoVes1991}, where billiards within quadrics also considered as discrete time systems.
In recent papers \cites{DragRadn2004,DragRadn2004cor,DragRadn2006jmpa,DragRadn2006jms}, algebro-geometric conditions for existence of periodical billiard trajectories within $k$ quadrics in $d$-di\-men\-sio\-nal Euclidean space were derived.

Most of the results on the subject obtained by 2008 have been presented in the book \cite{DragRadn2011book}.
An important part of the book has been devoted to the Griffiths-Harris program of development of a synthetic approach to higher genera addition theorems. The book also offered a historical overview of the subject, and it included a detailed analysis of the results of Darboux and the contributions of Lebesgue.

About the same time, an extremely interesting book \cite{DuistermaatBOOK} appeared, devoted to discrete integrable systems from the point of view of the QRT maps\footnote{The QRT maps are named after Quispel, Roberts, and Thompson \cite{QRT1988}.}, elliptic surfaces, and elliptic fibrations.
That book provides approach to the Poncelet Theorem as to an important example of the symmetric QRT maps, see \cite[Chapter 10]{DuistermaatBOOK}.
Connections to the elliptic billiards are presented in \cite[Chapter 11.2]{DuistermaatBOOK}.

The present paper is devoted to the bicentennial jubilee of the celebrated Poncelet Theorem and it mostly exposes the current advances of the subject, the results which have been obtained in the last four years.
The interrelation between geometry of pencils of quadrics  and related billiard dynamics is continuing to play a crucial role. 

The projective geometry nucleus of that billiard dynamics is the Double reflection theorem, see Theorem \ref{th:DRT} below.
The meaning of that theorem is that reflections off two confocal quadrics commute.
Namely, there are four lines belonging to a certain linear space and forming a \emph{double reflection configuration}: these four lines reflect to each other according to the billiard law at two confocal quadrics, see Figure \ref{fig:double-reflection-theorem} in Section \ref{sec:drc}.

The double reflection configuration is a cornerstone of a new type of integrable line congruences -- so-called \emph{double reflection nets}.
In these discrete integrable systems of geometric origin, double reflection configurations are playing the role of quad-equations.
The integrability condition is a consequence of an operational consistency of the billiard algebra from \cite{DragRadn2008}, see \emph{the Six-pointed star theorem} -- Theorem \ref{th:zvezda}.
These results are presented in Section \ref{sec:cong}.

Section \ref{sec:pseudo.poncelet} is devoted to pencils of quadrics in the pseudo-Euclidean spaces and to related billiard dynamics and Poncelet configurations.
A novelty of our approach is based on the notion of \emph{relativistic quadrics}, see \cite{DragRadn2012adv}.
The general theory of relativistic quadrics is exposed in the Section \ref{sec:relativistic}.
Section \ref{sec:minkowski.plane} is devoted to the case of the Minkowski plane and serves as an introduction to the higher-dimensional cases.
It also contains a detailed and quite elementary description of periodic trajectories of elliptic billiards and an analysis of topological properties of such systems.
Generalized Cayley-type conditions for pseudo-Euclidean spaces of arbitrary dimension are derived in Section \ref{sec:periodic}.

The last section is devoted to the billiards in Euclidean plane with more complex boundary, formed by arcs of conics from a confocal family, see \cite{DragRadn2012nosonja}.

For the case of billiard systems within confocal conics without non-convex angles on the boundary, it is well known that the famous Poncelet Porism holds:
\begin{itemize}
\item[(A)]
if there is a periodic billiard trajectory with one initial point of the boundary, then there are infinitely many such periodic trajectories with the same period, sharing the same caustic;
\item[(B)]
even more is true, if there is one periodic trajectory, then all trajectories sharing the same caustic are periodic with the same period.
\end{itemize}
See \cites{DragRadn2006jms,DragRadn2006jmpa,DragRadn2011book} where the corresponding conditions of Cayley-type were derived.

However, when non-convex angles exist on the boundary, which is the case studied in Section \ref{sec:nosonja}, one sees  that (A) is still generally true.
However, (B) is not true any more.
\emph{The Poncelet Porism is true locally, but not globally.}
Algebro-geometric conditions of Cayley's type in such a case provide only sufficient but not necessary conditions for periodicity.
A deeper analysis of dynamics in this case is related to a class of interval exchange transformations and to the use of a modified Keane's condition. 
Section \ref{sec:keane} is concluded by an explicit example of the billiard system which satisfies the Cayley type condition, but is still not periodic since it satisfies the Keane's condition as well.
In Section \ref{sec:unique.erg}, we show that there are infinitely many billiard tables bounded by arcs of confocal conics such that the corresponding flow will not be uniquely ergodic.
In particular, Theorem \ref{th:ergodic} shows that the interval exchange transformation corresponding to certain tables is equivalent to the Veech example of minimal and not uniquely ergodic systems \cites{Veech1969,MasurTab2002}.

\section{Billiards and quadrics}
\label{sec:billiards}
As a starting point, we collect some basic notions and facts, which are going to be used in sequel.

\subsection{Elliptical billiards and confocal conics in the Euclidean plane}
\label{sec:billiards.confocal}
We will start with the billiard systems within confocal conics as a mechanical setting for the Poncelet porism.

A general family of confocal conics in the plane can be represented in the following way:
\begin{equation}\label{eq:confocal-in-plane}
\mathcal{C}_{\lambda}\ :\ \frac{x^2}{a-\lambda}+\frac{y^2}{b-\lambda}=1,
\quad\lambda\in\mathbf{R},
\end{equation}
with $a>b>0$ being constants, see Figure \ref{fig:confocal-e}.

\begin{figure}[h]
\input{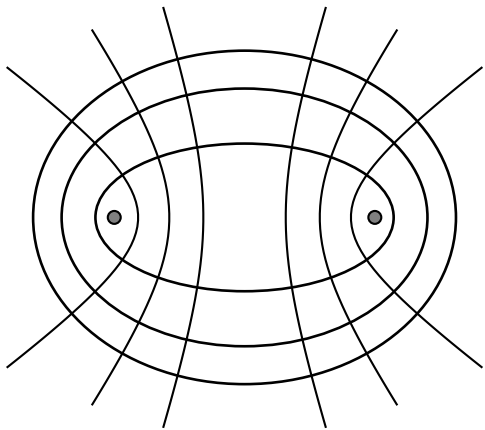}
\caption{Family of confocal conics in the Euclidean plane.}\label{fig:confocal-e}
\end{figure}

By the famous \emph{Chasles' theorem} \cite{Chasles}, each segment of a given billiard trajectory is tangent to a fixed conic that is confocal to the boundary (see also
\cites{KozTrBIL,DragRadn2011book}).
This conic is called \emph{the caustic} of the given trajectory, see Figure \ref{fig:caustic3}.

Now, fix a constant $\alpha_0<b$ and consider billiard trajectories within confocal ellipses
$\mathcal{C}_{\lambda}$ ($\lambda<\alpha_0$) having ellipse 
$\mathcal{C}_{\alpha_0}$ as the caustic.

By the Full Poncelet Theorem, if there is a billiard trajectory with the caustic $\mathcal{C}_{\alpha_0}$ which becomes closed after successive reflections on ellipses $\mathcal{C}_{\lambda_1}$, \dots, $\mathcal{C}_{\lambda_n}$, then each point of $\mathcal{C}_{\lambda_1}$ is a vertex of such a closed trajectory.

The following generalization was proved by Darboux in \cite[Volume 3, Book VI, Chapter I]{DarbouxSUR}.

\begin{theorem}[Darboux theorem on grids]\label{th:darboux.grids}
Suppose that billiard trajectories with the caustic $\mathcal{C}_{\alpha_0}$ become closed after successive reflections on ellipses $\mathcal{C}_{\lambda_1}$, \dots, $\mathcal{C}_{\lambda_n}$.
Then, the intersection points of $i$-th and $j$-th sides of all such trajectories belong to an ellipse $\mathcal{C}_{\lambda^{ij}}$.
\end{theorem}

\begin{example}
Consider the case when $n=7$ and $\lambda_1=\dots=\lambda_6\neq\lambda_7$.
By the Darboux theorem of grids the intersection points of the first and the fourth side, the second and the fifth, the third and the sixth, as well as the third and the seventh side of each corresponding heptagon will belong to the same ellipse.
On Figure \ref{fig:darboux} two such heptagons are shown, and the ellipse containing the mentioned intersection points is gray.
\end{example}

\begin{figure}[h]
\input{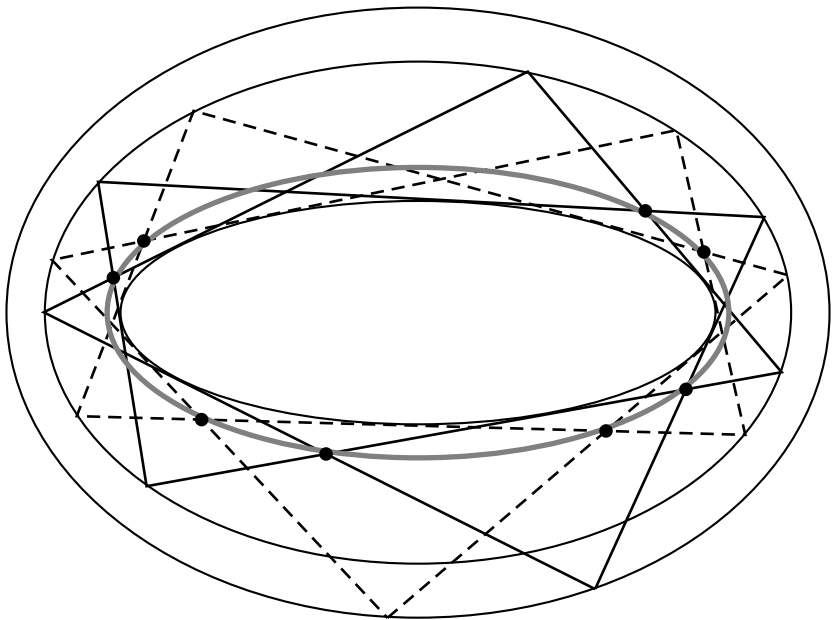}
\caption{Poncelet heptagons and the Darboux grid.}\label{fig:darboux}
\end{figure}

Let us note that the Darboux theorem on grids from \cite{DarbouxSUR} was even more general, since it related to geodesic polygons on Liouville surfaces.
For discussion and the generalization of the Darboux theorem to pairs of non-closed billiard trajectories and to the arbitrary dimension, see \cites{DragRadn2006jmpa,DragRadn2008,DragRadn2011book}.
The Darboux theorem on grids has also recently been subject of interest in \cites{Schw2006,LeviTab}.

\begin{proposition}\label{prop:rotation}
There exist metric $\mu$ on a conic $\mathcal{C}_{\alpha_0}$ and a function 
$$
\rho:(-\infty,\alpha_0)\to\mathbf{R}
$$
satisfying:
\begin{itemize}
\item
metric $\mu$ is non-atomic, i.e.\ $\mu(\{X\})=0$ for each point $X$ on 
$\mathcal{C}_{\alpha_0}$;

\item
$\mu(\ell)\neq0$ for each open arc $\ell$ of $\mathcal{C}_{\alpha_0}$;

\item
for any $\lambda<\alpha_0$, and each triplet of points $X\in\mathcal{C}_{\lambda}$, $A\in\mathcal{C}_{\alpha_0}$, $B\in\mathcal{C}_{\alpha_0}$ such that $XA$, $XB$ are tangent to $\mathcal{C}_{\alpha_0}$, the following equality holds:
$$
\mu(AB)=\rho(\lambda);
$$

\item
$\mu(\mathcal{C}_{\alpha_0})=1$.
\end{itemize}
\end{proposition}

Notice that the third property means that all arcs whose endpoints are on two tangents issued from a point on $\mathcal{C}_{\lambda}$ is the same, see Figure \ref{fig:metric1}.
\begin{figure}[h]
\input{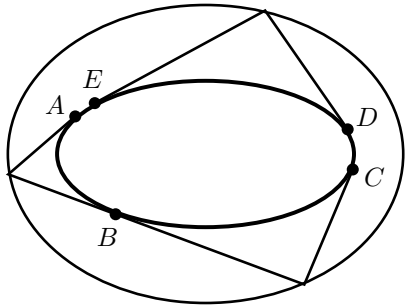}
\caption{$\mu(AB)=\mu(BC)=\mu(DE)$.}\label{fig:metric1}
\end{figure}

\begin{proof}
Take $\lambda_0$ such that there is a closed billiard trajectory in $\mathcal{C}_{\lambda_0}$ with the caustic $\mathcal{C}_{\alpha_0}$.
There is a metric $\mu$ satisfying the first three requested properties for $\lambda=\lambda_0$ -- moreover, such a metric is unique up to the multiplication by constant (see \cite{King1994}).
Thus, there is unique metric satisfying the last property as well.

Suppose the closed trajectory has $n$ vertices -- then by Theorem \ref{th:darboux.grids} there is $\lambda_1$ such that the billiard trajectories within $\mathcal{C}_{\lambda_1}$ with the caustic $\mathcal{C}_{\alpha_0}$ become closed after $2n$ reflections and $\mathcal{C}_{\lambda_0}$ contains intersections of $i$-th and $(i+2)$-nd sides of those trajectories.
Moreover, the metric $\mu$ will satisfy the requested properties for $\mathcal{C}_{\lambda_1}$.

By induction, we get the sequence $\mathcal{C}_{\lambda_k}$ of ellipses, such that billiard trajectories within $\mathcal{C}_{\lambda_k}$ with caustic $\mathcal{C}_{\alpha_0}$ are $2^kn$-periodic and $\mu$ satisfies the listed properties for these ellipses as well.
Because of the Darboux theorem, the metric will satisfy the properties for each $\mathcal{C}_{\lambda}$ that has closed billiard trajectories whose period is multiple of $n$ and the caustic $\mathcal{C}_{\alpha_0}$.

For a periodic trajectory which becomes closed after $n$ bounces on 
$\mathcal{C}_{\lambda}$ and $m$ windings about $\mathcal{C}_{\lambda_0}$,
$\rho(\lambda)=\dfrac{m}{n}$.
Since rational numbers are dense in the reals, $\mu$ will have the required properties for all $\lambda<\alpha_0$.
\end{proof}

\begin{remark}
The function $\rho$ from Proposition \ref{prop:rotation} is called \emph{the rotation function} and its values \emph{the rotation numbers}.
Note that $\rho$ is a continuously strictly decreasing function with
$\left(0,\dfrac12\right)$ as the image:
$$
\lim_{\lambda\to-\infty}\rho(\lambda)=\dfrac12,
\quad
\lim_{\lambda\to\alpha_0}\rho(\lambda)=0.
$$
\end{remark}

\subsubsection*{Elliptical billiard as a Hamiltonian system}

The standard Poisson bracket for the billiard system is defined as:
$$
\{f,g\}=\frac{\partial f}{\partial x}\frac{\partial g}{\partial\dot{x}}-
\frac{\partial f}{\partial\dot{x}}\frac{\partial g}{\partial{x}}+
\frac{\partial f}{\partial y}\frac{\partial g}{\partial\dot{y}}-
\frac{\partial f}{\partial\dot{y}}\frac{\partial g}{\partial{y}}.
$$

Define the following functions:
$$
K_{\lambda}(x,y,\dot{x},\dot{y})=\frac{\dot{x}^2}{a-\lambda}+
\frac{\dot{y}^2}{b-\lambda}-\frac{(\dot{x}y-\dot{y}x)^2}{(a-\lambda)(b-\lambda)}.
$$
These functions represent well-known first integrals of the billiard systems, see \cite{KozTrBIL}. 

\begin{proposition}
Each two functions $K_{\lambda}$ commute:
$$
\{K_{\lambda_1},K_{\lambda_2}\}=0
$$
and for $\lambda_1\neq\lambda_2$, they are functionally independent.
\end{proposition}

It is straightforward to prove the following
\begin{proposition}
Along a billiard trajectory within any conic $\mathcal{C}_{\lambda_0}$, with caustic 
$\mathcal{C}_{\alpha_0}$ and the speed of the billiard particle being equal to $s$, the value of each function $K_{\lambda}$ is constant and equal to
$$
K_{\lambda}=\frac{\alpha_0-\lambda}{(a-\lambda)(b-\lambda)}\cdot s^2.
$$
\end{proposition}

\begin{corollary}
Each $K_{\lambda}$ is integral for the billiard motion in any domain with border composed of a few arcs of confocal conics.
\end{corollary}

Let us recall that a Hamiltonian system on a $2n$-dimensional symplectic manifold is completely integrable if it possesses $n$ functionally independent first integrals such that the Poisson bracket between any two of them is zero.
It is well known that the Liouville-Arnold theorem (see \cite{ArnoldMMM}) describes regular compact leaves of a completely integrable Hamiltonian system, which are common level sets of the first integrals, as tori, with the dynamics being  quasi-periodic and uniform on these invariant tori.

Although being with one-sided constraints, the billiard systems can be seen as Hamiltonian.
Previous considerations show that the billiard system within an ellipse can be considered as a completely integrable system, since it has $n=2$ functionally independent and commuting first integrals.
The symplectic manifold is the four-dimensional cotangent bundle of the plane.
This system can be reduced to the two-dimensional symplectic manifold of lines in the plane, with the standard area form in the role of the symplectic form.
The metric from the Proposition \ref{prop:rotation} can be related to the flat structure on a one-dimensional invariant torus, represented by the caustic conic.

The invariant tori can be also viewed geometrically: the domain between the billiard border and the caustic, which is filled with the corresponding trajectories, is the projection of such tori.
With the fixed caustic, each point within the domain is the projection of four points from the corresponding level set of the phase space, see Figure \ref{fig:ring-phase}.
\begin{figure}[h]
\input{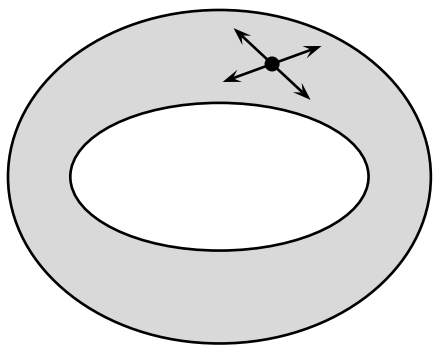}
\caption{Four possible directions of motion from a given point with the fixed caustic.}\label{fig:ring-phase}
\end{figure}

When the caustic is an ellipse, then the ring where the trajectories are placed is the projection of two Liouville tori -- each one corresponding to one direction of winding around the caustic, see Figure \ref{fig:ring-phase-ellipse}.
\begin{figure}[h]
\input{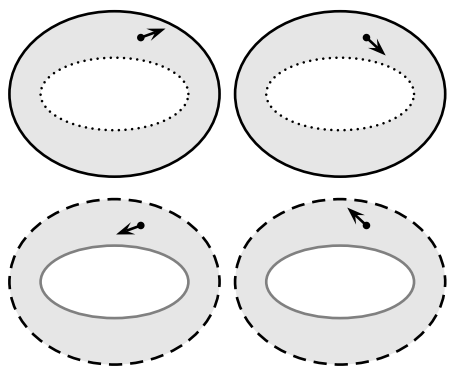}
\caption{Gluing of rings along their borders that gives two tori in the phase space.}\label{fig:ring-phase-ellipse}
\end{figure}
If the caustic is a hyperbola, then the curvilinear quadrangle bounded by the branches of the hyperbola and the ellipse is the projection of a single torus, see Figure \ref{fig:ring-phase-hyperbola}.
\begin{figure}[h]
\input{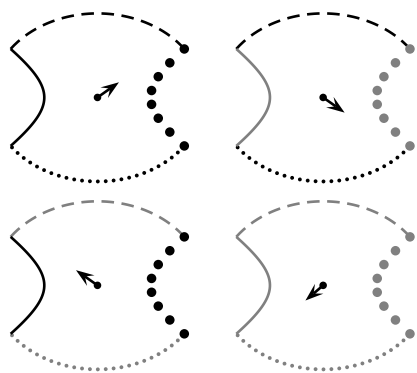}
\caption{Gluing of curvilinear quadrangles rings along their borders that gives a torus in the phase space.}\label{fig:ring-phase-hyperbola}
\end{figure}

In Section \ref{sec:nosonja} we will introduce a more general class of systems -- \emph{the pseudo-integrable systems}, and formulate a generalization of the Liouville-Arnold theorem -- see Theorem \ref{th:maier}.

\subsection{Confocal quadrics the Euclidean space and billiards}
\label{sec:confocal.euclid}
A general family of confocal quadrics in the $d$-dimensional Euclidean space is given by:
\begin{equation}\label{eq:conf.Euclid}
\frac{x_1^2}{b_1-\lambda}+\dots+\frac{x_d^2}{b_d-\lambda}=1,\quad\lambda\in\mathbf{R}
\end{equation}
with $b_1>b_2>\dots>b_d>0$, see Figure \ref{fig:konfE3}.

\begin{figure}[h]
\includegraphics[width=7.58cm, height=7.95cm]{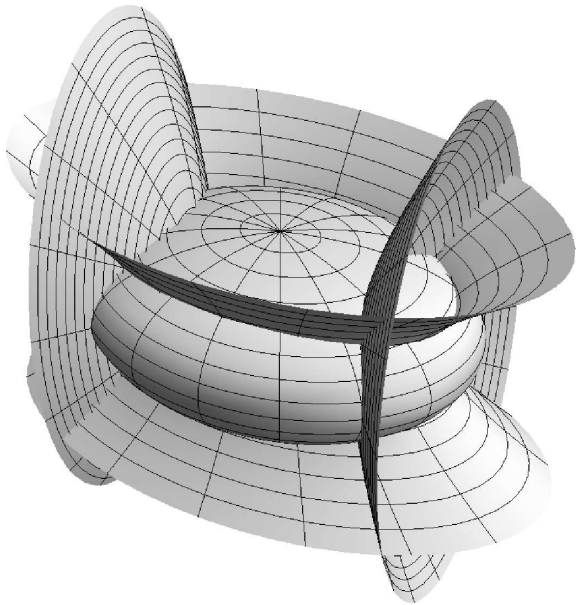}
\caption{Confocal quadrics in the three-dimensional Euclidean space.}\label{fig:konfE3}
\end{figure}

Such a family has the following properties:
\begin{itemize}
\item[E1]
each point of the space $\mathbf{E}^d$ is the intersection of exactly $d$ quadrics from (\ref{eq:conf.Euclid});
moreover, all these quadrics are of different geometrical types;

\item[E2]
family (\ref{eq:conf.Euclid}) contains exactly $d$ geometrical types of non-degenerate quadrics 
-- each type corresponds to one of the disjoint intervals of the parameter $\lambda$:
$(-\infty,b_d)$, $(b_d,b_{d-1})$, \dots, $(b_2,b_1)$.
\end{itemize}

The parameters $(\lambda_1,\dots,\lambda_d)$ corresponding to the quadrics of (\ref{eq:conf.Euclid}) that contain a given point in $\mathbf{E}^d$ are called \emph{Jacobi coordinates}.
We order them $\lambda_1>\dots>\lambda_d$.

Now, let us consider the motion of a billiard ball within ellipsoid $\mathcal{E}$.
Without losing generality, take that the parameter $\lambda$ corresponding to this ellipsoid to be equal to $0$.
Recall that, by Chasles' theorem, each line in $\mathbf{E}^d$ is touching some $d-1$ quadrics from (\ref{eq:conf.Euclid}).
Moreover, for a line and its billiard reflection on a quadric from (\ref{eq:conf.Euclid}), the $d-1$ quadrics are the same.
This means that each segment of a given trajectory within $\mathcal{E}$ has the same $d-1$ \emph{caustics} --
denote their parameters by $\beta_1$, \dots, $\beta_{d-1}$, and introduce the following:
$$
\{\bar{b}_1,\dots,\bar{b}_{2d}\}=\{b_1,\dots,b_d,0,\beta_1,\dots,\beta_{d-1}\},
$$
such that $\bar{b}_1\ge\bar{b}_2\ge\dots\ge\bar{b}_{2d}$.
In this way, we will have $0=\bar{b}_{2d}<\bar{b}_{2d-1}$, $b_1=\bar{b}_1>\bar{b}_1$.
Moreover, it is always: $\beta_i\in\{\bar{b}_{2i},\bar{b}_{2i+1}\}$, for each $i\in\{1,\dots,d\}$, see \cite{Audin1994}.

Now, we can summarize the main properties of the flow of the Jacobi coordinates along the billiard trajectories:
\begin{itemize}
 \item[E3]
along a fixed billiard trajectory, the Jacobi coordinate $\lambda_i$ ($1\le i\le d$) takes values in segment 
$[\bar{b}_{2i-1},\bar{b}_{2i}]$;

 \item[E4]
along a trajectory, each $\lambda_i$ achieves local minima and maxima exactly at touching points with corresponding caustics, intersection points with corresponding coordinate hyper-planes, and, for $i=d$, at reflection points;

 \item[E5]
values of $\lambda_i$ at those points are $\bar{b}_{2i-1}$, $\bar{b}_{2i}$;
between the critical points, $\lambda_i$ is changed monotonously.
\end{itemize}

Those properties represent the key in the algebro-geometrical analysis of the billiard flow.

\subsection{Double reflection configurations}
\label{sec:drc}
Billiards within pencils of quadrics induce fruitful dynamical systems in arbitrary dimension.
They are meaningful in spaces with non-Euclidean metric as well, and even in spaces without any metric at all.   

In this section, we review a fundamental projective geometry configuration of double reflection in the $d$-dimensional projective space $\mathbf{P}^d$ over an arbitrary field of characteristic not equal to $2$.
A detailed discussion on this matter can be found in \cites{DragRadn2008,DragRadn2011book} (see also \cite{CCS1993}).

The section is concluded by Proposition \ref{prop:drc.quad}, where we show that double reflection configuration can take the role of the quad-equation, that is every line in such a configuration is determined by the remaining three.
This simple observation is going to play a key role in Section \ref{sec:cong}, in particular in Section \ref{sec:nets}.

Let us start with recalling the notions of quadrics and confocal families in the projective space.

\emph{A quadric} in $\mathbf{P}^d$ is the set given by equation of the form:
$$
(Q\xi,\xi)=0,
$$
where $Q$ is a symmetric $(d+1)\times(d+1)$ matrix, and $\xi=[\xi_0:\xi_1:\dots:\xi_d]$ are homogeneous coordinates of a point in the space.

Assume two quadrics are given:
$$
\mathcal{Q}_1\ :\ (Q_1\xi,\xi)=0,
\qquad
\mathcal{Q}_2\ :\ (Q_2\xi,\xi)=0.
$$
\emph{A pencil of quadrics} is the family of quadrics given by equations:
$$
\left((Q_1+\lambda Q_2)\xi,\xi\right)=0,
\quad
\lambda\in\mathbf{P}^1.
$$

\emph{A confocal system of quadrics} is a family of quadrics such that its projective dual is a pencil of quadrics.

Now, let us recall the definition of reflection off a quadric in the projective space, where metrics is not defined.
This definition, together with its crucial properties -- the One reflection theorem and the Double reflection theorem, can be found in \cite{CCS1993}.

Denote by $u$ the tangent plane to $\mathcal{Q}_1$ at point $x$ and by $z$ the pole
of $u$ with respect to $\mathcal{Q}_2$.
Suppose lines $\ell_1$ and $\ell_2$ intersect at $x$, and the plane containing these two lines meet $u$ along $\ell$, see Figure \ref{fig:projective-reflection}.
\begin{figure}[h]
\centering
\input{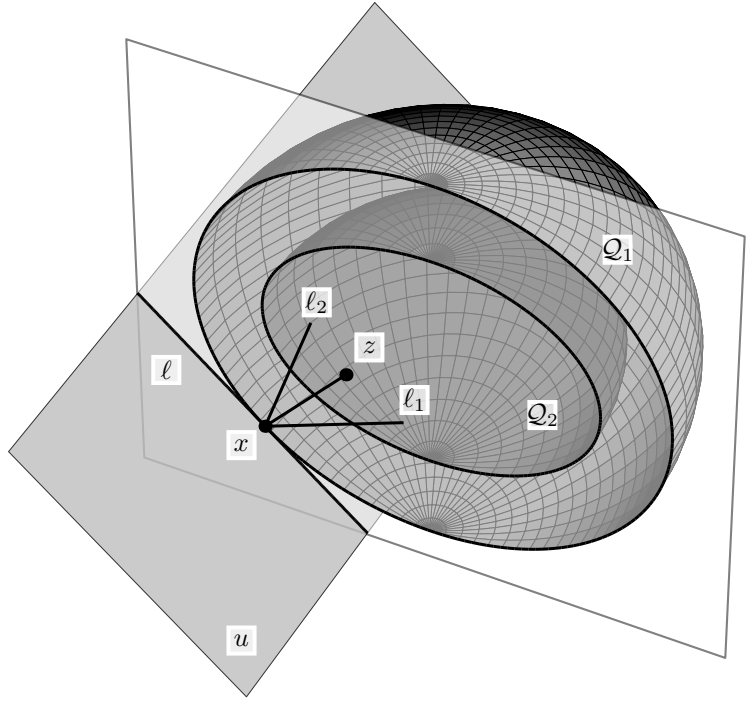}
\caption{The reflection in the projective space.}\label{fig:projective-reflection}
\end{figure}

\begin{definition}\label{def:refl}
If lines $\ell_1$, $\ell_2$, $xz$, $\ell$ are coplanar and harmonically conjugated, we say that $\ell_1$ \emph{is reflected to} $\ell_2$ off quadric $\mathcal{Q}_1$.
\end{definition}

It can be proved that this definition does not depend on the choice of quadric 
$\mathcal{Q}_2$ from a given confocal system \cite{CCS1993}.

If we introduce a coordinate system in which quadrics $\mathcal{Q}_1$ and
$\mathcal{Q}_2$ are confocal in the usual sense, the reflection introduced by Definition \ref{def:refl} is the same as the standard, metric one.

\begin{theorem}[One reflection theorem]\label{th:ORT}
Suppose line $\ell_1$ is reflected to $\ell_2$ off $\mathcal{Q}_1$ at point $x$,
with respect to the confocal system determined by quadrics
$\mathcal{Q}_1$ and $\mathcal{Q}_2$.
Let $\ell_1$ intersect $\mathcal{Q}_2$ at $y_1'$ and $y_1$, $u$ be the tangent plane to 
$\mathcal{Q}_1$ at $x$, and $z$ the pole of $u$ with respect to $\mathcal{Q}_2$.
Then lines $y_1'z$ and $y_1z$ respectively contain intersecting points $y_2'$ and $y_2$ of line $\ell_2$ with $\mathcal{Q}_2$.
The converse is also true.
(See Figure \ref{fig:one-reflection-theorem}).
\end{theorem}

\begin{figure}[h]
\centering
\input{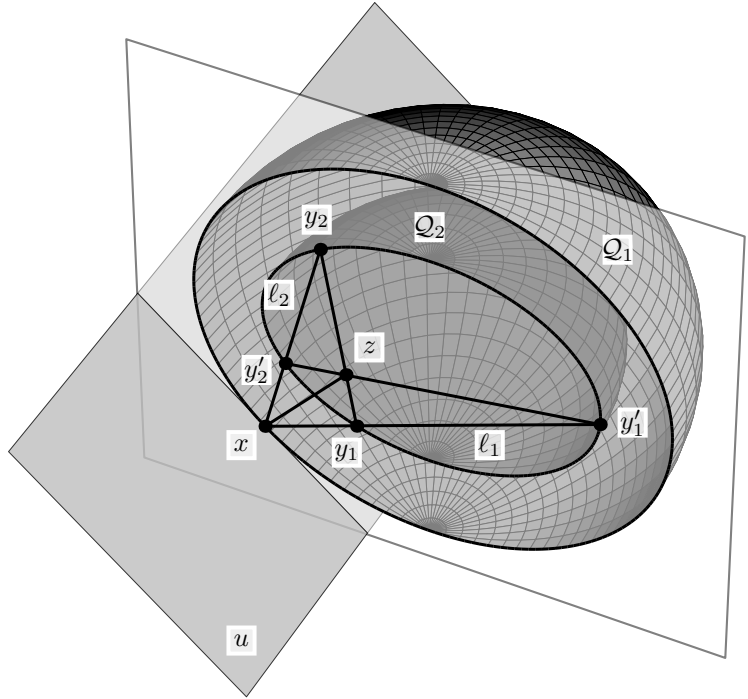}
\caption{One reflection theorem.}\label{fig:one-reflection-theorem}
\end{figure}

Theorem \ref{th:ORT} enables us to prove that the caustics are preserved by the reflection:

\begin{corollary}\label{cor:ORT}
Let lines $\ell_1$ and $\ell_2$ reflect to each other off $\mathcal{Q}_1$ with respect to the confocal system determined by quadrics $\mathcal{Q}_1$ and $\mathcal{Q}_2$.
Then $\ell_1$ is tangent to $\mathcal{Q}_2$ if and only if $\ell_2$ is tangent to 
$\mathcal{Q}_2$;
$\ell_1$ intersects $\mathcal{Q}_2$ at two points if and only if $\ell_2$ intersects
$\mathcal{Q}_2$ at two points.
\end{corollary}

Next theorem is crucial for our further considerations -- its meaning is that billiard reflections off confocal quadrics commute.

\begin{theorem}[Double reflection theorem]\label{th:DRT}
Suppose that $\mathcal{Q}_1$, $\mathcal{Q}_2$ are given quadrics and
$x_1\in\mathcal{Q}_1$, $y_1\in\mathcal{Q}_2$.
Let $u_1$ be the tangent plane of $\mathcal{Q}_1$ at $x_1$;
$z_1$ the pole of $u_1$ with respect to $\mathcal{Q}_2$;
$v_1$ the tangent plane of $\mathcal{Q}_2$ at $y_1$;
and $w_1$ the pole of $v_1$ with respect to $\mathcal{Q}_1$.
Denote by $x_2$ the intersecting point of line $w_1x_1$ with $\mathcal{Q}_1$,
$x_2\neq x_1$;
by $y_2$ the intersection of $y_1z_1$ with $\mathcal{Q}_2$, $y_2\neq y_1$;
and $\ell_1=x_1y_1$, $\ell_2=x_1y_2$, $\ell_1'=y_1x_2$, $\ell_2'=x_2y_2$.

Then pair $\ell_1$, $\ell_2$ obey the reflection law off $\mathcal{Q}_1$ at $x_1$;
$\ell_1$, $\ell_1'$ obey the reflection law off $\mathcal{Q}_2$ at $y_1$;
$\ell_2$, $\ell_2'$ obey the reflection law off $\mathcal{Q}_2$ at $y_2$;
and $\ell_1'$, $\ell_2'$ obey the reflection law off $\mathcal{Q}_1$ at point $x_2$.
(See Figure \ref{fig:double-reflection-theorem}).
\end{theorem}

\begin{figure}[h]
\centering
\input{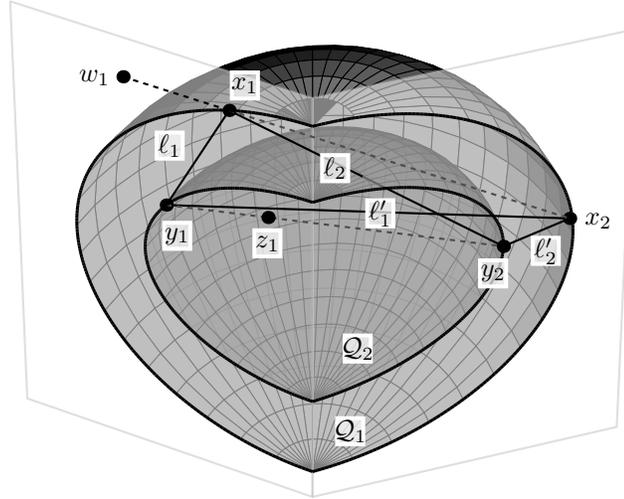}
\caption{Double reflection theorem.}\label{fig:double-reflection-theorem}
\end{figure}

Let us remark that in Theorem \ref{th:DRT} the four tangent planes at the reflection points belong to a pencil, see Figure \ref{fig:virtual_reflection}.


\begin{corollary}\label{cor:refl}
If line $\ell_1$ is tangent to a quadric $\mathcal{Q}'$ confocal with $\mathcal{Q}_1$ and $\mathcal{Q}_2$, then $\ell_2$, $\ell_1'$, $\ell_2'$ also touch $\mathcal{Q}'$.
\end{corollary}

The following definition of virtual reflection configuration and double reflection configuration is from \cite{DragRadn2008}, where these configurations played the central role.
In Theorem \ref{th:virt.refl}, which is also proved in \cite{DragRadn2008}, some important properties of these configurations are given.


Let points  $x_1$, $x_2$ belong to $\mathcal{Q}_1$ and  $y_1$, $y_2$ to $\mathcal{Q}_2$.

\begin{definition}\label{def:VRC}
We will say that the quadruple of points $x_1, x_2, y_1, y_2$ constitutes 
\emph{a virtual reflection configuration} if pairs of lines 
$x_1 y_1$, $x_1 y_2$; $x_2 y_1$, $x_2 y_2$; $x_1 y_1$, $x_2 y_1$; $x_1 y_2$, $x_2 y_2$ satisfy the reflection law at points
$x_1$, $x_2$ off $\mathcal Q_1$ and $y_1$, $y_2$ off $\mathcal Q_2$
respectively, with respect to the confocal system determined by
$\mathcal Q_1$ and $\mathcal Q_2$.

If, additionally, the tangent planes to $\mathcal Q_1, \mathcal Q_2$
at $x_1, x_2$; $y_1, y_2$ belong to a pencil, we say that these
points constitute \emph{a double reflection configuration} (see
Figure \ref{fig:virtual_reflection}).
\end{definition}

\begin{figure}[h]
\centering
\input{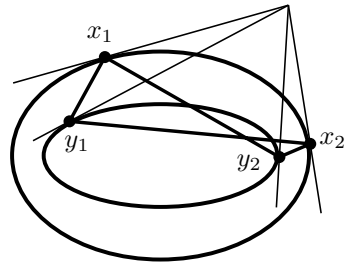}
\caption{Double reflection configuration}\label{fig:virtual_reflection}
\end{figure}

Now, we list some of the basic facts about double reflection configurations.

\begin{theorem}\label{th:virt.refl}
Let $\mathcal Q_1$, $\mathcal Q_2$ be two quadrics in the projective space
$\mathbf{P}^d$, $x_1$, $x_2$ points on $\mathcal Q_1$ and $y_1$, $y_2$ on 
$\mathcal{Q}_2$.
If the tangent hyperplanes at these points to the quadrics belong to a pencil, then $x_1, x_2, y_1, y_2$
constitute a virtual reflection configuration.

Furthermore, suppose that the projective space is defined over the field of reals.
Introduce a coordinate system, such that $\mathcal{Q}_1$, $\mathcal{Q}_2$ become confocal ellipsoids in the Euclidean space.
If $\mathcal{Q}_2$ is placed inside $\mathcal{Q}_1$, then the sides of the quadrilateral $x_1y_1x_2y_2$ obey the real reflection from $\mathcal{Q}_1$ and the virtual reflection from $\mathcal{Q}_2$.
\end{theorem}

The statement converse to Theorem \ref{th:virt.refl} is the following
\begin{proposition}\label{prop:drc-ellipsoids}
In the Euclidean space $\mathbf{E}^d$, two confocal ellipsoids $\mathcal{E}_1$ and $\mathcal{E}_2$ are given.
Let points $X_1$, $X_2$ belong to $\mathcal{E}_1$, $Y_1$, $Y_2$ to $\mathcal{E}_2$, and let $\alpha_1$, $\alpha_2$, $\beta_1$, $\beta_2$ be the corresponding tangent planes.
If a quadruple $X_1, X_2, Y_1, Y_2$ is a virtual reflection configuration, then planes $\alpha_1$, $\alpha_2$, $\beta_1$, $\beta_2$ belong to a pencil.
\end{proposition}

The next proposition shows that three lines of a double reflection configuration uniquely determine the fourth one.

\begin{proposition}\label{prop:drc.quad}
Let $\ell$, $\ell_1$, $\ell_2$ be lines and $\mathcal{Q}_1$, $\mathcal{Q}_2$ quadrics in the projective space.
Suppose that $\ell$, $\ell_1$ reflect to each other off $\mathcal{Q}_1$,
and $\ell$, $\ell_2$ off $\mathcal{Q}_2$, with respect to the confocal system determined by these two quadrics.
Then there is a unique line $\ell_{12}$ such that four lines 
$\ell$, $\ell_1$, $\ell_2$, $\ell_{12}$ form a double reflection configuration.
\end{proposition}

\begin{remark}\label{rem:drc.quad}
Proposition \ref{prop:drc.quad} shows that double reflection configuration is playing the role of the quad-equation for lines in the projective space, see Section \ref{sec:cong}.
\end{remark}


\subsection{Pseudo-Euclidean spaces and confocal families of quadrics}
\label{sec:pseudo.confocal}

In this section, we first give a necessary account of basic notions connected with pseudo-Euclidean spaces and their confocal families of quadrics.

\subsubsection*{Pseudo-Euclidean spaces}


\emph{The pseudo-Euclidean space} $\mathbf{E}^{k,l}$ is a
$d$-dimensional space $\mathbf{R}^d$ with \emph{the pseudo-Euclidean scalar product}:
\begin{equation}\label{eq:scalar.product}
\langle x,y
\rangle_{k,l}=x_1y_1+\dots+x_ky_k-x_{k+1}y_{k+1}-\dots-x_dy_d.
\end{equation}
Here, $k,l\in\{1,\dots, d-1\}$, $k+l=d$.
The pair $(k,l)$ is called \emph{the signature} of the space.
Denote $E_{k,l}=\diag(1,1,\dots,1,-1,\dots,-1)$, with $k$ $1$'s and $l$ $-1$'s.
Then the pseudo-Euclidean scalar product is:
$$
\langle x,y \rangle_{k,l} = E_{k,l}x\circ y,
$$
where $\circ$ is the standard Euclidean product.

\emph{The pseudo-Euclidean distance} between points $x$, $y$ is:
$$
\dist_{k,l}(x,y)=\sqrt{\langle{x-y,x-y}\rangle_{k,l}}.
$$
Since the scalar product can be negative, notice that the pseudo-Euclidean distance can have imaginary values as well.

Let $\ell$ be a line in the pseudo-Euclidean space, and $v$ its
vector. $\ell$ is called:
\begin{itemize}
\item
\emph{space-like} if $\langle{v,v}\rangle_{k,l}>0$;
\item
\emph{time-like} if $\langle{v,v}\rangle_{k,l}<0$;
\item
and \emph{light-like} if $\langle{v,v}\rangle_{k,l}=0$.
\end{itemize}
Two vectors $x$, $y$ are \emph{orthogonal} in the pseudo-Euclidean space if $\langle x,y \rangle_{k,l}=0$.
Note that a light-like line is orthogonal to itself.

For a given vector $v\neq0$, consider a hyper-plane $v\circ x=0$.
Vector $E_{k,l}v$ is orthogonal to the hyper-plane; moreover, all other orthogonal vectors are collinear with $E_{k,l}v$.
If $v$ is light-like, then so is $E_{k,l}v$, and $E_{k,l}v$ belongs to the hyper-plane.

\subsubsection*{Billiard reflection in the pseudo-Euclidean space}

Let $v$ be a vector and $\alpha$ a hyper-plane in the pseudo-Euclidean space.
Decompose vector $v$ into the sum $v=a+n_{\alpha}$ of a vector $n_{\alpha}$ orthogonal to $\alpha$ and $a$ belonging to $\alpha$.
Then vector $v'=a-n_{\alpha}$ is \emph{the billiard reflection} of $v$ on $\alpha$.
It is easy to see that then $v$ is also the billiard reflection of $v'$ with respect to $\alpha$.

Moreover, let us note that lines containing vectores $v$, $v'$, $a$, $n_{\alpha}$ are harmonically conjugated \cite{KhTab2009}.

Note that $v=v'$ if $v$ is contained in $\alpha$ and $v'=-v$ if it is orthogonal to $\alpha$.
If $n_{\alpha}$ is light-like, which means that it belongs to $\alpha$, then the reflection is not defined.

Line $\ell'$ is the billiard reflection of $\ell$ off a smooth surface $\mathcal{S}$ if their intersection point $\ell\cap\ell'$ belongs to $\mathcal{S}$ and the vectors of $\ell$, $\ell'$ are reflections of each other with respect to the tangent plane of $\mathcal{S}$ at this point.

\begin{remark}\label{remark:type}
It can be seen directly from the definition of reflection that the type of line is preserved by the billiard reflection.
Thus, the lines containing segments of a given billiard trajectory within $\mathcal{S}$ are all of the same type: they are all either space-like, time-like, or light-like.
\end{remark}

If $\mathcal{S}$ is an ellipsoid,
then it is possible to extend the reflection mapping to those points
where the tangent planes contain the orthogonal vectors. At such
points, a vector reflects into the opposite one, i.e.~$v'=-v$ and
$\ell'=\ell$. For the explanation, see \cite{KhTab2009}.
As follows from the explanation given there, it is natural to consider each such reflection as two reflections.

\subsection*{Families of confocal quadrics}
For a given set of positive constants $a_1$, $a_2$, \dots, $a_d$, an
ellipsoid is given by:
\begin{equation}\label{eq:ellipsoid}
\mathcal{E}\ :\
\frac{x_1^2}{a_1}+\frac{x_2^2}{a_2}+\dots+\frac{x_d^2}{a_d}=1.
\end{equation}
Let us remark that equation of any ellipsoid in the pseudo-Euclidean
space can be brought into the canonical form (\ref{eq:ellipsoid})
using transformations that preserve the scalar product
(\ref{eq:scalar.product}).

The family of quadrics confocal with $\mathcal{E}$ is:
\begin{equation}\label{eq:confocal-pseudo}
 \mathcal{Q}_{\lambda}\ :\
\frac{x_1^2}{a_1-\lambda} +\dots+ \frac{x_k^2}{a_k-\lambda} +
\frac{x_{k+1}^2}{a_{k+1}+\lambda} +\dots+
\frac{x_d^2}{a_d+\lambda}=1,\qquad\lambda\in\mathbf{R}.
\end{equation}

Unless stated differently, we are going to consider the
non-degenerate case, when set
$\{a_1,\dots,a_k,-a_{k+1},\dots,-a_{d}\}$ consists of $d$ different
values:
$$
a_1>a_2>\dots>a_k>0>-a_{k+1}>\dots>-a_d.
$$

For $\lambda\in\{a_1,\dots,a_k,-a_{k+1},\dots,-a_{d}\}$, the quadric
$\mathcal Q_{\lambda}$ is degenerate and it coincides with the
corresponding coordinate hyper-plane.

It is natural to join one more degenerate quadric to the family
(\ref{eq:confocal-pseudo}): the one corresponding to the value
$\lambda=\infty$, that is the hyper-plane at the infinity.

For each point $x$ in the space, there are exactly $d$ values of
$\lambda$, such that the relation (\ref{eq:confocal-pseudo}) is satisfied.
However, not all the values are necessarily real: either all $d$ of
them are real or there are $d-2$ real and $2$ conjugate complex
values. Thus, through every point in the space, there are either $d$
or $d-2$ quadrics from the family (\ref{eq:confocal-pseudo})
\cite{KhTab2009}.

The line $x+tv$ ($t\in\mathbf{R}$) is tangent to quadric
$\mathcal{Q}_{\lambda}$ if quadratic equation:
\begin{equation}\label{eq:tangent}
A_{\lambda}(x+tv)\circ(x+tv)=1,
\end{equation}
has a double root.
Here we denoted:
$$
A_{\lambda}=\diag\left( \frac{1}{a_1-\lambda}, \cdots,
\frac{1}{a_k-\lambda}, \frac{1}{a_{k+1}+\lambda}, \cdots,
\dfrac{1}{a_d+\lambda} \right).
$$

Now, calculating the discriminant of (\ref{eq:tangent}), we get:
\begin{equation}\label{eq:diskriminanta}
(A_{\lambda}x\circ v)^2 -
(A_{\lambda}v{\circ}v)(A_{\lambda}x{\circ}x-1) =0,
\end{equation}
which is equivalent to:
\begin{equation}\label{eq:discr}
\sum_{i=1}^d \frac{\varepsilon_i F_i(x,v)}{a_i-\varepsilon_i\lambda}=0,
\end{equation}
where
\begin{equation}\label{eq:integralsF}
F_i(x,v)=\varepsilon_iv_i^2 + \sum_{j{\neq}i}
\frac{(x_iv_j-x_jv_i)^2}{\varepsilon_ja_i-\varepsilon_ia_j},
\end{equation}
with $\varepsilon$'s given by:
\begin{equation*}
\varepsilon_i=\begin{cases}
1, & 1\le i\le k;\\
-1, & k+1\le i\le d.
\end{cases}
\end{equation*}

The equation (\ref{eq:discr}) can be transformed to:
\begin{equation}\label{eq:polinom.P}
\frac{\mathcal{P}(\lambda)}{\prod_{i=1}^d(a_i-\varepsilon_i\lambda)}=0,
\end{equation}
where the coefficient of $\lambda^{d-1}$ in $\mathcal{P}(\lambda)$
 is equal to $\langle{v,v}\rangle_{k,l}$.
Thus, polynomial $\mathcal{P}(\lambda)$ is of degree $d-1$ for space-like and time-like
lines, and of a smaller degree for light-like lines.
However, in the latter case, it turns out to be natural to consider the
polynomial $\mathcal{P}(\lambda)$ also as of degree $d-1$, taking
the corresponding roots to be equal to infinity.
So, light-like
lines are characterized by being tangent to the quadric
$\mathcal{Q}_{\infty}$.

Having this setting in mind, we note that  it is proved in \cite{KhTab2009} that the polynomial
$\mathcal{P}(\lambda)$ has at least $d-3$ roots in $\mathbf{R}\cup\{\infty\}$.

Thus, we have:

\begin{proposition}\label{prop:kaustike}
Any line in the space is tangent to either $d-1$ or $d-3$ quadrics
of the family (\ref{eq:confocal-pseudo}).
If this number is equal to $d-3$,
then there are two conjugate complex values of $\lambda$, such that
the line is tangent also to these two quadrics in $\mathbf{C}^d$.
\end{proposition}

This statement with the proof is given in \cite{KhTab2009}.
Let us remark that in \cite{KhTab2009} is claimed that light-like line have
only $d-2$ or $d-4$ caustic quadrics.
That is because $\mathcal{Q}_{\infty}$ is not considered there as a member of the confocal family.

As noted in \cite{KhTab2009}, a line having non-empty intersection with an ellipsoid from (\ref{eq:confocal-pseudo}) will be tangent to $d-1$ quadrics from the confocal family.
However, the next theorem, proved in \cite{DragRadn2012adv}, will provide some more insight into the distribution of the parameters of the caustics of a given line, together with a detailed description of the distribution of the parameters of quadrics containing a given point placed inside an ellipsoid from (\ref{eq:confocal-pseudo}).

\begin{theorem}\label{th:parametri.kaustike}
In pseudo-Euclidean space $\mathbf{E}^{k,l}$ consider a line intersecting ellipsoid $\mathcal{E}$ (\ref{eq:ellipsoid}).
Then this line is touching $d-1$ quadrics from (\ref{eq:confocal-pseudo}).
If we denote their parameters by $\alpha_1$, \dots, $\alpha_{d-1}$ and take:
\begin{gather*}
\{b_1,\ \dots,\ b_p,\ c_1,\ \dots,\ c_q\}=\{\varepsilon_{1}a_1,\ \dots,\ \varepsilon_{d}a_d,\ \alpha_1,\ \dots,\ \alpha_{d-1}\},\\
c_q\le\dots\le c_2\le c_1<0<b_1\le b_2\le\dots\le b_p,\quad p+q=2d-1,
\end{gather*}
we will additionally have:
\begin{itemize}
 \item
if the line is space-like, then $p=2k-1$, $q=2l$, $a_1=b_p$, $\alpha_i\in\{b_{2i-1},b_{2i}\}$ for $1\le i\le k-1$, and $\alpha_{j+k-1}\in\{c_{2j-1},c_{2j}\}$ for $1\le j\le l$;
 \item
if the line is time-like, then $p=2k$, $q=2l-1$, $c_q=-a_d$, $\alpha_i\in\{b_{2i-1},b_{2i}\}$ for $1\le i\le k$, and $\alpha_{j+k}\in\{c_{2j-1},c_{2j}\}$ for $1\le j\le l-1$;
 \item
if the line is light-like, then $p=2k$, $q=2l-1$, $b_p=\infty=\alpha_k$, $b_{p-1}=a_1$, $\alpha_i\in\{b_{2i-1},b_{2i}\}$ for $1\le i\le k-1$, and $\alpha_{j+k}\in\{c_{2j-1},c_{2j}\}$ for $1\le j\le l-1$.
\end{itemize}
Moreover, for each point on $\ell$ inside $\mathcal{E}$, there is exactly $d$ distinct quadrics from (\ref{eq:confocal-pseudo}) containing it.
More precisely, there is exactly one parameter of these quadrics in each of the intervals:
$$
(c_{2l-1},c_{2l-2}),\ \dots,\ (c_3,c_2),\ (c_1,0),\ (0,b_1),\ (b_2,b_3),\ \dots,\ (b_{2k-2},b_{2k-1}).
$$
\end{theorem}

The analogue of Theorem \ref{th:parametri.kaustike} for the Euclidean space is proved in \cite{Audin1994}.

\begin{corollary}
For each point placed inside an ellipsoid in the pseudo-Euclidean space, there are exactly two other ellipsoids from the confocal family containing this point.
\end{corollary}

\section{Pseudo-Euclidean spaces and Poncelet theorem}
\label{sec:pseudo.poncelet}
\subsection{Minkowski plane, confocal conics and billiards}
\label{sec:minkowski.plane}
\subsubsection*{Confocal conics in the Minkowski plane}

Here, we give a review of basic properties of families of confocal conics in the Minkowski plane, see \cite{DragRadn2012adv}.

Denote by
\begin{equation}\label{eq:ellipse}
\mathcal{E}\ :\ \frac{x^2}{a}+\frac{y^2}{b}=1
\end{equation}
an ellipse in the plane, with $a$, $b$ being fixed positive numbers.

The associated family of confocal conics is:
\begin{equation}\label{eq:confocal.conics} 
\mathcal C_{\lambda}\ :\
\frac{x^2}{a-\lambda}+\frac{y^2}{b+\lambda}=1, \quad
\lambda\in\mathbf{R}.
\end{equation}

The family is shown on Figure \ref{fig:confocal.conics}.
\begin{figure}[h]
\input{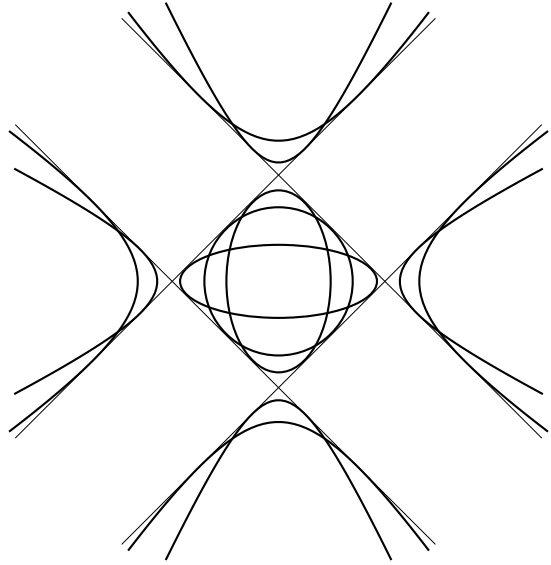}
\caption{Family of confocal conics in the Minkowski
plane.}\label{fig:confocal.conics}
\end{figure}
We may distinguish the following three subfamilies in the family
(\ref{eq:confocal.conics}):
\begin{itemize}
\item
for $\lambda\in(-b,a)$, conic $\mathcal{C}_{\lambda}$ is an ellipse;

\item
for $\lambda<-b$, conic $\mathcal{C}_{\lambda}$ is a hyperbola with $x$-axis as the major one;

\item
for $\lambda>a$, it is a hyperbola again, but now its major axis is $y$-axis.
\end{itemize}
In addition, there are three degenerated quadrics: $\mathcal{C}_{a}$, $\mathcal{C}_{b}$, $\mathcal{C}_{\infty}$ corresponding to $y$-axis, $x$-axis, and the line at the infinity respectively.
Note the following three pairs of foci:
$F_1(\sqrt{a+b},0)$, $F_2(-\sqrt{a+b},0)$; 
$G_1(0,\sqrt{a+b})$, $G_2(0,-\sqrt{a+b})$; and 
$H_1(1:-1:0)$, $H_2(1:1:0)$ on the line at the infinity.

We notice four distinguished lines: 
\begin{align*}
&x+y=\sqrt{a+b},\quad x+y=-\sqrt{a+b},\\
&x-y=\sqrt{a+b},\quad x-y=-\sqrt{a+b}.
\end{align*}
These lines are common tangents to all conics from the confocal family.

Some geometric properties of conics in the Minkowski plane are analogous the Euclidean ones.
For example, for each point on conic $\mathcal{C}_{\lambda}$, either sum or difference of its Minkowski distances from the foci $F_1$ and $F_2$ is equal to $2\sqrt{a-\lambda}$;
either sum or difference of the distances from the other pair of foci $G_1$, $G_2$ is equal to $2\sqrt{-b-\lambda}$ \cite{DragRadn2012adv}.

We invite the reader to make further comparisons of confocal families of conics in the Minkowski and Euclidean planes (see Figures \ref{fig:confocal-e} and \ref{fig:confocal.conics}).

\subsubsection*{Relativistic conics}

Since a family of confocal conics in the Minkowski plane contains three geometric types of conics, it is natural to introduce relativistic conics, which are suggested in \cite{BirkM1962}.
In this section, we give a brief account of related analysis.

Consider points $F_1(\sqrt{a+b},0)$ and $F_2(-\sqrt{a+b},0)$ in the plane.

For a given constant $c\in\mathbf{R}^{+}\cup i\mathbf{R}^{+}$, \emph{a relativistic ellipse} is the set of points $X$ satisfying:
$$
\dist_{1,1}(F_1,X)+\dist_{1,1}(F_2,X)=2c,
$$
while \emph{a relativistic hyperbola} is the union of the sets given by the following equations:
\begin{gather*}
\dist_{1,1}(F_1,X)-\dist_{1,1}(F_2,X)=2c,\\ 
\dist_{1,1}(F_2,X)-\dist_{1,1}(F_1,X)=2c.
\end{gather*}

Relativistic conics can be described as follows.
\begin{description}
 \item[$0<c<\sqrt{a+b}$]
The corresponding relativistic conics lie on ellipse $\mathcal{C}_{a-c^2}$ from family (\ref{eq:confocal.conics}).
The ellipse $\mathcal{C}_{a-c^2}$ is split into four arcs by touching points with the four common tangent lines; thus, the relativistic ellipse is the union of the two arcs intersecting the $y$-axis, while the relativistic hyperbola is the union of the other two arcs.

\item[$c>\sqrt{a+b}$]
The relativistic conics lie on $\mathcal{C}_{a-c^2}$ -- a hyperbola with $x$-axis as the major one.
Each branch of the hyperbola is split into three arcs by touching points with the common tangents; thus, the relativistic ellipse is the union of the two finite arcs, while the relativistic hyperbola is the union of the four infinite ones.

\item[$c$ is imaginary]
The relativistic conics lie on hyperbola $\mathcal{C}_{a-c^2}$ -- a hyperbola with $y$-axis as the major one.
As in the previous case, the branches are split into six arcs in total by common points with the four tangents.
The relativistic ellipse is the union of the four infinite arcs, while the relativistic hyperbola is the union of the two finite ones. 
\end{description}

The conics are shown on Figure \ref{fig:relativistic.conics}.
\begin{figure}[h]
\input{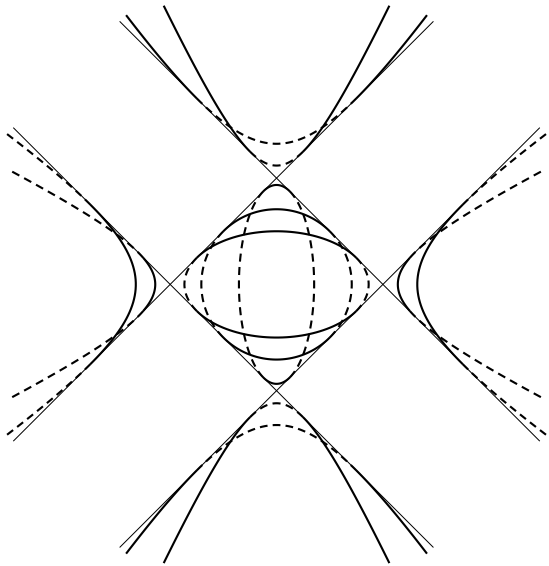}
\caption{Relativistic conics in the Minkowski plane: relativistic ellipses are represented by full lines, and hyperbolas by dashed ones.}\label{fig:relativistic.conics}
\end{figure}

Notice that all relativistic ellipses are disjoint with each other, as well as all relativistic hyperbolas.
Moreover, at the intersection point of a relativistic ellipse which is a part of the geometric conic $\mathcal{C}_{\lambda_1}$ from the confocal family (\ref{eq:confocal.conics}) and a relativistic hyperbola belonging to $\mathcal{C}_{\lambda_2}$, it is always $\lambda_1<\lambda_2$.

\subsubsection*{Periodic trajectories of elliptical billiard}

Analytic conditions for existence of closed polygonal lines inscribed in one conic and circumscribed about another one in the projective plane are derived by Cayley \cites{Cayley1854,Cayley1861}.
They can be applied to billiard trajectories within ellipses in the Minkowski plane as well, since each such trajectory has a caustic among confocal conics.
In this section, we shall analyze in more detail some particular properties related to the Minkowski geometry.

\begin{theorem}\label{th:cayley2}
In the Minkowski plane, consider a billiard trajectory $\mathcal{T}$ within ellipse $\mathcal{E}$ given by equation (\ref{eq:ellipse}).

The trajectory is periodic with period $n=2m$ if and only if the following condition is satisfied:
\begin{equation}\label{eq:cayley2even}
\det\left(
\begin{array}{llll}
B_3 & B_4 & \dots & B_{m+1}\\
B_4 & B_5 & \dots & B_{m+2}\\
\dots & \dots & \dots & \dots\\
B_{m+1} & B_{m+2} &\dots & B_{2m-1}
\end{array}
\right)=0.
\end{equation}
Trajectory $\mathcal{T}$ is periodic with period $n=2m+1$ if and only if $\mathcal{C}_{\alpha}$ is an ellipse and the following condition is satisfied:
\begin{equation}\label{eq:cayley2odd}
\det\left(
\begin{array}{llll}
B_3 & B_4 & \dots & B_{m+2}\\
\dots & \dots & \dots & \dots\\
B_{m+1} & B_{m+2} &\dots & B_{2m}\\
C_{m+1} & C_{m+2} &\dots & C_{2m}
\end{array}
\right)=0.
\end{equation}
Here:
\begin{gather*}
\sqrt{(a-t)(b+t)(\alpha-t)}=B_0+B_1t+B_2t^2+\dots,\\
\sqrt{\frac{(a-t)(b+t)}{\alpha-t}}=C_0+C_1t+C_2t^2+\dots
\end{gather*}
are the Taylor expansions around $\lambda=0$.
\end{theorem}

\begin{proof}
Each point inside $\mathcal{E}$ is the intersection of exactly two ellipses $\mathcal{C}_{\lambda_1}$ and $\mathcal{C}_{\lambda_2}$ from (\ref{eq:confocal.conics}).
Parameters $\lambda_1$, $\lambda_2$ are generalized Jacobi coordinates.
Take $\lambda_1<\lambda_2$.

Consider first the case when $\mathcal{C}_{\alpha}$ is a hyperbola.
Then along $\mathcal{T}$ these coordinates will take values in segments $[-b,0]$ and $[0,a]$ respectively with the endpoints of the segments as the only local extrema.
$\lambda_1$ achieves value $-b$ at the intersections of $\mathcal{T}$ with the $x$-axis, while $\lambda_2$ achieves $a$ at the intersections with $y$-axis.
At each reflection point, one of the coordinates achieves value $0$.
They can both be equal to $0$ only at the points where $\mathcal{E}$ has a light-like tangent, and there reflection is counted twice.

This means that on a closed trajectory the number of reflections is equal to the number of intersection points with the coordinate axes.
Since a periodic trajectory crosses each of the coordinate axes even number of times, the first part of the theorem is proved.

The condition on $\mathcal{T}$ to become closed after $n$ reflections on $\mathcal{E}$, $n_1$ crossings over $x$-axis, and  $n_2$ over $y$-axis is that the following equality 
$$
n_1P_{a}+n_2P_{-b}=nP_{0}
$$ 
holds on the elliptic curve:
$$
s^2=(a-t)(b+t)(\alpha-t),
$$ 
where by $P_{\beta}$ we denoted a point on the curve correspondig to $t=\beta$, and $P_{\infty}$ is taken to be the neutral for the elliptic curve group.

From the previous discussion, $n_1+n_2=n$ and all three numbers are even.
$P_{a}$ and $P_{-b}$ are branching points of the curve, thus $2P_{a}=2P_{-b}=2P_{\infty}$, so the condition becomes $nP_{0}=nP_{\infty}$, which is equivalent to (\ref{eq:cayley2even}).

Now suppose $\mathcal{C}_{\alpha}$ is an ellipse.
The generalized Jacobi coordinates take values in segments $[-b,0]$, $[0,\alpha]$ or in $[\alpha,0]$, $[0,a]$, depending on the sign of $\alpha$.
Since both cases are proceeded in a similar way, we assume $\alpha<0$.

Coordinate $\lambda_1$ has extrema on $\mathcal{T}$ at the touching points with the caustic and some of the reflection points, while $\lambda_2$ has extrema at the crossing points with $y$-axis and some of the reflection points.

The condition on $\mathcal{T}$ to become closed after $n$ reflections on $\mathcal{E}$, with $n_1$ crossings over $y$-axis, and $n_2$ touching points with the caustic is: 
$$
n_1P_{a}+n_2P_{\alpha}=nP_{0},
$$
with $n_1+n_2=n$ and $n_1$ even.

Thus, for $n$ even we get (\ref{eq:cayley2even}) in the same manner as for a hyperbola as a caustic.

For $n$ odd, the condition is equivalent to $nP_{0}=(n-1)P_{\infty}+P_{\alpha}$.
Notice that one basis of the space $\mathcal{L}((n-1)P_{\infty}+P_{\alpha})$ is:
$$
1,t,\dots,t^m,s,ts,\dots, t^{m-2}s, \dfrac{s}{t-\alpha}.
$$
Using this basis, as it is shown in \cites{DragRadn2006jms,DragRadn2011book}, we obtain (\ref{eq:cayley2odd}).
\end{proof}

\begin{example}[$3$-periodic trajectories]
Let us find all $3$-periodic trajectories within ellipse $\mathcal{E}$ given by
(\ref{eq:ellipse}) in the Minkowski plane, i.e.~ all conics $\mathcal{C}_{\alpha}$ from the confocal family (\ref{eq:confocal.conics}) corresponding to such trajectories.

The condition is:
$$
C_2=\frac{3a^2 b^2 + 2 a^2 b\alpha - 2 a b^2 \alpha - a^2 \alpha^2 - 
 2 a b \alpha^2 - b^2 \alpha^2}{8 (ab)^{3/2} \alpha^{5/2}}=0,
$$
which gives the following solutions for the parameter $\alpha$ of the caustic:
\begin{gather*}
\alpha_1=\frac{ab}{(a+b)^2}(a-b - 2\sqrt{a^2 + ab + b^2}),
\\
\alpha_2=\frac{ab}{(a+b)^2}(a-b + 2\sqrt{a^2 + ab + b^2}).
\end{gather*}
Notice that
$-b<\alpha_1<0<\alpha_2<a$
so both caustics $\mathcal{C}_{\alpha_1}$, $\mathcal{C}_{\alpha_2}$ are ellipses.
\end{example}

\begin{example}[$4$-periodic trajectories]
By Theorem \ref{th:cayley2}, the condition is $B_3=0$.
Since
$$
B_3=\frac{(-ab-a\alpha+b\alpha)(-ab+a\alpha+b\alpha)(ab+a\alpha+b\alpha)}
{16(ab\alpha)^{5/2}},
$$
we obtain the following solutions:
$$
\alpha_1=\frac{ab}{b-a},\quad
\alpha_2=\frac{ab}{a+b},\quad
\alpha_3=-\frac{ab}{a+b}.
$$
Since $\alpha_1\not\in(-b,a)$ and $\alpha_2,\alpha_3\in(-b,a)$, conic
$\mathcal{C}_{\alpha_1}$ is a hyperbola, while $\mathcal{C}_{\alpha_2}$,
$\mathcal{C}_{\alpha_3}$ are ellipses.
\end{example}

\begin{example}[$5$-periodic trajectories]
The condition is:
$$
\det\left(
\begin{array}{ll}
 B_3 & B_4\\
 C_3 & C_4
\end{array}
\right)=0.
$$
Taking $a=b=1$, we get that this is equivalent to $64\alpha^6-16\alpha^4-52\alpha^2+5=0$.
This equation has four solutions in $\mathbf{R}$, all four contained in $(-1,1)$, and two conjugated solutions in $\mathbf{C}$.
\end{example}

\subsubsection*{Light-like trajectories of elliptical billiard}
\label{sec:light-like}

In this section we consider in more detail light-like trajectories of elliptical billiard, see Figure \ref{fig:traj} for an example of such a trajectory.
\begin{figure}[h]
\input{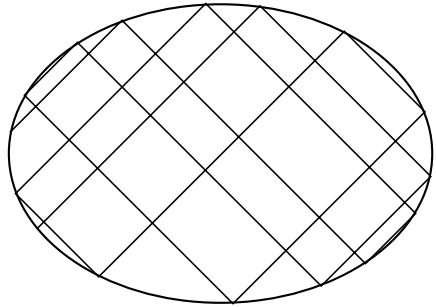}
\caption{Light-like billiard trajectory.}\label{fig:traj}
\end{figure}
We are going to review results from \cite{DragRadn2012adv} and illustrate them by some examples.

\subsubsection*{Periodic light-like trajectories}

Let us first notice that successive segments of light-like billiard trajectories are orthogonal to each other (see Figure \ref{fig:traj}), thus a trajectory can close only after even number of reflections.

The analytic condition for $n$-periodicity of light-like billiard trajectory within the ellipse $\mathcal{E}$ given by equation (\ref{eq:ellipse}) can be derived as in Theorem \ref{th:cayley2}.
We get the condition stated in (\ref{eq:cayley2even}), with $\alpha=\infty$, i.e.~ $(B_i)$ are coefficients in the Taylor expansion around $t=0$ of
$\sqrt{(a-t)(b+t)} = B_0 + B_1t + B_2t^2 + \dots$.

Now, we are going to derive analytic condition for periodic light-like trajectories in another way, which will lead to a more compact form of (\ref{eq:cayley2even}).

By applying affine transformations, one can transform an ellipse into a circle, and the billiard map on the light-like lines becomes conjugated to a rotation of the circle.
Computing the angle of the rotation gives the following: 

\begin{theorem}\label{th:arctan}
Light-like billiard trajectory within ellipse $\mathcal{E}$ is
periodic with period $n$, where $n$ is an even integer if and only
if
\begin{equation}\label{eq:cayley_arctan}
\arctan\sqrt{\frac ab}\in\left\{\ \frac{k\pi}n\ \left|\ 1\le k<\frac
n2,\ \left(k,\frac{n}{2}\right)=1\right.\ \right\}.
\end{equation}
\end{theorem}

As an immediate consequence, we get
\begin{corollary}\label{cor:cayley.euler}
For a given even integer $n$, the number of different ratios of the
axes of ellipses having $n$-periodic light-like billiard
trajectories is equal to:
$$
\begin{cases}
\varphi(n)/2 & \text{if}\ \ n\ \text{is not divisible by}\ 4,
 \\
\varphi(n)/4 & \text{if}\ \ n\ \text{is divisible by}\ 4.
\end{cases}
$$
$\varphi$ is the Euler's totient function, i.e.\ the number of
positive integers not exceeding $n$ that are relatively prime to
$n$.
\end{corollary}

\begin{remark}\label{rem:nk}
There are $4$ points on $\mathcal{E}$ where the tangents are
light-like. Those points cut $4$ arcs on $\mathcal{E}$. An
$n$-periodic trajectory within $\mathcal{E}$ hits each one of a pair
of opposite arcs exactly $k$ times, and $\dfrac{n}2-k$ times the
arcs from the other pair.
\end{remark}

\begin{example}[$10$-perodic light-like trajectories]
For $n=10$, the condition (\ref{eq:cayley2even}) is:
$$
\det\left(
\begin{array}{llll}
 B_3 & B_4 & B_5 & B_6\\
 B_4 & B_5 & B_6 & B_7\\
 B_5 & B_6 & B_7 & B_8\\
 B_6 & B_7 & B_8 & B_9
\end{array}\right)
 =
\frac{(a + b)^{20} (5 a^2 - 10 a b + b^2) (a^2 - 10 a b +
   5 b^2)}{(4ab)^{22}}=0.
$$
From here, we get that light-like billiard trajectories are
$10$-periodic in ellipses with the ratio of the axes equal to either
$\sqrt{1+\dfrac{2}{\sqrt{5}}}$ or $\sqrt{5+2\sqrt{5}}$.

From condition (\ref{eq:cayley_arctan}), we get:
$$
\sqrt{\frac{a}{b}}\in
 \left\{
\tan\frac{\pi}{10}, \tan\frac{2\pi}{10}, \tan\frac{3\pi}{10},
\tan\frac{4\pi}{10}
 \right\}.
$$
Since
\begin{align*}
\tan\frac{\pi}{10}=\sqrt{1-\frac2{\sqrt5}}&=\frac1{\sqrt{5+2\sqrt5}}=\frac1{\tan\frac{4\pi}{10}}
 \\
\tan\frac{3\pi}{10}=\sqrt{1+\frac2{\sqrt5}}&=\frac1{\sqrt{5-2\sqrt5}}=\frac1{\tan\frac{2\pi}{10}},
\end{align*}
both conditions give the same result.

A few of such trajectories are shown on Figures \ref{fig:period10}
and \ref{fig:period10b}.
\end{example}
\begin{figure}[h]
\input{0-figures/period10}
\caption{Light-like billiard trajectories with period 10 in the
ellipse satisfying $a=(1+2/\sqrt{5})b$.} \label{fig:period10}
\end{figure}

\begin{figure}[h]
\input{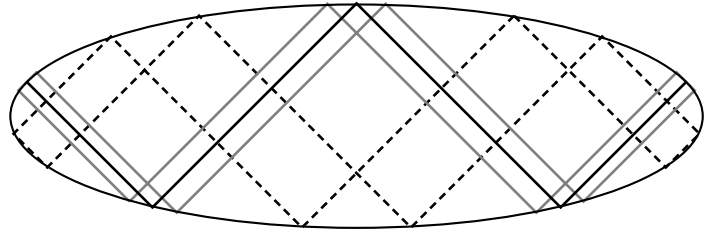}
\caption{Light-like billiard trajectories with period 10 in the
ellipse satisfying $a=(5+2\sqrt{5})b$.} \label{fig:period10b}
\end{figure}

\subsubsection*{Light-like trajectories in ellipses and rectangular billiards}

\begin{theorem}\label{th:rectangle}
The flow of light-like billiard trajectories within ellipse
$\mathcal{E}$ is trajectorially equivalent to the flow of those
billiard trajectories within a rectangle whose angle with the sides
is $\dfrac{\pi}4$. The ratio of the sides of the rectangle is equal
to:
$$
\frac{\pi}{2\arctan\sqrt{\dfrac{a}{b}}}-1.
$$
\end{theorem}

\begin{remark}
The flow of light-light billiard trajectories within a given oval in the Minkowski plane will be trajectorially equivalent to the flow of certain trajectories within a rectangle whenever invariant measure $m$ on the oval exists such that $m(AB)=m(CD)$ and $m(BC)=m(AD)$, where $A$, $B$, $C$, $D$ are points on the oval where the tangents are light-like.
\end{remark}

\subsection{Relativistic quadrics}
\label{sec:relativistic}
The aim of this Section is to present relativistic quadrics, a very recent object, which has been a main tool in our study of billiard dynamics, see \cite{DragRadn2012adv}.
First, we study  geometrical types of quadrics in a confocal family in the three-dimensional Minkowski space; then we analyze tropic curves on quadrics in the three-dimensional case and we introduce an important notion of discriminant sets $\Sigma^{\pm}$ corresponding to a confocal family.
The main facts about discriminant sets we present in Propositions \ref{prop:tropic.surface}, \ref{prop:tropic.light}, \ref{prop:tetra.hyp}, \ref{prop:hyp.y.presek}.
After, we study curved tetrahedra $\mathcal{T}^{\pm}$, which represent singularity sets of $\Sigma^{\pm}$ and we collect related results in Proposition \ref{prop:tetra}.
As the next important step, we introduce decorated Jacobi coordinates for the three-dimensional Minkowski space, and we give a detailed description of the colouring into three colours. 
Each colour corresponds to a relativistic type, and we describe decomposition of a geometric quadric of each of the four  geometric types into relativistic quadrics. 
This appears to be a rather involved combinatorial-geometric problem, and we solve it by using previous analysis of discriminant surfaces. 
We give a complete description of all three relativistic types of quadrics.
At the end, we generalize definition of decorated Jacobi coordinates in arbitrary dimensions, and, finally,
in Proposition \ref{prop:svojstva12} we prove properties PE1 and PE2.

\subsubsection*{Confocal quadrics in the three-dimensional Minkowski space and their geometrical types}
Let us start with the three-dimensional Minkowski space $\mathbf{E}^{2,1}$.
A general confocal family of quadrics in this space is given by:
\begin{equation}\label{eq:confocal.quadrics3}
\mathcal{Q}_{\lambda}\ :\ \frac{x^2}{a-\lambda}+\frac{y^2}{b-\lambda}+\frac{z^2}{c+\lambda}=1,
\quad
\lambda\in\mathbf{R},
\end{equation}
with $a>b>0$, $c>0$, see Figure \ref{fig:konfM3}.
\begin{figure}[h]
\includegraphics[width=5.87cm, height=8.38cm]{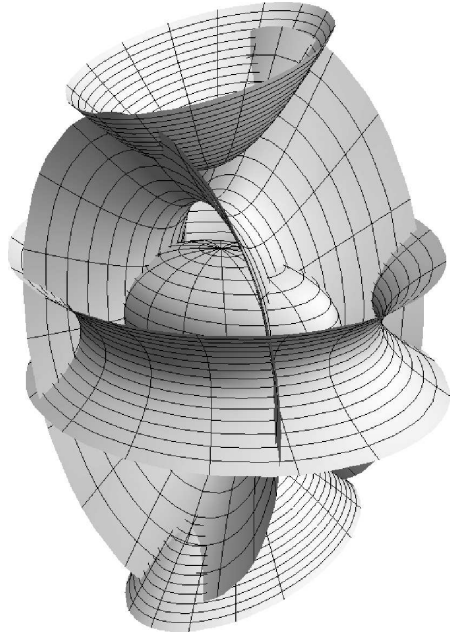}
\caption{Confocal quadrics in the three-dimensional Minkowski space.}\label{fig:konfM3}
\end{figure}

The family (\ref{eq:confocal.quadrics3}) contains four geometrical types of quadrics:
\begin{itemize}
\item
$1$-sheeted hyperboloids oriented along $z$-axis, for $\lambda\in(-\infty,-c)$;

\item
ellipsoids, corresponding to $\lambda\in(-c,b)$;

\item
$1$-sheeted hyperboloids oriented along $y$-axis, for $\lambda\in(b,a)$;

\item
$2$-sheeted hyperboloids, for $\lambda\in(a,+\infty)$ -- these hyperboloids are oriented along $z$-axis.
\end{itemize}
In addition, there are four degenerated quadrics: $\mathcal{Q}_{a}$, $\mathcal{Q}_{b}$, $\mathcal{Q}_{-c}$, $\mathcal{Q}_{\infty}$, that is planes $x=0$, $y=0$, $z=0$, and the plane at the infinity respectively.
In the coordinate planes, we single out the following conics: 
\begin{itemize}
\item
hyperbola $\mathcal{C}^{yz}_{a}\ :\ -\dfrac{y^2}{a-b}+\dfrac{z^2}{c+a}=1$ in the plane $x=0$;

\item
ellipse $\mathcal{C}^{xz}_{b}\ :\ \dfrac{x^2}{a-b}+\dfrac{z^2}{c+b}=1$ in the plane $y=0$;

\item
ellipse $\mathcal{C}^{xy}_{-c}\ :\ \dfrac{x^2}{a+c}+\dfrac{y^2}{b+c}=1$ in the plane $z=0$.
\end{itemize}

Notice that a confocal family of quadrics in the three-dimensional Euclidean space contains only $3$ types of quadrics, see Figure \ref{fig:konfE3}.

\subsubsection*{Tropic curves on quadrics in the three-dimensional Minkowski space and discriminant sets $\Sigma^{\pm}$}
On each quadric, notice the \emph{tropic curves} -- the set of points where the induced metrics on the tangent plane is degenerate.

Since the tangent plane at point $(x_0,y_0,z_0)$ of $\mathcal{Q}_{\lambda}$ is given by the equation:
$$
\frac{xx_0}{a-\lambda}+\frac{yy_0}{b-\lambda}+\frac{zz_0}{c+\lambda}=1,
$$
and the induced metric is degenerate if and only if the parallel plane that contains the origin is tangent to the light-like cone $x^2+y^2-z^2=0$, i.e.:
$$
\frac{x_0^2}{(a-\lambda)^2}+\frac{y_0^2}{(b-\lambda)^2}-\frac{z_0^2}{(c+\lambda)^2}=0,
$$
we get that the tropic curves on $\mathcal{Q}_{\lambda}$ are the intersection of the quadric with the cone:
\begin{equation*}
\frac{x^2}{(a-\lambda)^2}+\frac{y^2}{(b-\lambda)^2}-\frac{z^2}{(c+\lambda)^2}=0,
\end{equation*}
see \cite{KhTab2009}.

Now, consider the set of the tropic curves on all quadrics of the family (\ref{eq:confocal.quadrics3}).

\begin{proposition}\label{prop:tropic.surface}
The union of the tropic curves on all quadrics of (\ref{eq:confocal.quadrics3}) is a union of two ruled surfaces $\Sigma^+$ and $\Sigma^-$ which can be parametrically represented as:
\begin{align*}
&\Sigma^+\ :\quad
x = \frac{a-\lambda}{\sqrt{a+c}}\cos t,\quad
y = \frac{b-\lambda}{\sqrt{b+c}}\sin t,\quad
z = (c+\lambda)\sqrt{\frac{\cos^2t}{a+c}+\frac{\sin^2t}{b+c}},\\
&\Sigma^-\ :\quad
x = \frac{a-\lambda}{\sqrt{a+c}}\cos t,\quad
y = \frac{b-\lambda}{\sqrt{b+c}}\sin t,\quad
z = -(c+\lambda)\sqrt{\frac{\cos^2t}{a+c}+\frac{\sin^2t}{b+c}},\\
&\text{with}\ \lambda\in\mathbf{R},\quad t\in[0,2\pi).
\end{align*}
The intersection of these two surfaces is an ellipse in the $xy$-plane:
$$
\Sigma^+\cap\Sigma^-\ : \ \frac{x^2}{a+c}+\frac{y^2}{b+c}=1,\ z=0.
$$
The two surfaces $\Sigma^+$, $\Sigma^-$ are developable as embedded into Euclidean space.
Moreover, their generatrices are all light-like.
\end{proposition}

Surfaces $\Sigma^+$ and $\Sigma^-$ from Proposition \ref{prop:tropic.surface} are represented on Figure \ref{fig:tropic.surface}.
\begin{figure}[h]
\includegraphics[width=5.6cm, height=7.5cm]{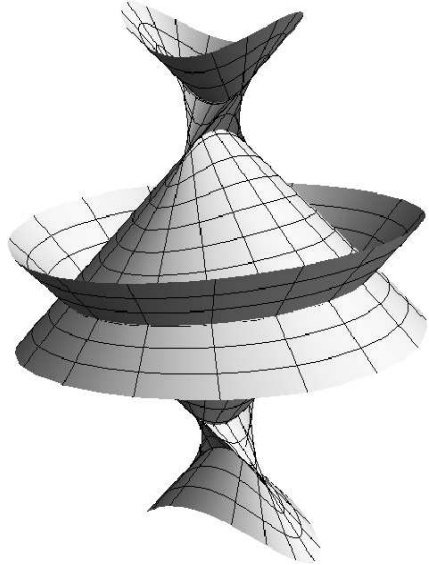}
\caption{The union of all tropic curves of a confocal family.}\label{fig:tropic.surface}
\end{figure}

In \cite{Pei1999}, a definition of generalization of Gauss map to surfaces in the three-dimensional Minkowski space is suggested.
Namely, the \emph{pseudo vector product} is introduced as:
$$
\mathbf{x}\wedge\mathbf{y}=
\left|\begin{array}{rrr}
x_1 & x_2 & x_3\\
y_1 & y_2 & y_3\\
e_1 & e_2 & -e_3
\end{array}\right|=(x_2y_3-x_3y_2,\ x_3y_1-x_1y_3,\ -(x_1y_2-x_2y_1))=E_{2,1}(\mathbf{x}\times\mathbf{y}).
$$
It is easy to check that
$\left<\mathbf{x}\wedge\mathbf{y},\mathbf{x}\right>_{2,1}=\left<\mathbf{x}\wedge\mathbf{y},\mathbf{y}\right>_{2,1}=0$.

Then, for surface $S:U\to\mathbf{E}^{2,1}$, with $U\subset\mathbf{R}^2$,
\emph{the Minkowski Gauss map} is defined as:
$$
\mathcal{G}\ :\ U\to\mathbf{RP}^2,
\quad
\mathcal{G}(x_1,x_2)= \mathbf{P}\left(\frac{\partial S}{\partial x_1}\wedge\frac{\partial S}{\partial x_2}\right),
$$
where $\mathbf{P}:\mathbf{R}^3\setminus\{(0,0,0)\}\to\mathbf{RP}^2$ is the usual projectivization.

Since $\mathbf{r}_{\lambda}\wedge\mathbf{r}_t$ is light-like for all $\lambda$ and $t$, the Minkowski Gauss map of surfaces $\Sigma^{\pm}$ is singular at all points.

\emph{The pseudo-normal vectors} to $\Sigma^{\pm}$ are all light-like, thus these surfaces are \emph{light-like developable}, as defined in \cite{CI2010}.
There, a classification of such surfaces is given -- each is one part-by-part contained in the following:
\begin{itemize}
 \item
a light-like plane;
 \item
a light-like cone;
 \item
a tangent surface of a light-like curve.
\end{itemize}
Since $\Sigma^{+}$ and $\Sigma^{-}$ are contained neither in a plane nor in a cone, we expect that they will be tangent surfaces of some light-like curve, which is going to be shown in the sequel, see Corollary \ref{cor:tangent.surface} later in this section. 

On each of the surfaces $\Sigma^{+}$, $\Sigma^{-}$, we can notice that tropic lines corresponding to $1$-sheeted hyperboloids oriented along $y$-axies form one curved tetrahedron, see Figure \ref{fig:tropic.surface}.
Denote the tetrahedra by $\mathcal{T}^{+}$ and $\mathcal{T}^{-}$ respectively: they are symmetric with respect to the $xy$-plane.
On Figure \ref{fig:tetra}, tetrahedron $\mathcal{T}^{+}\subset\Sigma^{+}$ is shown.
\begin{figure}[h]
\includegraphics[width=6cm, height=6cm]{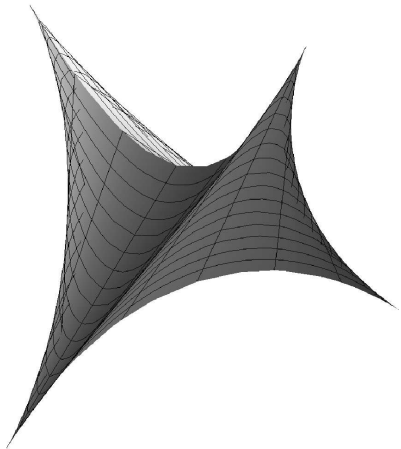}
\caption{Curved thetrahedron $\mathcal{T}^{+}$: the union of all tropic curves on $\Sigma^{+}$ corresponding to $\lambda\in(b,a)$.}\label{fig:tetra}
\end{figure}

Let us summarize the properties of these tetrahedra.
\begin{proposition}\label{prop:tetra}
Consider the subset $\mathcal{T}^{+}$ of $\Sigma^{+}$ determined by the condition $\lambda\in[b,a]$.
This set is a curved tetrahedron, with the following properties:
\begin{itemize}
 \item 
its verteces are:
\begin{align*}
&V_1\left(\frac{a-b}{\sqrt{a+c}},\ 0,\ \frac{b+c}{\sqrt{a+c}}\right),\quad
V_2\left(-\frac{a-b}{\sqrt{a+c}},\ 0,\ \frac{b+c}{\sqrt{a+c}}\right),\\
&V_3\left(0,\ \frac{a-b}{\sqrt{b+c}},\ \frac{a+c}{\sqrt{b+c}}\right),\quad
V_4\left(0,\ -\frac{a-b}{\sqrt{b+c}},\ \frac{a+c}{\sqrt{b+c}}\right);
\end{align*}
 \item
the shorter arcs of conics $\mathcal{C}^{xz}_{b}$ and $\mathcal{C}^{yz}_{a}$ determined by $V_1$, $V_2$ and $V_3$, $V_4$ respectively are two edges of the tetrahedron;
 \item
those two edges represent self-intersection of $\Sigma^{+}$;
 \item
other four edges are determined by the relation:
\begin{equation}\label{eq:cusp}
-a-b+2\lambda+(a-b)\cos2t=0,
\end{equation}
 \item
those four edges are cuspidal edges of $\Sigma^{+}$;
 \item
thus, at each vertex of the tetrahedron, a swallowtail singularity of $\Sigma^+$ occurs.
\end{itemize}
\end{proposition}

It can be proved that the tropic curves of the quadric $\mathcal{Q}_{\lambda_0}$ represent exactly the locus of points $(x,y,z)$ where equation
\begin{equation}\label{eq:jednacina}
\frac{x^2}{a-\lambda}+\frac{y^2}{b-\lambda}+\frac{z^2}{c+\lambda}=1
\end{equation}
has $\lambda_0$ as a multiple root.

\begin{proposition}\label{prop:tropic.light}
A tangent line to the tropic curve of a non-degenarate quadric of the family (\ref{eq:confocal.quadrics3}) is always space-like, except on a $1$-sheeted hyperboloid oriented along $y$-axis.

Tangent lines of a tropic on $1$-sheeted hyperboloids oriented along $y$-axis are light-like exactely at four points, while at other points of the tropic curve, the tangents are space-like.

Moreover, a tangent line to the tropic of a quadric from (\ref{eq:confocal.quadrics3}) belongs to the quadric if and only if it is light-like.
\end{proposition}

\begin{figure}[h]
\includegraphics[width=6cm, height=8.4cm]{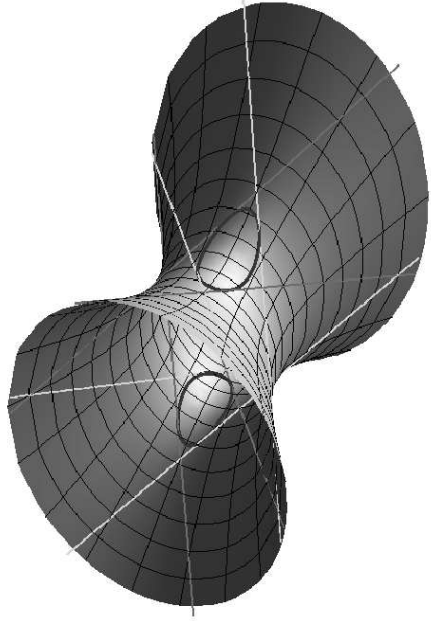}
\caption{The tropic curves and its light-like tangents on a hyperboloid.}\label{fig:hip-tropic-light}
\end{figure}

\begin{remark}
In other words,
the only quadrics of the family (\ref{eq:confocal.quadrics3}) that may contain a tangent to its tropic curve are  $1$-sheeted hyperboloids oriented along $y$-axis,
and those tangents are always light-like.
The tropic curves and their light-like tangents on such an hyperboloid are shown on Figure \ref{fig:hip-tropic-light}.
\end{remark}

Notice that equations obtained in the proof of Proposion \ref{prop:tropic.light} are equivalent to equation (\ref{eq:cusp}) of Proposition \ref{prop:tetra}, which leads to the following:

\begin{proposition}\label{prop:tetra.hyp}
Each generatrix of $\Sigma^{+}$ and $\Sigma^{-}$ is contained in one $1$-sheeted hyperboloid oriented along $y$-axis from (\ref{eq:confocal.quadrics3}).
Moreover, such a generatrix is touching at the same point one of the tropic curves of the hyperboloid and one of the cusp-like edges of the corresponding curved tetrahedron. 
\end{proposition}

\begin{corollary}\label{cor:tangent.surface}
Surfaces $\Sigma^{+}$ and $\Sigma^{-}$ are tangent surfaces of the cuspidal edges of thetrahedra $\mathcal{T}^+$ and $\mathcal{T}^-$ respectively.
\end{corollary}

In next propositions, we give further analysis the light-like tangents to the tropic curves on an $1$-sheeted hyperboloid oriented along $y$-axis.

\begin{proposition}\label{prop:hyp.y.presek}
For a fixed $\lambda_0\in(b,a)$,
consider a hyperboloid $\mathcal{Q}_{\lambda_0}$ from (\ref{eq:confocal.quadrics3}) and an arbitrary point $(x,y,z)$ on $\mathcal{Q}_{\lambda_0}$.
Equation (\ref{eq:jednacina}) has, along with $\lambda_0$, two other roots in $\mathbf{C}$: denote them by $\lambda_1$ and $\lambda_2$.
Then $\lambda_1$=$\lambda_2$ if and only if $(x,y,z)$ is placed on a light-like tangent to a tropic curve of $\mathcal{Q}_{\lambda_0}$.
\end{proposition}

\begin{proof}
Follows from the fact that the light-like tangents are contained in the $\Sigma^{+}\cup\Sigma^{-}$, see Propositions \ref{prop:tropic.surface} and \ref{prop:tetra.hyp}.
\end{proof}

\begin{proposition}
Two light-like lines on a one-sheeted hyperboloid oriented along $y$-axis from (\ref{eq:confocal.quadrics3}) are either skew or intersect each other on a degenerate quadric from (\ref{eq:confocal.quadrics3}).
\end{proposition}

\begin{proof}
Follows from the fact that the hyperboloid is symmetric with respect to the coordinate planes.
\end{proof}

\begin{lemma}\label{lemma:tropic.non.y}
Consider a non-degenerate quadric $\mathcal{Q}_{\lambda_0}$, which is not a hyperboloid oriented along $y$-axis,
i.e.\ $\lambda_0\not\in[b,a]\cup\{-c\}$.
Then each point of $\mathcal{Q}_{\lambda_0}$ which is not on one of the tropic curves is contained in two additional distinct quadrics from the family (\ref{eq:confocal.quadrics3}).

Consider two points $A$, $B$ of $\mathcal{Q}_{\lambda_0}$, which are placed in the same connected component bounded by the tropic curves, and denote by $\lambda'_A$, $\lambda''_A$ and $\lambda'_B$, $\lambda''_B$ the solutions, different than $\lambda_0$, of equation (\ref{eq:jednacina}) corresponding to $A$ and $B$ respectively.
Then, if $\lambda_0$ is smaller (resp.\ bigger, between) than $\lambda'_A$, $\lambda''_A$, it is also smaller (resp.\ bigger, between) than $\lambda'_B$, $\lambda''_B$.
\end{lemma}

\begin{lemma}\label{lemma:tropic.y}
Let $\mathcal{Q}_{\lambda_0}$ be a hyperboloid oriented along $y$-axis, $\lambda_0\in(b,a)$, and $A$, $B$ two points of $\mathcal{Q}_{\lambda_0}$, which are placed in the same connected component bounded by the tropic curves and light-like tangents. 
Then, if $A$ is contained in two more quadrics from the family (\ref{eq:confocal.quadrics3}), the same is true for $B$.

In this case, denote by $\lambda'_A$, $\lambda''_A$ and $\lambda'_B$, $\lambda''_B$ the real solutions, different than $\lambda_0$, of equation (\ref{eq:jednacina}) corresponding to $A$ and $B$ respectively.
Then, if $\lambda_0$ is smaller (resp.\ bigger, between) than $\lambda'_A$, $\lambda''_A$, it is also smaller (resp.\ bigger, between) than $\lambda'_B$, $\lambda''_B$.

On the other hand, if $A$ is not contained in any other quadric from (\ref{eq:confocal.quadrics3}), then the same is true for all points of its connected component.
\end{lemma}

\begin{proof}
The proof of both Lemmae \ref{lemma:tropic.non.y} and \ref{lemma:tropic.y} follows from the fact that the solutions of (\ref{eq:jednacina}) are continuously changed through the space and that two of the solutions coincide exactly on tropic curves and their light-like tangents.
\end{proof}

\subsubsection*{Generalized Jacobi coordinates and relativistic quadrics in the three-dimensional Minkowski space}
\begin{definition}
\emph{Generalized Jacobi coordinates} of point $(x,y,z)$ in the three-di\-men\-sio\-nal Minkowski space $\mathbf{E}^{2,1}$ is the unordered triplet of solutions of equation (\ref{eq:jednacina}).
\end{definition}

Note that any of the following cases may take place:
\begin{itemize}
 \item
generalized Jacobi coordinates are real and different;
 \item
only one generalized Jacobi coordinate is real;
 \item
generalized Jacobi coordinates are real, but two of them coincide;
 \item
all three generalized Jacobi coordinates are equal. 
\end{itemize}

Lemmae \ref{lemma:tropic.non.y} and \ref{lemma:tropic.y} will help us to define relativistic types of quadrics in the $3$-dimensional Minkowski space.
Consider connected components of quadrics from (\ref{eq:confocal.quadrics3}) bounded by tropic curves and, for $1$-sheeted hyperboloids oriented along $y$-axis, their light-light tangent lines.
Each connected component will represent \emph{a relativistic quadric}.

\begin{definition}
A component of quadric $\mathcal{Q}_{\lambda_0}$ is \emph{of relativistic type $E$} if, at each of its points, $\lambda_0$ is smaller than the other two generalized Jacobi coordinates.

A component of quadric $\mathcal{Q}_{\lambda_0}$ is \emph{of relativistic type $H^1$} if, at each of its points, $\lambda_0$ is between the other two generalized Jacobi coordinates.

A component of quadric $\mathcal{Q}_{\lambda_0}$ is \emph{of relativistic type $H^2$} if, at each of its points, $\lambda_0$ is bigger than the other two generalized Jacobi coordinates.

A component of quadric $\mathcal{Q}_{\lambda_0}$ is \emph{of relativistic type $0$} if, at each of its points, $\lambda_0$ is the only real generalized Jacobi coordinate.
\end{definition}

Lemmae \ref{lemma:tropic.non.y} and \ref{lemma:tropic.y} guarantee that types of relativistic quadrics are well-defined, i.e.\ that to each such a quadric a unique type $E$, $H^1$, $H^2$, or $0$ can be assigned.

\begin{definition}
Suppose $(x,y,z)$ is a point of the three-di\-men\-sio\-nal Minkowski space $\mathbf{E}^{2,1}$ where equation \ref{eq:jednacina} has real and different solutions.
\emph{Decorated Jacobi coordinates} of that point is the ordered triplet of pairs:
$$
(E,\lambda_1),\quad
(H^1,\lambda_2),\quad
(H^2,\lambda_3),
$$
of generalized Jacobi coordinates and the corresponding types of relativistic quadrics.
\end{definition}

Now, we are going to analyze the arrangement of the relativistic quadrics.
Let us start with their intersections with the coordinate planes.

\subsubsection*{Intersection with the $xy$-plane}

In the $xy$-plane, the Minkowski metrics is reduced to the Euclidean one.
The family (\ref{eq:confocal.quadrics3}) is intersecting this plane by the following family of confocal conics:
\begin{equation}\label{eq:conf.xy}
\mathcal{C}^{xy}_{\lambda}\ :\ \frac{x^2}{a-\lambda}+\frac{y^2}{b-\lambda}=1,
\end{equation}
see Figure \ref{fig:relativistic_xy}.
\begin{figure}[h]
\input{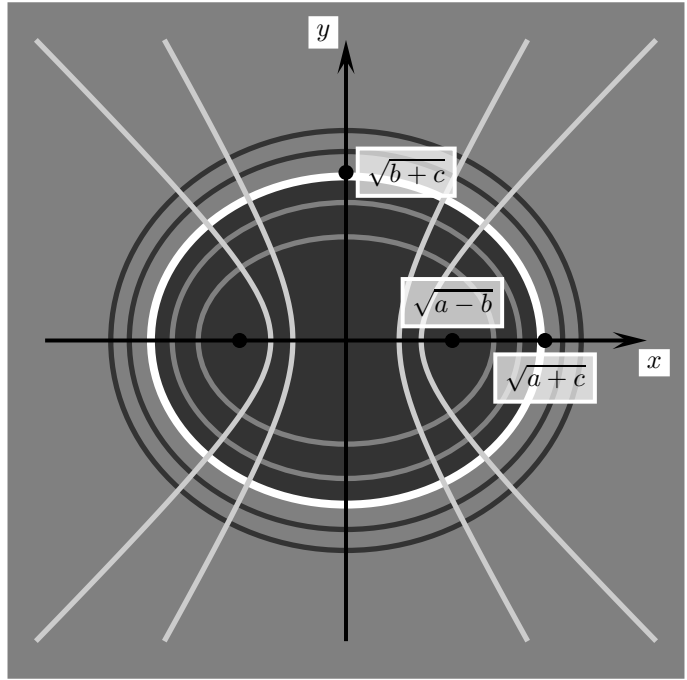}
\caption{Intersection of relativistic quadrics with the $xy$-plane.}\label{fig:relativistic_xy}
\end{figure}

We conclude that the $xy$-plane is divided by ellipse $\mathcal{C}^{xy}_{-c}$ into two relativistic quadrics:
\begin{itemize}
 \item 
the region within $\mathcal{C}^{xy}_{-c}$ is a relativistic quadric of $E$-type;
 \item
the region outside this ellipse is of $H^1$-type.
\end{itemize}

Moreover, the types of relativic quadrics intersecting the $xy$-plane are:
\begin{itemize}
 \item
the components of ellipsoids are of $H^1$-type;
 \item
the components of $1$-sheeted hyperboloids oriented along $y$-axis of $H^2$-type;
 \item
the components of $1$-sheeted hyperboloids oriented along $z$-axis of $E$-type.
\end{itemize}

On Figure \ref{fig:relativistic_xy}, type $E$ quadrics are coloured in dark gray, type $H^1$ medium gray, and type $H^2$ light gray.
The same colouring rule is applied on Figures \ref{fig:relativistic_xz}-\ref{fig:relativistic.3d}.

\subsubsection*{Intersection with the $xz$-plane}

In the $xz$-plane, the reduced metrics is the Minkowski one.
The intersection of family (\ref{eq:confocal.quadrics3}) with this plane is the following family of confocal conics:
\begin{equation}\label{eq:conf.xz}
\mathcal{C}^{xz}_{\lambda}\ :\ \frac{x^2}{a-\lambda}+\frac{z^2}{c+\lambda}=1,
\end{equation}
see Figure \ref{fig:relativistic_xz}.
\begin{figure}[h]
\input{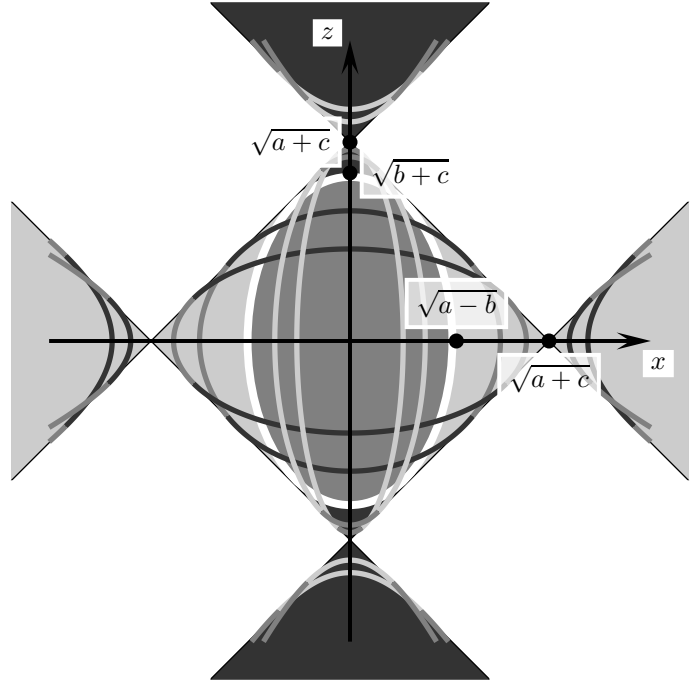}
\caption{Intersection of relativistic quadrics with the $xz$-plane.}\label{fig:relativistic_xz}
\end{figure}

The plane is divided by ellipse $\mathcal{C}^{xz}_{b}$ and the four joint tangents of (\ref{eq:conf.xz}) into $13$ parts:
\begin{itemize}
\item
the part within $\mathcal{C}^{xz}_{b}$ is a relativistic quadric of $H^1$-type;

\item
four parts placed outside of $\mathcal{C}^{xz}_{b}$ that have non-empty intersection with the $x$-axis are of $H^2$-type;

\item
four parts placed outside of $\mathcal{C}^{xz}_{b}$ that have non-empty intersection with the $z$-axis are of $E$-type;

\item
the four remaining parts are of $0$-type and no quadric from the family (\ref{eq:confocal.quadrics3}), except the degenerated $\mathcal{Q}_{b}$, is passing through any of their points.
\end{itemize}

\subsubsection*{Intersection with the $yz$-plane}

As in the previous case, in the $yz$-plane, the reduced metrics is the Minkowski one.
The intersection of family (\ref{eq:confocal.quadrics3}) with this plane is the following family of confocal conics:
\begin{equation}\label{eq:conf.yz}
\mathcal{C}^{yz}_{\lambda}\ :\ \frac{y^2}{b-\lambda}+\frac{z^2}{c+\lambda}=1,
\end{equation}
see Figure \ref{fig:relativistic_yz}.
\begin{figure}[h]
\input{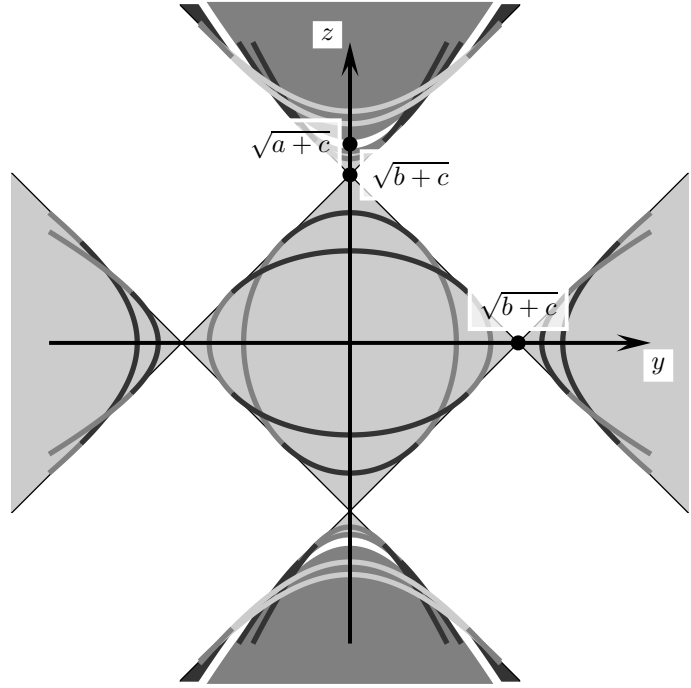}
\caption{Intersection of relativistic quadrics with the $yz$-plane.}\label{fig:relativistic_yz}
\end{figure}

The plane is divided by hyperbola $\mathcal{C}^{yz}_{a}$ and joint tangents of (\ref{eq:conf.yz}) into $15$ parts:
\begin{itemize}
\item
the two convex parts determined by $\mathcal{C}^{yz}_{a}$ are relativistic quadric of $H^1$-type;

\item
five parts placed outside of $\mathcal{C}^{yz}_{a}$ that have non-empty intersection with the coordinate axes are of $H^2$-type;

\item
four parts, each one placed between $\mathcal{C}^{yz}_{a}$ and one of the joint tangents of (\ref{eq:conf.yz}) are of $E$-type;

\item
through points of the four remaining parts no quadric from the family (\ref{eq:confocal.quadrics3}), 
except the degenerated $\mathcal{Q}_{a}$, is passing.
\end{itemize}

Intersection of relativistic quadrics with the coordinate planes is shown in Figure \ref{fig:relativistic.3d}.

\begin{figure}[h]
\input{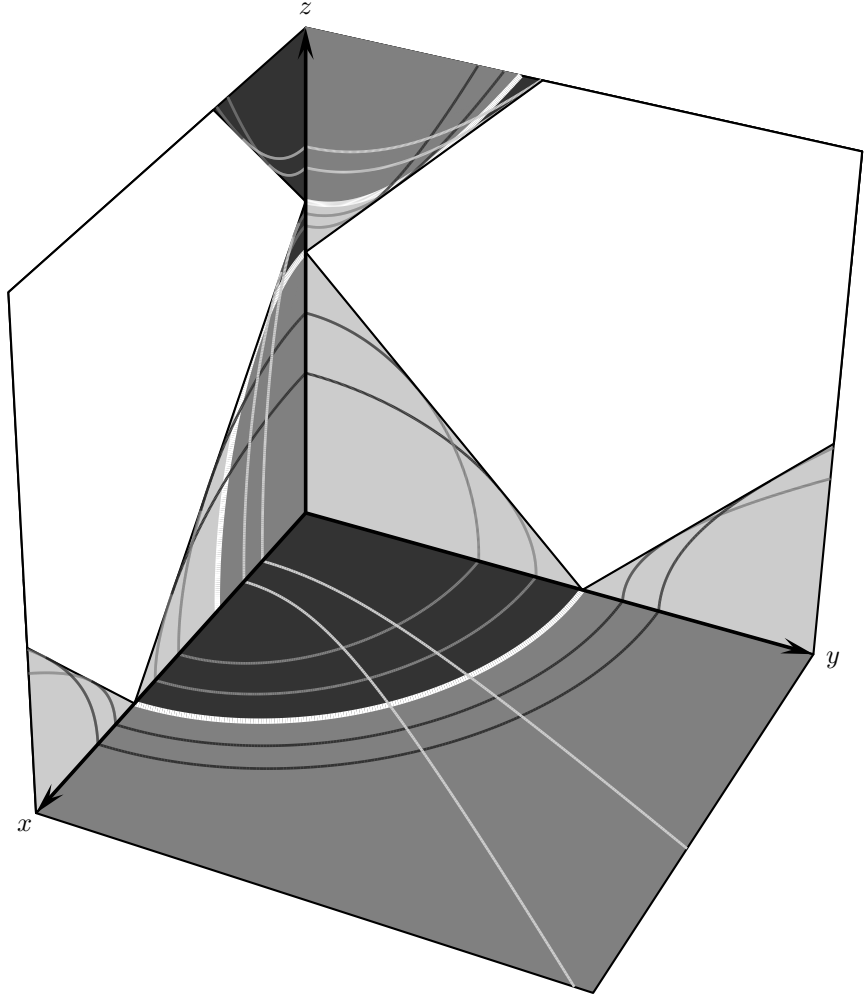}
\caption{Intersection of relativistic quadrics with coordinate planes.}\label{fig:relativistic.3d}
\end{figure}

Let us notice that from the above analysis, using Lemmae \ref{lemma:tropic.non.y} and \ref{lemma:tropic.y}, we can determine the type of each relativistic quadric with a non-empty intersection with some of the coordinate hyper-planes.

\subsubsection*{$1$-sheeted hyperboloids oriented along $z$-axis: $\lambda\in(-\infty,-c)$}
Such a hyperboloid is divided by its tropic curves into three connected components -- two of them are unbounded and mutually symmetric with respect to the $xy$-plane, while the third one is the bounded annulus placed between them.
The two symmetric ones are of $H^1$-type, while the third one is of $E$-type.

\subsubsection*{Ellipsoids: $\lambda\in(-c,b)$}
An ellipsoid is divided by the tropic curves into three bounded connected components -- two of them are mutually symmetric with respect to the $xy$-plane, while the third one is the annulus placed between them.
In this case, the symmetric components represent relativistic quadrics of $E$-type.
The annulus is of $H^1$-type.

\subsubsection*{$1$-sheeted hyperboloids oriented along $y$-axis: $\lambda\in(b,a)$}
The decomposition of those hyperboloids into relativistic quadrics is more complicated and interesting than for the other types of quadrics from (\ref{eq:confocal.quadrics3}).
By its two tropic curves and their eight light-like tangent lines, such a hyperboloid is divided into $28$ connected components:
\begin{itemize}
 \item
two bounded components placed inside the tropic curves are of $H^1$-type;

 \item
four bounded components placed between the tropic curves and light-like tangents, such that they have non-empty intersections with $xz$-plane are of $H^2$-type;

 \item
four bounded components placed between the tropic curves and light-like tangents, such that they have non-empty intersections with $yz$-plane are of $E$-type;

 \item
two bounded components, each limited by four light-like tangents, are of $H^2$-type;

 \item
four unbounded components, each limited by two light-like tangents, such that they have non-empty intersections with the $xy$-plane, are of $H^2$-type;

 \item
four unbounded components, each limited by two light-like tangents, such that they have non-empty intersections with the $yz$-plane, are of $E$-type;

 \item
eight unbounded components, each limited by four light-like tangents, are sets of points not contained in any other quadric from (\ref{eq:confocal.quadrics3}).
\end{itemize}

\subsubsection*{$2$-sheeted hyperboloids: $\lambda\in(a,+\infty)$}
Such a hyperboloid is by its tropic curves divided into four connected components:
two bounded ones are of $H^2$-type, while the two unbounded are of $H^1$-type.

\subsubsection*{Decorated Jacobi coordinates and relativistic quadrics in $d$-di\-men\-sio\-nal pse\-udo-Euclidean space}
Now we are going to introduce relativistic quadrics and their types in confocal family (\ref{eq:confocal-pseudo}) in the $d$-dimensional pseudo-Euclidean space $\mathbf{E}^{k,l}$.

\begin{definition}
\emph{Generalized Jacobi coordinates} of point $x$ in the $d$-di\-men\-sio\-nal pseudo-Euclidean space $\mathbf{E}^{k,l}$ is the unordered $d$-tuple of solutions $\lambda$ of equation:
\begin{equation}\label{eq:jednacinad}
\frac{x_1^2}{a_1-\lambda} +\dots+ \frac{x_k^2}{a_k-\lambda} +
\frac{x_{k+1}^2}{a_{k+1}+\lambda} +\dots+
\frac{x_d^2}{a_d+\lambda}=1.
\end{equation} 
\end{definition}

As already mentioned in Section \ref{sec:pseudo.confocal}, this equation has either $d$ or $d-2$ real solutions.
Besides, some of the solutions may be multiple.

The set $\Sigma_d$ of points $x$ in $\mathbf{R}^d$ where equation (\ref{eq:jednacinad}) has multiple solutions is an algebraic hyper-suface.
$\Sigma_d$ divides each quadric from (\ref{eq:confocal-pseudo}) into several connected components.
We call these components \emph{relativistic quadrics}.

Since the generalized Jacobi coordinates depend continuosly on $x$, the following definition can be made:
\begin{definition}
We say that a relativistic quadric placed on $\mathcal{Q}_{\lambda_0}$ is \emph{of type $E$} if, at each of its points, $\lambda_0$ is smaller than the other $d-1$ generalized Jacobi coordinates.

We say that a relativistic quadric placed on $\mathcal{Q}_{\lambda_0}$ is \emph{of type $H^i$} $(1<i<d-1)$ if, at each of its points, $\lambda_0$ is greater than other $i$ generalized Jacobi coordinates, and smaller than $d-i-1$ of them.

We say that a relativistic quadric placed on $\mathcal{Q}_{\lambda_0}$ is \emph{of type $0^i$} $(0<i<d-2)$ if, at each of its points, $\lambda_0$ is greater than other $i$ real generalized Jacobi coordinates, and smaller than $d-i-2$ of them.
\end{definition}

It would be interesting to analyze properties of the discriminant manifold $\Sigma_d$, as well as the combinatorial structure of the arrangement of relativistic quadrics, as it is done for $d=3$.
Remark that this description would have $[d/2]$ substantially different cases in each dimension, depending on choice of $k$ and $l$.

\begin{definition}
Suppose $(x_1,\dots,x_d)$ is a point of the $d$-dimensional Minkowski space $\mathbf{E}^{k,l}$ where equation (\ref{eq:jednacinad}) has real and different solutions.
\emph{Decorated Jacobi coordinates} of that point is the ordered $d$-tuplet of pairs:
$$
(E,\lambda_1),\quad
(H^1,\lambda_2),\quad
\dots,\quad
(H^{d-1},\lambda_d),
$$
of generalized Jacobi coordinates and the corresponding types of relativistic quadrics.
\end{definition}

Since we will consider billiard system within ellipsoids in the pseudo-Euclidean space, it is of interest to analyze behaviour of decorated Jacobi coordinates inside an ellipsoid.

\begin{proposition}\label{prop:svojstva12}
Let $\mathcal{E}$ be ellipsoid in $\mathbf{E}^{k,l}$ given by (\ref{eq:ellipsoid}).
We have:
\begin{itemize}
 \item[PE1]
each point inside $\mathcal{E}$ is the intersection of exactly $d$ quadrics from (\ref{eq:confocal-pseudo});
moreover, all these quadrics are of different relativistic types;
 \item[PE2]
the types of these quadrics are $E$, $H^1$, \dots, $H^{d-1}$ -- each type corresponds to one of the disjoint intervals of the parameter $\lambda$:
$$(-a_d,-a_{d-1}),\ (-a_{d-1},-a_{d-2}),\ \dots,\ (-a_{k+1},0),\ (0,a_k),\ (a_k,a_{k-1}),\ \dots,\ (a_2,a_1).$$
\end{itemize}
\end{proposition}

\begin{proof}
The function given by the left-hand side of (\ref{eq:jednacinad}) is continous and strictly monotonous in each interval $(-a_d,-a_{d-1})$, $(-a_{d-1},-a_{d-2})$, \dots, $(-a_{k+2},-a_{k+1})$, $(a_k,a_{k-1})$, \dots, $(a_2,a_1)$ with infinite values at their endpoints.
Thus, equation (\ref{eq:jednacinad}) has one solution in each of them.
On the other hand, in $(-a_{k+1},a_k)$, the function is tending to $+\infty$ at the endpoints, and has only one extreme value -- the minimum.
Since the value of the function for $\lambda=0$ is less than $1$ for a point inside $\mathcal{E}$, it follows that equation (\ref{eq:jednacinad}) will have two solutions in $(-a_{k+1},a_k)$ -- one positive and one negative.
\end{proof}


\subsection{Billiards within quadrics and their periodic trajectories}
\label{sec:periodic}
In this section, we are going to derive first further properties of ellipsoidal billiards in the pseudo-Euclidean spaces.
First we find in Theorem \ref{th:type} a simple and effective criterion for determining the type of a billiard trajectory, knowing its caustics.
Then we derive properties PE3--PE5 in Propostion \ref{prop:PE3-PE5}.
After that, we prove the generalization of Poncelet theorem for ellipsoidal billiards in pseudo-Euclidean spaces and derive the corresponding Cayley-type conditions, giving a complete analytical description of periodic billiard trajectories in arbitrary dimension.
These results are contained in Theorems \ref{th:cayley-pseudo} and \ref{th:poncelet-pseudo}.

\subsubsection*{Ellipsoidal billiards}
\subsubsection*{Ellipsoidal billiard}

Billiard motion within an ellipsoid in the pseudo-Euclidean space is a motion
which is uniformly straightforward inside the ellipsoid, and obeys
the reflection law on the boundary.
Further, we will consider billiard motion within ellipsoid $\mathcal{E}$, given by equation
(\ref{eq:ellipsoid}) in $\mathbf{E}^{k,l}$.
The family of quadrics confocal with $\mathcal{E}$ is (\ref{eq:confocal-pseudo}).

Since functions $F_i$ given by (\ref{eq:integralsF}) are integrals
of the billiard motion (see \cites{Mo1980,Audin1994,KhTab2009}), we
have that for each zero $\lambda$  of the equation (\ref{eq:discr}),
the corresponding quadric $\mathcal{Q}_{\lambda}$ is a caustic of
the billiard motion, i.e.\ it is tangent to each segment of the
billiard trajectory passing through the point $x$ with the velocity
vector $v$.

Note that, according to Theorem \ref{th:parametri.kaustike}, for a point placed inside $\mathcal{E}$, there are $d$ real solutions of equation (\ref{eq:jednacinad}). In other words, there are $d$ quadrics from the family (\ref{eq:confocal-pseudo}) containing such a point, although some of them may be multiple. Also, by Proposition \ref{prop:kaustike} and Theorem \ref{th:parametri.kaustike}, a billiard trajectory within an ellipsoid will always have $d-1$ caustics.

According to Remark \ref{remark:type}, all segments of a billiard trajectory within $\mathcal{E}$ will be of the same type.
Now, we can apply the reasoning from Section \ref{sec:pseudo.confocal} to billiard trajectories:

\begin{theorem}\label{th:type}
In the $d$-dimensional pseudo-Euclidean space $\mathbf{E}^{k,l}$, consider a billiard trajectory within ellipsoid $\mathcal{E}=\mathcal{Q}_0$, and let quadrics
$\mathcal{Q}_{\alpha_1}$, \dots, $\mathcal{Q}_{\alpha_{d-1}}$ from the family (\ref{eq:confocal-pseudo}) be its caustics.
Then all billiard trajectories within $\mathcal{E}$ sharing the same caustics are of the same type: space-like, time-like, or light-like, as the initial trajectory.
Moreover, the type is determined as follows:
\begin{itemize}
\item
if $\infty\in\{\alpha_1,\dots,\alpha_{d-1}\}$, the trajectories are light-like;
\item
if $(-1)^l\cdot\alpha_1\cdot\dotsc\cdot\alpha_{d-1}>0$, the trajectories are space-like;
\item
if $(-1)^l\cdot\alpha_1\cdot\dotsc\cdot\alpha_{d-1}<0$, the trajectories are time-like.
\end{itemize}
\end{theorem}
\begin{proof}
Since values of functions $F_i$ given by (\ref{eq:integralsF}) are preserved by the billiard reflection and
$$
\sum_{i=1}^dF_i(x,v) = \langle{v,v}\rangle_{k,l},
$$
the type of the billiard trajectory depends on the sign of the sum
$\sum_{i=1}^dF_i(x,v)$.
From the equivalence of relations (\ref{eq:discr}) and (\ref{eq:polinom.P}), it follows that the sum depends only of the roots of $\mathcal{P}$, i.e.\ of parameters $\alpha_1$, \dots, $\alpha_{d-1}$ of the caustics.

Notice that the product $\alpha_1\cdot\dotsc\cdot\alpha_{d-1}$ is changed continuously on the variety of lines in $\mathbf{E}^{k,l}$ that intersect $\mathcal{E}$, with infinite singularities at light-like lines.
Besides, the subvariety of light-like lines divides the variety of all lines into subsets of space-like and time-like ones.
When passing through light-like lines, one of parameters $\alpha_i$ will pass through the infinity from positive to the negative part of the reals or vice versa;
thus, a change of sign of the product occurs simultaneously with a change of the type of line.

Now, take $\alpha_{j}=-a_{k+j}$ for $1\le j\le l$, and notice that all lines placed in the $k$-dimensional coordinate subspace $\mathbf{E}^k\times\mathbf{0}^l$ will have the corresponding degenerate caustics.
The reduced metrics is Euclidean in this subspace, thus such lines are space-like.
Since $\alpha_1$, \dots, $\alpha_k$ are positive for those lines of $\mathbf{E}^k\times\mathbf{0}^l$ that intersect $\mathcal{E}$, the statement is proved.
\end{proof}

Let us note that, in general, for the fixed $d-1$ quadrics from the confocal family, there can be found joint tangents of different types, which makes Theorem \ref{th:type} in a way unexpected. 
However, it turns out that, with fixed caustics, only lines having one type may have intersection with a given ellipsoid --- and only these lines give rise to billiard trajectories.

Next, we are going to investigate the behaviour of decorated Jacobi coordinates along ellipsoidal billiard trajectories.

\begin{proposition}\label{prop:PE3-PE5}
Let $\mathcal{T}$ be a trajectory of the billiard within ellipsoid $\mathcal{E}$ in pseudo-Euclidean space $\mathbf{E}^{k,l}$.
Denote by $\alpha_1$, \dots, $\alpha_{d-1}$ the parameters of the caustics from the confocal family (\ref{eq:confocal-pseudo}) of $\mathcal{T}$, and take $b_1$, \dots, $b_p$, $c_1$, \dots, $c_q$ as in Theorem \ref{th:parametri.kaustike}.
Then we have:
\begin{itemize}
 \item[PE3]
along $\mathcal{T}$, each generalized Jacobi coordinate takes values in exactly one of the segments:
$$
[c_{2l-1},c_{2l-2}],\ \dots,\ [c_2,c_1],\ [c_1,0],\ [0,b_1],\ [b_2,b_3],\ \dots,\ [b_{2k-2},b_{2k-1}];
$$

 \item[PE4]
along $\mathcal{T}$, each generalized Jacobi coordinate can achieve local minima and maxima only at touching points with corresponding caustics, intersection points with corresponding coordinate hyper-planes, and at reflection points;

 \item[PE5]
values of generalized Jacobi coordinates at critical points are $0$, $b_1$, \dots, $b_{2k-1}$, $c_1$, \dots, $c_{2l-1}$;
between the critical points, the coordinates are changed monotonously.
\end{itemize}
\end{proposition}

\begin{proof}
Property PE3 follows from Theorem \ref{th:parametri.kaustike}.
Along each line, the generalized Jacobi coordinates are changed continuously.
Moreover, they are monotonous at all points where the line has a transversal intersection with a non-degenerate quadric.
Thus, critical points on a line are exactly touching points with corresponding caustics and intersection points with corresponding coordinate hyper-planes.

Note that reflection points of $\mathcal{T}$ are also points of transversal intersection with all quadrics containing those points, except with $\mathcal{E}$.
Thus, at such points, $0$ will be a critical value of the corresponding generalized Jacobi coordinate, and all other coordinates are monotonous.
This proves PE4 and PE5.
\end{proof}

The properties we obtained are pseudo-Euclidean analogs of properties E3--E5, which are true for ellipsoidal billiards in Euclidean spaces.

\subsubsection*{Analytic conditions for periodic trajectories}

Now, we are going to derive the corresponding analytic conditions of Cayley's type for periodic trajectories of the ellipsoidal billiard in the pseudo-Euclidean space, and therefore to obtain the generalization of the Poncelet theorem to pseudo-Euclidean spaces.

\begin{theorem}[Generalized Cayley-type conditions]\label{th:cayley-pseudo}
In the pseudo-Euclidean space $\mathbf{E}^{k,l}$ ($k+l=d$), consider a billiard trajectory $\mathcal{T}$ within ellipsoid $\mathcal{E}$ given by equation (\ref{eq:ellipsoid}).
Let $\mathcal{Q}_{\alpha_1}$, \dots, $\mathcal{Q}_{\alpha_{d-1}}$ from confocal family (\ref{eq:confocal-pseudo}) be caustics of $\mathcal{T}$.

Then $\mathcal{T}$ is periodic with period $n$ if and only if the following condition is satisfied:
\begin{itemize}
 \item
for $n=2m$:
\begin{gather*}
\rank\left(
\begin{array}{llll}
B_{d+1} & B_{d+2} & \dots & B_{d+m-1}\\
B_{d+2} & B_{d+3} & \dots & B_{d+m}\\
\dots & \dots & \dots & \dots\\
B_{m+1} & B_{m+2} & \dots & B_{2m-1}
\end{array}
\right) < m-d+1
\quad\text{or}\\
\rank\left(
\begin{array}{llll}
B_{d+1} & B_{d+2} & \dots & B_{d+m}\\
\dots & \dots & \dots & \dots \\
B_{m} & B_{m+1} & \dots & B_{2m-1}\\
C_{m} & C_{m+1} & \dots & C_{2m-1}\\
D_{m} & D_{m+1} & \dots & D_{2m-1}
\end{array}
\right) < m-d+2;
\end{gather*}
\item
for $n=2m+1$:
\begin{gather*}
\rank\left(
\begin{array}{llll}
B_{d+1} & B_{d+2} & \dots & B_{d+m} \\
B_{d+2} & B_{d+3} & \dots & B_{d+m+1}\\
\dots & \dots & \dots & \dots\\
B_{m+1} & B_{m+2} & \dots & B_{2m}\\
C_{m+1} & C_{m+2} & \dots & C_{2m}
\end{array}
\right) < m-d+2
\quad\text{or}\\
\rank\left(
\begin{array}{llll}
B_{d+1} & B_{d+2} & \dots & B_{d+m} \\
B_{d+2} & B_{d+3} & \dots & B_{d+m+1}\\
\dots & \dots & \dots & \dots\\
B_{m+1} & B_{m+2} & \dots & B_{2m}\\
D_{m+1} & D_{m+2} & \dots & D_{2m}
\end{array}
\right) < m-d+2.
\end{gather*}
\end{itemize}
Here,
$(B_i)$, $(C_i)$, $(D_i)$ are coefficients in the Taylor expansions around $\lambda=0$ of the functions
$f(\lambda)=\sqrt{(\alpha_1-\lambda)\cdot\ldots\cdot(\alpha_{d-1}-\lambda)\cdot(a_1-\varepsilon_1\lambda)\cdot\ldots\cdot(a_d-\varepsilon_d\lambda)}$,
$\dfrac{f(\lambda)}{b_1-\lambda}$, $\dfrac{f(\lambda)}{c_1-\lambda}$ respectively.
\end{theorem}

\begin{proof}
Denote:
$$
\mathcal{P}_1(\lambda)=(\alpha_1-\lambda)\cdot\ldots\cdot(\alpha_{d-1}-\lambda)\cdot(a_1-\varepsilon_1\lambda)\cdot\ldots\cdot(a_d-\varepsilon_d\lambda).
$$
Following Jacobi \cite{JacobiGW}, along a given billiard trajectory, we consider the integrals:
\begin{equation}\label{eq:int}
\sum_{s=1}^d\int\frac{d\lambda_s}{\sqrt{\mathcal{P}_1(\lambda_s)}},
\quad
\sum_{s=1}^d\int\frac{\lambda_{s}d\lambda_s}{\sqrt{\mathcal{P}_1(\lambda_s)}},
\quad\dots,\quad
\sum_{s=1}^d\int\frac{\lambda_{s}^{d-2}d\lambda_s}{\sqrt{\mathcal{P}_1(\lambda_s)}}.
\end{equation}
By PE3 of Proposition \ref{prop:PE3-PE5}, we may suppose that:
\begin{gather*}
\lambda_1\in[0,b_1],\ \lambda_i\in[b_{2i-2},b_{2i-1}]\ \text{for}\ 2\le i\le k;
 \\
\lambda_{k+1}\in[c_1,0],\ \lambda_{k+j}\in[c_{2j-1},c_{2j-2}]\ \text{for}\ 2\le j\le l.
\end{gather*}
Along a billiard trajectory, by PE4 and PE5 of Proposition \ref{prop:PE3-PE5}, each $\lambda_{s}$ will pass through the corresponding interval monotonously from one endpoint to another and vice versa alternately.
Notice also that values $b_1$, \dots, $b_{2k-1}$, $c_1$, \dots, $c_{2l-1}$ correspond to the branching points of hyper-elliptic curve: 
\begin{equation}\label{eq:curve}
\mu^2=\mathcal{P}_1(\lambda).
\end{equation}
Thus, calculating integrals (\ref{eq:int}), we get that the billiard trajectory is closed after $n$ reflections if and only if, for some $n_1$, $n_2$ such that $n_1+n_2=n$:
$$
n\mathcal{A}(P_0)\equiv n_1\mathcal{A}(P_{b_1})+n_2\mathcal{A}(P_{c_1}),
$$
on the Jacobian of curve (\ref{eq:curve}).
Here, $\mathcal{A}$ is the Abel-Jacobi map, and $P_t$ is a point on the curve corresponding to $\lambda=t$.
Further, in the same manner as in \cite{DragRadn1998b}, we obtain the conditions as stated in the theorem.
\end{proof}

As an immediate consequence, we get:

\begin{theorem}[Generalized Poncelet theorem]\label{th:poncelet-pseudo}
In pseudo-Euclidean space $\mathbf{E}^{k,l}$ ($k+l=d$), consider a billiard trajectory $\mathcal{T}$ within ellipsoid $\mathcal{E}$.

If $\mathcal{T}$ is periodic and become closed after $n$ reflections on the ellipsoid, then any other trajectory within $\mathcal{E}$ having the same caustics as $\mathcal{T}$ is also periodic with period $n$.
\end{theorem}

\begin{remark}
The generalization of the Full Poncelet theorem from \cite{CCS1993} to pseudo-Euclidean spaces is obtained in \cite{WFSWZZ2009}.
However, only space-like and time-like trajectories were discussed there.

A Poncelet-type theorem for light-like geodesics on the ellipsoid in the three-dimensional Minkowski space is proved in \cite{GKT2007}.
\end{remark}

\begin{remark}
Theorems \ref{th:cayley-pseudo} and \ref{th:poncelet-pseudo} will also hold in symmetric and degenerated cases, that is when some of the parameters $\varepsilon_i a_i$, $\alpha_{j}$ concide, or in the case of light-like trajectories, when $\infty\in\{\alpha_{j}\mid1\le j\le d-1\}$.
In such cases, we need to apply the desingularisation of the corresponding curve, as explained in detail in our works \cites{DragRadn2006jmpa,DragRadn2008}.

When we consider light-like trajectories, then the factor containing the infinite parameter is ommited from polynomial $\mathcal{P}_1$.
\end{remark}

\begin{example}
Let us find all $4$-periodic trajectories within ellipse $\mathcal{E}$ given by
(\ref{eq:ellipse}) in the Minkowski plane, i.e.~ all conics $\mathcal{C}_{\alpha}$ from the confocal family (\ref{eq:confocal.conics}) corresponding to such trajectories.

By Theorem \ref{th:cayley-pseudo}, the condition is $B_3=0$, with
$$
\sqrt{(a-\lambda)(b+\lambda)(\alpha-\lambda)}=B_0+B_1\lambda+B_2\lambda^2+B_3\lambda^3+\dots
$$
being the Taylor expansion around $\lambda=0$.
Since
$$
B_3=\frac{(-ab-a\alpha+b\alpha)(-ab+a\alpha+b\alpha)(ab+a\alpha+b\alpha)}
{16(ab\alpha)^{5/2}},
$$
we obtain the following solutions:
$$
\alpha_1=\frac{ab}{b-a},\quad
\alpha_2=\frac{ab}{a+b},\quad
\alpha_3=-\frac{ab}{a+b}.
$$
Since $\alpha_1\not\in(-b,a)$ and $\alpha_2,\alpha_3\in(-b,a)$, conic
$\mathcal{C}_{\alpha_1}$ is a hyperbola, while $\mathcal{C}_{\alpha_2}$,
$\mathcal{C}_{\alpha_3}$ are ellipses.
\end{example}

\section{Integrable line congruences and double reflection nets}
\label{sec:cong}
As a modern scientific discipline, the Discrete Differential Geometry emerged quite recently (see \cite{BS2008book}), within a study of the lattice geometry.
So-called integrability conditions for quad-graphs have a fundamental role there.
On the other hand, the geodesics on an ellipsoid are one of the most important and exciting  examples of the classical Differential Geometry.
Billiard systems within quadrics are known to be  seen as natural discretizations of the systems of geodesics on ellipsoids.
In the sequel we are going to present them as a part of the building blocks of the foundations of the Discrete Differential Geometry.

The main elements of the systems on quad-graphs are the equations of the form $Q(x,x_1,x_2,x_{12})=0$ on quadrilaterals, where $Q$ is a multiaffine polynomial, that is a polynomial of degree one in each argument.
Such equations  are called \emph{quad-equations}.
The field variables $x_i$  are assigned to four vertices of a quadrilateral as in Figure \ref{fig:quad}.
\begin{figure}[h]
\centering
\input{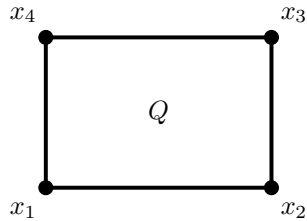}
\caption{Quad-equation $Q(x_1,x_2,x_3,x_4)=0$.}\label{fig:quad}
\end{figure}
The quad equation  can be solved for each variable, and the solution is a rational function of the other three variables.
Following \cite{ABS2009}, we consider the idea of integrability as consistency, see Figure \ref{fig:cube}.
\begin{figure}[h]
\centering
\input{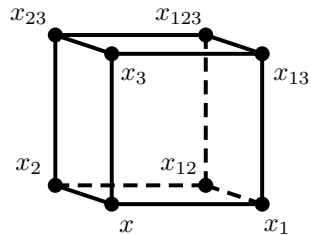}
 \caption{3D-consistency.}\label{fig:cube}
\end{figure}
We assign six quad-equations to the faces of a coordinate cube.
The system is said to be \emph{$3D$-consistent} if the three values for $x_{123}$  obtained from the equations on right, back, and top faces coincide for arbitrary initial data.

We will be interested  here in a geometric version of the integrable quad graphs, with lines in $P^d$ playing a role of the vertex fields.
We will denote by $\mathcal L^d$  the Grassmannian $\Gr(2,d+1)$ of two-dimensional vector subspaces of the $(d+1)$-dimensional vector space, $d\ge2$.

\subsection{Billiard algebra and quad-graphs}
\label{sec:algebra.quad}

This section is devoted to a quad-graph interpretation from \cite{DragRadn2012jnmp} of some results obtained previously using billiard algebra.

Let us start from a theorem on confocal families of quadrics from \cite{DragRadn2008}:

\begin{theorem}[Six-pointed star theorem]\label{th:zvezda}
Let $\mathcal F$ be a family of confocal quadrics in $\mathbf P^3$.
There exist configurations consisting of twelve planes in $\mathbf{P}^3$ with the following properties:
\begin{itemize}
\item
The planes may be organized in eight triplets, such that each plane in a triplet is tangent to a different quadric from $\mathcal F$ and the three touching points are collinear.
Every plane in the configuration is a member of two triplets.

\item
The planes may be organized in six quadruplets, such that the planes in each quadruplet belong to a pencil and are tangent to two different quadrics from $\mathcal F$.
Every plane in the configuration is a member of two quadruplets.
\end{itemize}
Moreover, such a configuration is determined by three planes tangent to three different quadrics from $\mathcal F$, with collinear touching points.
\end{theorem}

Such a configuration of planes in the dual space $\mathbf{P}^{3*}$ is shown in Figure \ref{fig:zvezda}: each plane corresponds to a vertex of the polygonal line.
\begin{figure}[h]
\centering
\input{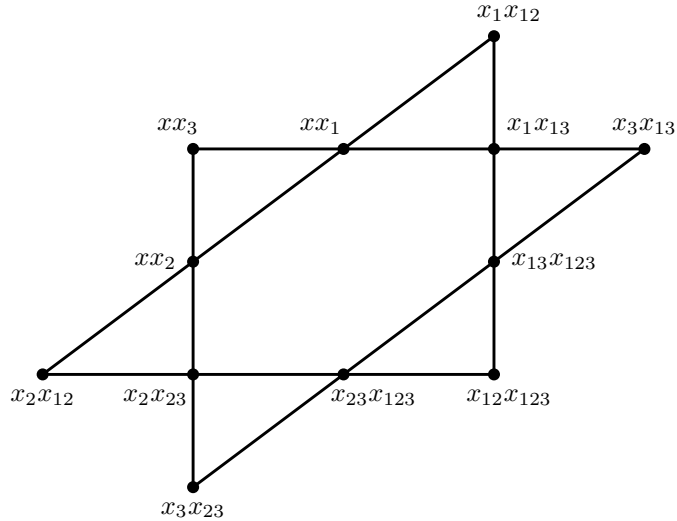}
\caption{A configuration of planes from Theorem {th:zvezda}}\label{fig:zvezda}
\end{figure}

To understand the notation used in Figure \ref{fig:zvezda}, let us recall the construction leading to configurations from Theorem \ref{th:zvezda}.
Take $\mathcal{Q}_1$, $\mathcal{Q}_2$, $\mathcal{Q}_3$ to be quadrics from 
$\mathcal{F}$, and $\alpha$, $\beta$, $\gamma$ respectively their tangent planes such that the touching points $A$, $B$, $C$ are collinear.
Denote by $x$ the line containing these three points, and by $x_1$, $x_2$, $x_3$ the lines obtained from $x$ by reflections off
$\mathcal{Q}_1$, $\mathcal{Q}_2$, $\mathcal{Q}_3$ at $A$, $B$, $C$ respectively.

Now, as in Proposition \ref{prop:drc.quad}, determine lines
$x_{12}$, $x_{13}$, $x_{23}$, $x_{123}$ such that they respectively complete triplets
$\{x,x_1,x_2\}$, $\{x,x_1,x_3\}$, $\{x,x_2,x_3\}$, $\{x_3,x_{13},x_{23}\}$
to double reflection configurations.\footnote{Let us note that in \cite{DragRadn2008}, the lines $x$, $x_1$, $x_2$, $x_3$, $x_{12}$, $x_{13}$, $x_{23}$, $x_{123}$ were respectively denoted by $\mathcal{O}$, $p$, $q$, $s$, $-x$, $p_1$, $q_1$, $x+s$, where the addition is defined in the billiard algebra introduced in that paper.}

Notice the following objects in Figure \ref{fig:zvezda}:
\begin{description}
\item[twelve verteces]
to each vertex, a plane tangent to one of the three quadrics 
$\mathcal{Q}_1$, $\mathcal{Q}_2$, $\mathcal{Q}_3$
and a pair of lines are assigned --- the lines of any pair are reflected to each other off the quadric at the touching point with the assigned plane;

\item[eight triangles]
in any triangle, the planes assigned to the verteces are touching the corresponding quadrics at three collinear points --- thus to each triangle, the line containing these points is naturaly assigned;

\item[six edges]
each edge containes four verteces --- four planes assigned to these verteces are in the same pencil; thus a double reflection configuration corresponds to each edge.
\end{description}

Now, we are ready to prove the $3D$-consistency of the quad-relation introduced via double reflection configurations.
The meaning of the following theorem is that reflections on three quadrics commute.

\begin{theorem}\label{th:3d}
Let $x$, $x_1$, $x_2$, $x_3$ be lines in the projective space, such that $x_1$, $x_2$, $x_3$ are obtained from $x$ by reflections off confocal quadrics
$\mathcal{Q}_1$, $\mathcal{Q}_2$, $\mathcal{Q}_3$ respectively.
Introduce lines $x_{12}$, $x_{13}$, $x_{23}$, $x_{123}$ such that the following quadruplets are double reflection configurations:
\begin{gather*}
\{x,x_1,x_{12},x_2\}, \quad
\{x,x_1,x_{13},x_3\}, \quad
\{x,x_2,x_{23},x_3\}, \quad
\{x_1,x_{12},x_{123},x_{13}\}.
\end{gather*} 
Then the following quadruplets are also double reflection configurations:
$$
\{x_2,x_{12},x_{123},x_{23}\},\quad
\{x_3,x_{13},x_{123},x_{23}\}.
$$ 
\end{theorem}
\begin{proof}
Let us remark that the configuration described in Theorem \ref{th:zvezda} has obviously a combinatorial structure of a cube, with planes corresponding to the edges of the cube.
In this way, lines $x$, $x_1$, $x_2$, $x_3$, $x_{12}$, $x_{13}$, $x_{23}$, $x_{123}$ will correspond to the vertices of the cube as shown in Figure \ref{fig:cube}.
A pair of lines is represented by endpoints of an edge if they reflect to each other off the plane joined to this edge.
Faces of the cube represent double reflection configurations.
Notice also that planes joined to parallel edges of the cube are tangent to the same quadric.
The statement follows from Theorem \ref{th:zvezda} and the construction given after, see
Figure \ref{fig:zvezda}.
\end{proof}

\subsection{Double reflection nets}
\label{sec:nets}
Assume a family of confocal quadrics is given in $\mathbf{P}^d$.
Notice that, by the Chasles theorem \cite{Chasles}, every line in $\mathbf{P}^d$ touches $d-1$ quadrics from the family.

Moreover, by Corollary \ref{cor:refl}, these $d-1$ quadrics are preserved by the billiard reflection.
Confocal quadrics touched by a line are called \emph{caustics} of this line, or consequently, caustics of the billiard trajectory that contains the line.

Now, fix $d-1$ quadrics from the pencil an take $\mathcal{A}\subset\mathcal{L}^d$ to be the set of all lines touching these $d-1$ quadrics.

\begin{definition}\label{def:dr-net}
\emph{A double reflection net} is a map 
\begin{equation}\label{eq:drnet}
\varphi\ :\ \mathbf{Z}^m \to \mathcal{A},
\end{equation}
such that there exist $m$ quadrics $\mathcal{Q}_1$, \dots, $\mathcal{Q}_m$ from the confocal pencil, satisfying the following conditions:
\begin{enumerate}
\item
sequence $\{\varphi(\mathbf{n}_0+i\mathbf{e}_j)\}_{i\in\mathbf{Z}}$ represents a billiard trajectory within $\mathcal{Q}_j$, for each $j\in\{1,\dots,m\}$ and 
$\mathbf{n}_0\in\mathbf{Z}^m$;

\item
lines $\varphi(\mathbf{n}_0)$, $\varphi(\mathbf{n}_0+\mathbf{e}_i)$,
$\varphi(\mathbf{n}_0+\mathbf{e}_j)$, $\varphi(\mathbf{n}_0+\mathbf{e}_i+\mathbf{e}_j)$ form a double reflection configuration, for all $i,j\in\{1,\dots,m\}$, $i\neq j$ and $\mathbf{n}_0\in\mathbf{Z}^m$.
\end{enumerate}
\end{definition}

In other words, for each edge in $\mathbf{Z}^m$ of direction $\mathbf{e}_i$,
the lines corresponding to its vertices intersect at $\mathcal{Q}_i$, while
the four tangent planes at the intersection points, associated to an elementary quadrilateral, belong to a pencil.

In the following subsections, we describe some examples of double reflection nets.
After that, we construct $F$-transformations of double reflection nets and conclude this section by establishing connection with the Grassmanian Darboux nets from \cite{ABS2009}.

\subsubsection*{Example of a double reflection net in the Minkowski space}\label{ex:hip-tropic-light}
Consider three-dimensional Minkowski space $\mathbf{E}^{2,1}$.
In this space, let a general confocal family is given by (\ref{eq:confocal.quadrics3}).


Fix $\lambda_0\in(b,a)$, and consider hyperboloid $\mathcal{Q}_{\lambda_0}$.
Denote by $a_1$, $a_2$, $a_3$, $a_4$, $b_1$, $b_2$, $b_3$, $b_4$ the light-like generatrices of $\mathcal{Q}_{\lambda_0}$, in the following way (see Figure \ref{fig:tropic.tangents}):
\begin{figure}[h]
\centering
\input{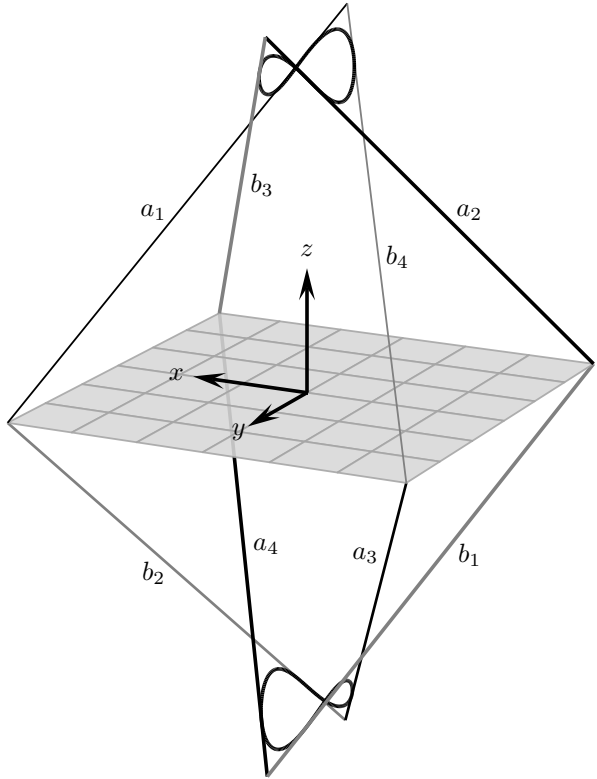}
\caption{The tropic curves of $\mathcal{Q}_{\lambda_0}$ and their light-like tangents.}\label{fig:tropic.tangents}
\end{figure}
\begin{itemize}
\item
lines $a_i$ belong to one, and $b_i$ to the other family of generatrices of $\mathcal{Q}_{\lambda_0}$; that is, $a_i$ and $a_j$ are always skew for $i\neq j$, while $a_i$ and $b_j$ are coplanar for all $i,j$;

\item
$a_1$, $a_2$, $b_3$, $b_4$ are tangent to the tropic curve contained in the half-space $z>0$, while $a_3$, $a_4$, $b_1$, $b_2$ are touching the other tropic curve;

\item
$a_i$ is parallel to $b_i$ for each $i$;

\item
pairs $(a_1,b_2)$, $(a_2,b_1)$, $(a_3,b_4)$, $(a_4,b_3)$ have intersection points in the $xy$-plane;

\item
pairs $(a_1,b_3)$, $(a_2,b_4)$, $(a_3,b_1)$, $(a_4,b_2)$ have intersection points in the $xz$-plane;

\item
pairs $(a_1,b_4)$, $(a_2,b_3)$, $(a_3,b_2)$, $(a_4,b_1)$ have intersection points in the $yz$-plane.
\end{itemize}

Take $\mathcal{A}$ to be the set of all generatrices of hyperboloid 
$\mathcal{Q}_{\lambda_0}$, i.e.~ the set of all lines having $\mathcal{Q}_{\lambda_0}$ as a double caustic.
In particular, $\mathcal{A}$ contains all lines $a_i$, $b_i$.

It is possible to define a map
$$
\varphi_M\ :\ \mathbf{Z}^4\to\mathcal{A},
$$
such that the image of $\varphi_M$ is the set $\{a_1,a_2,a_3,a_4,b_1,b_2,b_3,b_4\}$ and
for each $\mathbf{n}\in\mathbf{Z}^4$ lines $\varphi_{M}(\mathbf{n}+\mathbf{e}_1)$,
$\varphi_{M}(\mathbf{n}+\mathbf{e}_2)$, $\varphi_{M}(\mathbf{n}+\mathbf{e}_3)$,
$\varphi_{M}(\mathbf{n}+\mathbf{e}_4)$ are obtained from $\varphi_{M}(\mathbf{n})$ by reflection off $\mathcal{Q}_{a}$, $\mathcal{Q}_{b}$, $\mathcal{Q}_{-c}$, $\mathcal{Q}_{\infty}$ respectively.

More precisely, $\varphi_M$ will be periodic with period $2$ in each coordinate and:
\begin{gather*}
\varphi_M(0,0,0,0)=\varphi_M(1,1,1,1)=a_1,\quad \varphi_M(1,1,0,0)=\varphi_M(0,0,1,1)=a_2,
\\
\varphi_M(1,0,1,0)=\varphi_M(0,1,0,1)=a_3,\quad \varphi_M(0,1,1,0)=\varphi_M(1,0,0,1)=a_4,
\\
\varphi_M(0,0,0,1)=\varphi_M(1,1,1,0)=b_1,\quad \varphi_M(1,1,0,1)=\varphi_M(0,0,1,0)=b_2,
\\
\varphi_M(1,0,1,1)=\varphi_M(0,1,0,0)=a_3,\quad \varphi_M(0,1,1,1)=\varphi_M(1,0,0,0)=b_4,
\end{gather*}

It is shown in Figure \ref{fig:hiper-kocka} how vertices of the unit tesseract in $\mathbf{Z}^4$ are mapped by $\varphi_M$.

\begin{figure}[h]
\centering
\input{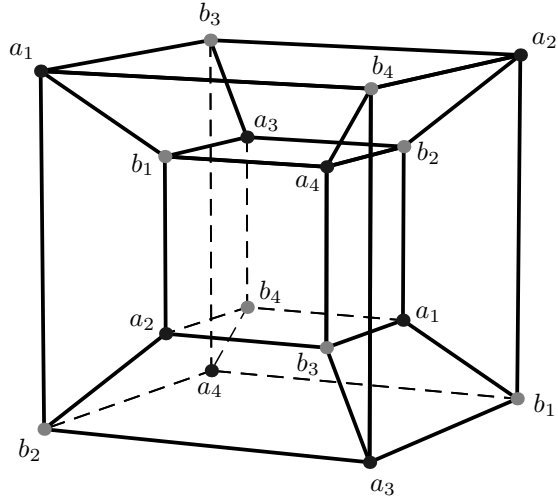}
\caption{Mapping $\varphi_M$ on the unit tesseract.}\label{fig:hiper-kocka}
\end{figure}

It is straightforward to prove the following
\begin{proposition}
$\varphi_M$ is a double reflection net.
\end{proposition}

\subsubsection*{Poncelet-Darboux grids and double reflection nets}
Let $\mathcal{E}$ be an ellipse in the Euclidean plane:
$$
\mathcal{E}\ :\ \frac{x^2}{a}+\frac{y^2}{b},\quad a>b>0,
$$
and $(a_i)_{i\in\mathbf{Z}}$ a billiard trajectory within $\mathcal{E}$.

As it is well known, all lines $a_i$ are touching the same conic $\mathcal{C}$ confocal with $\mathcal{E}$.
Here, we will additionally suppose that $\mathcal{C}$ is an ellipse.
Denote by $\mathcal{A}$ the set of tangents of $\mathcal{C}$.

Fix $m$ positive integers $k_1$, \dots, $k_{m}$ and define the mapping:
$$
\varphi_D\ :\ \mathbf{Z}^{m}\to\mathcal{A},
\quad
\varphi_D(n_1,\dots,n_m)=a_{n_1k_1+\dots+n_mk_{m}}.
$$

\begin{proposition}\label{prop:phiD}
Map $\varphi_D$ is a double reflection net.
\end{proposition}

\begin{proof}
Since $\varphi_D(\mathbf{n}+i\mathbf{e}_j)=a_{n_1k_1+\dots+n_mk_{m}+ik_j}$,
($\mathbf{n}=(n_1,\dots,n_m)$), it follows by \cite[Theorem 18]{DragRadn2011book} that sequence $(\varphi_D(\mathbf{n}+i\mathbf{e}_j))_{i\in\mathbf{Z}}$ represents a billiard trajectory within some ellipse $\mathcal{E}_j$, confocal with $\mathcal{E}$ and $\mathcal{C}$.

Immediately, by Definition \ref{def:VRC}, lines $\varphi(\mathbf{n}_0)$, $\varphi(\mathbf{n}_0+\mathbf{e}_i)$, $\varphi(\mathbf{n}_0+\mathbf{e}_j)$, $\varphi(\mathbf{n}_0+\mathbf{e}_i+\mathbf{e}_j)$ form a virtual reflection configuration for each $\mathbf{n}_0\in\mathbf{Z}^m$, $i,j\in\{1,\dots,m\}$.

Moreover, by Propositon \ref{prop:drc-ellipsoids}, they also form a double reflection configuration.
\end{proof}

\begin{remark}
It is interesting to consider only nets where $m$ ellipses $\mathcal{E}_j$ appearing in the proof of Proposition \ref{prop:phiD} are distinct.
If some of them coincide, then we may consider a corresponding subnet.

Suppose that $(a_i)$ is a non-periodic trajectory.
Then, choosing any $m$, and any set of distinct positive numbers $k_1$, \dots, $k_m$, we get substantially different double reflection nets.

For $(a_i)$ being $n$-periodic, it is enough to consider the case $k_i=i$, $i\in\{1,\dots,[n/2]\}$, $(m=[n/2])$.
\end{remark}

\begin{example}
Suppose $(a_i)$ is a $5$-perodic billiard trajectory within $\mathcal{E}$, see Figure \ref{fig:petougao}.
\begin{figure}[h]
\centering
\input{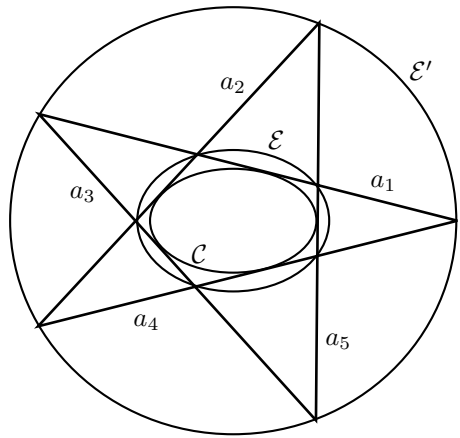}
\caption{A Poncelet pentagon.}\label{fig:petougao}
\end{figure}

The corresponding double reflection net is:
$$
\varphi_D\ :\ \mathbf{Z}^{2}\to\mathcal{A},
\quad
\varphi_D(n_1,n_2)=a_{n_1+2n_2}.
$$
\end{example}

\subsubsection*{$s$-skew lines and double reflection nets}
\label{sec:skew}
Now, let us consider a family of confocal quadrics in $\mathbf{E}^d$ ($d\ge3$) and fix its $d-1$ quadrics.
As usually, $\mathcal{A}$ is the set of all lines tangent to the fixed quadrics.

It is shown in \cite{DragRadn2008} that, from a line in $\mathcal{A}$, we can obtain any other line from that set in at most $d-1$ reflections on quadrics from the confocal family.
We called lines $a$, $b$ from $\mathcal{A}$ \emph{$s$-skew} if $s$ is the smallest number such that they can be obtained by $s+1$ such reflections.

Now, suppose lines $a$, $b$ are $s$-skew ($s\ge1$), and let
$\mathcal{Q}_1$, \dots, $\mathcal{Q}_{s+1}$
be the corresponding quadrics from the confocal family.

\begin{theorem}
There is a unique double reflection net
$$
\varphi_s\ :\ \mathbf{Z}^{s+1}\to\mathcal{A}
$$
which satisfies the following:
\begin{itemize}
\item
$\varphi_s(0,\dots,0)=a$;

\item
$\varphi_s(1,\dots,1)=b$;

\item
$\{\varphi(\mathbf{n}_0+i\mathbf{e}_j)\}_{i\in\mathbf{Z}}$ represents a billiard trajectory within $\mathcal{Q}_j$, for each $j\in\{1,\dots,s+1\}$ and $\mathbf{n}_0\in\mathbf{Z}^{s+1}$.
\end{itemize}
\end{theorem}

\begin{proof}
First, we are going to define mapping $\varphi_s$ on $\{0,1\}^{s+1}$.

For a permutation $\mathbf{p}=(p_1,\dots,p_{s+1})$ of the set $\{1,\dots,s+1\}$, we take a sequence of lines $(\ell_0^{\mathbf{p}},\dots,\ell_{s+1}^{\mathbf{p}})$ such that $\ell_0^{\mathbf{p}}=a$, $\ell_s^{\mathbf{p}}=b$, and $\ell_{i-1}^{\mathbf{p}}$, $\ell_{i}^{\mathbf{p}}$ satisfy the reflection law off $\mathcal{Q}_{p_i}$ for each $i\in\{1,\dots,s+1\}$.
Such a sequence exists and it is unique.
Moreover, if $k\in\{1,\dots,s+1\}$ is given, and permutations $\mathbf{p}$, $\mathbf{p}'$ coincide in the first $k$ coordinates, then $\ell_{i}^{\mathbf{p}}=\ell_{i}^{\mathbf{p}'}$ for $i\le k$.
Take $\{i_1,\dots,i_k\}$ to be a subset of $\{1,\dots,s+1\}$, and $\mathbf{p}$ any permutation of set $\{1,\dots,s+1\}$ with $p_1=i_1$, \dots, $p_k=i_k$.
We define:
$$
\varphi_s(\chi(1),\dots,\chi(s+1))=\ell_{k}^{\mathbf{p}},
$$
where $\chi=\chi_{\{i_1,\dots,i_k\}}$ is the corresponding characteristic function on 
$\{1,\dots,s+1\}$:
$$
\chi\ :\ \{1,\dots,s+1\}\to\{0,1\},
\quad
\chi(j)=
\begin{cases}
1, &j\in\{i_1,\dots,i_k\};\\
0, &j\not\in\{i_1,\dots,i_k\}.
\end{cases}
$$
In this way, we constructed $\varphi_s$ on $\{0,1\}^{s+1}$.

Subsequently, $\varphi_s$ can be extended to the rest of $\mathbf{Z}^{s+1}$, so that 
$\{\varphi_s(\mathbf{n}_0+i\mathbf{e}_j)\}_{i\in\mathbf{Z}}$ 
will represent billiard trajectories within $\mathcal{Q}_j$.

This construction is correct and unique due to Theorem \ref{th:3d}.
\end{proof}

\subsubsection*{Construction of double reflection nets}
\label{sec:const}
Let $\mathcal{Q}_1$, \dots, $\mathcal{Q}_m$ be distinct quadrics belonging to a confocal family and $\ell$ a line in $\mathbf{P}^d$.
Let us choose lines $\ell_i$ satisfying with $\ell$ the reflection law off
$\mathcal{Q}_i$, $1\le i\le m$.

\begin{theorem}\label{th:const}
There is a unique double reflection net $\varphi\ :\ \mathbf{Z}^m\to\mathcal{A}_{\ell}$, with the following properties:
\begin{itemize}
\item
$\varphi(0,\dots,0)=\ell$;

\item
$\varphi(\mathbf{e}_i)=\ell_i$, for each $i\in\{1,\dots,m\}$.
\end{itemize}
By $\mathcal{A}_{\ell}$, we denoted the set of all lines in $\mathbf{P}^d$ touching the same $d-1$ quadrics from the confocal family as $\ell$.
\end{theorem}
\begin{proof}
First, we define $\varphi$ on $\{0,1\}^m$, from the condition that lines corresponding to each $2$-face of the unit cube need to form a double reflection configuration.
This construction is unique because of Proposition \ref{prop:drc.quad} and correct, due to the $3D$-consistency property proved in Theorem \ref{th:3d}.

At all other points of $\mathbf{Z}^m$, $\varphi$ is uniquely defined from the request that $\{\varphi(\mathbf{n}_0+i\mathbf{e}_j)\}_{i\in\mathbf{Z}}$ will be billiard trajectories within $\mathcal{Q}_j$.

Consistency of the construction follows again from Theorem \ref{th:3d}.
\end{proof}

\subsubsection*{Focal nets and F-transformations of double reflection nets}
\label{sec:F}
Let
$\varphi\ :\ \mathbf{Z}^m \to \mathcal{A}$
be a double reflection net.

For given $\mathbf{n}_0\in\mathbf{Z}^m$ and distinct indices $i,j,k\in\{1,\dots m\}$, consider the following points of its $i$-th focal net:
\begin{gather*}
F_{i}=F^{(i)}(\mathbf{n}_0)=\varphi(\mathbf{n}_0)\cap\varphi(\mathbf{n}_0+\mathbf{e}_i),
   \\
F_{ij}=F^{(i)}(\mathbf{n}_0+\mathbf{e}_j)=
\varphi(\mathbf{n}_0+\mathbf{e}_j)\cap\varphi(\mathbf{n}_0+\mathbf{e}_j+\mathbf{e}_i),
   \\
F_{ik}=F^{(i)}(\mathbf{n}_0+\mathbf{e}_k)=
\varphi(\mathbf{n}_0+\mathbf{e}_k)\cap\varphi(\mathbf{n}_0+\mathbf{e}_k+\mathbf{e}_i),
	\\
F_{ijk}=F^{(i)}(\mathbf{n}_0+\mathbf{e}_j+\mathbf{e}_k)=
\varphi(\mathbf{n}_0+\mathbf{e}_j+\mathbf{e}_k)\cap\varphi(\mathbf{n}_0+\mathbf{e}_j+\mathbf{e}_k+\mathbf{e}_i).
\end{gather*}

\begin{proposition}\label{prop:coplanar}
Points 
$F_i$, $F_{ij}$, $F_{ik}$, $F_{ijk}$ are coplanar.
\end{proposition}

\begin{proof}
This is a consequence of the theorem of focal nets from \cite{BS2008book}.
However, we will show the direct proof, from configurations considered in
Section \ref{sec:algebra.quad}.

The four points belong to quadric $\mathcal{Q}_{i}$.
The tangent planes to $\mathcal{Q}_{i}$ at these points, divided into two pairs, determine two pencils of planes.
According to Theorem \ref{th:zvezda}, the two pencils are coplanar; thus they intersect.
As a consequence, the lines of poles with respect to the quadric
$\mathcal{Q}_{i}$, which correspond to these two pencils of planes, also intersect.
It follows that the four points are coplanar.
\end{proof}

We are going to construct an $F$-transformation of the double reflection net.

First, we select a quadric $\mathcal{Q}_{\delta}$ from the confocal family and introduce line $\ell'$ which satisfies with $\varphi(\mathbf{n}_0)$ the reflection law on
$\mathcal{Q}_{\delta}$.

By Theorem \ref{th:const}, it is possible to construct a double reflection net
$\bar{\varphi}\ :\ \mathbf{Z}^{m+1}\to\mathcal{A}$,
such that:
\begin{itemize}
\item
$\bar{\varphi}(\mathbf{n},0)=\varphi(\mathbf{n})$;

\item
$\bar{\varphi}(\mathbf{n}_0,1)=\ell'$.
\end{itemize}

Now, we define:
$$
\varphi^{+}\ :\  \mathbf{Z}^m\to\mathcal{A}_{\ell},
\quad
\varphi^{+}(\mathbf{n})=\bar{\varphi}(\mathbf{n},1).
$$

\begin{proposition}
Map $\varphi^{+}$ is an $F$-transformation of $\varphi$. 
\end{proposition}

\begin{proof}
Lines $\varphi^{+}(\mathbf{n})$ and $\varphi(\mathbf{n})$ intersect, since they satisfy the reflection law off $\mathcal{Q}_{\delta}$.
\end{proof}

\subsubsection*{Double reflection nets and Grassmannian Darboux nets}
\label{sec:grass}
Let us recall the definition of a Grassmannian Darboux net from \cite{ABS2009}:
a map from the edges of a regular square lattice $\mathbf{Z}^m$ to the Grassmannian
$\mathbf{G}^d_r$ of $r$-dimensional projective subspaces of the $d$-dimensional projective space is \emph{a Grassmanian Darboux net} if the four $r$-spaces of an elementary quadrilateral belong to a $(2r+1)$-space.
For $r=0$, the ordinary Darboux nets from \cite{Schief2003} are obtained, where the four points of intersection associated to a quadrilateral, belong to a line.

Now, consider a general double reflection net (\ref{eq:drnet}).

To each edge $(\mathbf{n_0},\mathbf{n_0}+\mathbf{e}_i)$ of $\mathbf{Z}^m$, we can associate the plane which is tangent to $\mathcal{Q}_i$ at point 
$\varphi(\mathbf{n_0})\cap\varphi(\mathbf{n_0}+\mathbf{e}_i)$.
Since the lines corresponding to the vertices of a face form a double reflection configuration, the four planes associated to the edges belong to a pencil.

In this way, we see that a double reflection net induces a map:
$$
E(\mathbf{Z}^m)\rightarrow \mathbf{G}^d_{d-1},
$$
where $E(\mathbf{Z}^m)$ is the set of all edges of the integer lattice $\mathbf{Z}^m$.

In this way, double reflection nets induce a subclass of dual Darboux nets.

It was shown in \cite{Schief2003} how to associate discrete integrable hierarchies to the Darboux nets.

\subsection{Yang-Baxter map}
\label{sec:yb}
\emph{A Yang-Baxter map} is a map
$R:\mathcal{X}\times\mathcal{X}\to\mathcal{X}\times\mathcal{X}$,
satisfying the Yang-Baxter equation:
$$
R_{23}\circ R_{13}\circ R_{12}=R_{12}\circ R_{13}\circ R_{23},
$$
where $R_{ij}:\mathcal{X}\times\mathcal{X}\times\mathcal{X}\to\mathcal{X}\times\mathcal{X}\times\mathcal{X}$
acts as $R$ on the $i$-th and $j$-th factor in the product, and as the identity on the remaining one, see \cite{ABS2004} and references therein.

Here, we are going to construct an example of Yang-Baxter map associated to confocal families of quadrics.
To begin, we fix a family of confocal quadrics in $\mathbf{CP}^{n}$:
\begin{equation}\label{eq:confocal.family}
\mathcal{Q}_{\lambda}\ :\ 
\frac{z_1^2}{a_1-\lambda}+\dots+\frac{z_d^2}{a_d-\lambda}
=
z_{n+1}^2,
\end{equation}
where $a_1$, \dots, $a_d$ are constants in $\mathbf{C}$, and $[z_1:z_2:\dots:z_{n+1}]$ are homogeneous coordinates in $\mathbf{CP}^{n}$.

Take $\mathcal{X}$ to be the space $\mathbf{CP}^{n*}$ dual to the $n$-dimensional projective space, i.e.\ the variety of all hyper-planes in $\mathbf{CP}^{n}$.
Note that a general hyper-plane in the space is tangent to exactly one quadric from family (\ref{eq:confocal.family}).
Besides, in a general pencil of hyper-planes, there are exactly two of them tangent to a fixed general quadric.

Now, consider a pair $x$, $y$ of hyper-planes.
They are touching respectively unique quadrics $\mathcal{Q}_{\alpha}$, $\mathcal{Q}_{\beta}$ from (\ref{eq:confocal.family}).
Besides, these two hyper-planes determine a pencil of hyper-planes.
This pencil contains unique hyper-planes $x'$, $y'$, other than $x$, $y$, that are tangent to $\mathcal{Q}_{\alpha}$, $\mathcal{Q}_{\beta}$ respectively.

We define
$R : \mathbf{CP}^{n*}\times\mathbf{CP}^{n*} \to \mathbf{CP}^{n*}\times\mathbf{CP}^{n*}$, in such a way that $R(x,y)=(x',y')$ if $(x',y')$ are obtained from $(x,y)$ as just described.

Maps
$$
R_{12},\ R_{13},\ R_{23}\ :\
\mathbf{CP}^{n*}\times\mathbf{CP}^{n*}\times\mathbf{CP}^{n*}
\to
\mathbf{CP}^{n*}\times\mathbf{CP}^{n*}\times\mathbf{CP}^{n*}
$$
are then defined as follows:
\begin{align*}
&R_{12}(x,y,z)=(x',y',z)\quad\text{for}\quad (x',y')=R(x,y);\\
&R_{13}(x,y,z)=(x',y,z')\quad\text{for}\quad (x',z')=R(x,z);\\
&R_{23}(x,y,z)=(x,y',z')\quad\text{for}\quad (y',z')=R(y,z).
\end{align*}

To prove the Yang-Baxter equation for map $R$, we will need the following

\begin{lemma}\label{lemma:YBE}
Let $\mathcal{Q}_{\alpha}$, $\mathcal{Q}_{\beta}$, $\mathcal{Q}_{\gamma}$
be three non-degenerate quadrics from family (\ref{eq:confocal.family}) and
$x$, $y$, $z$ respectively their tangent hyper-planes.
Take:
$$
(x_2,y_1)=R(x,y),\quad
(x_3,z_1)=R(x,z),\quad
(y_3,z_2)=R(y,z).
$$
Let $x_{23}$, $y_{13}$, $z_{12}$ be the joint hyper-planes of pencils determined by pairs $(x_3,y_3)$ and $(x_2,z_2)$, $(x_3,y_3)$ and $(y_1,z_1)$, $(y_1,z_1)$ and $(x_2,z_2)$ respectively.

Then $x_{23}$, $y_{13}$, $z_{12}$ touch quadrics $\mathcal{Q}_{\alpha}$, $\mathcal{Q}_{\beta}$, $\mathcal{Q}_{\gamma}$ respectively.
\end{lemma}

\begin{proof}
This statement, formulated for the dual space in dimension $n=2$ is proved as \cite[Theorem 5]{ABS2004}.

Consider the dual situation in an arbitrary dimension $n$.
The dual quadrics $\mathcal{Q}_{\alpha}^*$, $\mathcal{Q}_{\beta}^*$, $\mathcal{Q}_{\gamma}^*$ belong to a linear pencil, and points $x^*$, $y^*$, $z^*$, dual to hyper-planes $x$, $y$, $z$, are respectively placed on these quadrics.
Take the two-dimensional plane containing these three points.
The intersection of the pencil of quadrics with this, and any other plane as well, represents a pencil of conics.
Thus, Theorem 5 from \cite{ABS2004} will remain true in any dimension.

This lemma is dual to this statement, thus the proof is complete. 
\end{proof}

\begin{theorem}\label{th:YBE}
Map $R$ satisfies the Yang-Baxter equation.
\end{theorem}

\begin{proof}
Let $x$, $y$, $z$ be hyper-planes in $\mathbf{CP}^{n}$.
We want to prove that
$$
R_{23}\circ R_{13}\circ R_{12}(x,y,z)=R_{12}\circ R_{13}\circ R_{23}(x,y,z).
$$

Denote by $\mathcal{Q}_{\alpha}$, $\mathcal{Q}_{\beta}$, $\mathcal{Q}_{\gamma}$ the quadrics from (\ref{eq:confocal.family}) touching $x$, $y$, $z$ respectively.

Let:
\begin{align*}
&(x,y,z)
 \xrightarrow{R_{12}}
(x_2,y_1,z)
 \xrightarrow{R_{13}}
(x_{23},y_1,z_1)
 \xrightarrow{R_{23}}
(x_{23},y_{13},z_{12}),
 \\
&(x,y,z)
 \xrightarrow{R_{23}}
(x,y_3,z_2)
 \xrightarrow{R_{13}}
(x_3,y_3,z_{12}')
 \xrightarrow{R_{12}}
(x_{23}',y_{13}',z_{12}').
\end{align*}
Now, apply Lemma \ref{lemma:YBE} to hyper-planes $x$, $y$, $z_2$.
Since:
$$
(x_2,y_1)=R(x,y),\quad
(x_3,z_{12}')=R(x,z_2),\quad
(y_3,z)=R(y,z_2),
$$
we have that the joint hyper-plane of pencils $(x_3,y_3)$ and $(x_2,z)$ is touching $\mathcal{Q}_{\alpha}$ -- therefore, this plane must coincide with $x_{23}$ and $x_{23}'$, i.e.~$x_{23}=x_{23}'$.
Also, the joint hyper-plane of pencils $(y_1,z_{12}')$ and $(x_2,z)$ is touching $\mathcal{Q}_{\gamma}$ -- therefore, this is $z_1$ and $z_{12}=z_{12}'$.
Finally, the joint hyper-plane of pencils $(x_3,y_3)$ and $(y_1,z_{12}')$ is tangent to $\mathcal{Q}_{\beta}$ -- it follows this is $y_{13}=y_{13}'$, which completes the proof.
\end{proof}

\begin{remark}
Instead of defining $R$ to act on the whole space $\mathbf{CP}^{n*}\times\mathbf{CP}^{n*}$, we can restrict it to the product of two non-degenerate quadrics from (\ref{eq:confocal.family}), namely:
$$
R(\alpha,\beta)\ :\ \mathcal{Q}_{\alpha}^*\times\mathcal{Q}_{\beta}^*\to\mathcal{Q}_{\alpha}^*\times\mathcal{Q}_{\beta}^*,
$$
where pair $(x,y)$ of tangent hyper-planes is mapped into pair $(x_1,y_1)$ in such a way that $x$, $y$, $x_1$, $y_1$ belong to the same pencil.

The corresponding Yang-Baxter equation is:
$$
R_{23}(\beta,\gamma)\circ R_{13}(\alpha,\gamma)\circ R_{12}(\alpha,\beta)=
R_{12}(\alpha,\beta)\circ R_{13}(\alpha,\gamma)\circ R_{23}(\alpha,\beta),
$$
where both sides of the equation represent maps from
$\mathcal{Q}_{\alpha}^*\times\mathcal{Q}_{\beta}^*\times\mathcal{Q}_{\gamma}^*$ to itself.
\end{remark}

In \cite{ABS2004}, for irreducible algebraic varieties
$\mathcal{X}_1$ and $\mathcal{X}_2$,
\emph{a quadrirational mapping}
$F\ :\ \mathcal{X}_1\times\mathcal{X}_2$
is defined.
For such a map $F$ and any fixed pair 
$(x,y)\in\mathcal{X}_1\times\mathcal{X}_2$,
except from some closed subvarieties of codimension at least $1$,
the graph
$\Gamma_F \subset
\mathcal{X}_1\times\mathcal{X}_2\times\mathcal{X}_1\times\mathcal{X}_2$
intersects each of the sets
$\{x\}\times\{y\}\times\mathcal{X}_1\times\mathcal{X}_2$,
$\mathcal{X}_1\times\mathcal{X}_2\times\{x\}\times\{y\}$,
$\mathcal{X}_1\times\{y\}\times\{x\}\times\mathcal{X}_2$,
$\{x\}\times\mathcal{X}_2\times\mathcal{X}_1\times\{y\}$
exactly at one point (see \cite[Definition 3]{ABS2004}).
In other words, $\Gamma_{F}$ is the graph of four rational maps:
$F$, $F^{-1}$, $\bar{F}$, $\bar{F}^{-1}$.

The following Proposition is a generalization of \cite[Proposition 4]{ABS2004}.

\begin{proposition}
Map
$
R(\alpha,\beta)\ :\ \mathcal{Q}_{\alpha}^*\times\mathcal{Q}_{\beta}^*\to\mathcal{Q}_{\alpha}^*\times\mathcal{Q}_{\beta}^*,
$
is quadrirational.
It is an involution and it concides with its companion $\bar{R}(\alpha,\beta)$.
\end{proposition}


\section{Pseudo-integrable billiards and local Poncelet theorem}
\label{sec:nosonja}

\subsection{Billiards in domains bounded by a few confocal conics}
\label{sec:few.confocal}
In this section we analyze billiard dynamics in a domain bounded by arcs of a few confocal conics, such that there are reflex angles on the boundary.
In order to describe some phenomena appearing in such systems, let us consider the domain $\mathcal{D}_0$ bounded by two confocal ellipses from family 
(\ref{eq:confocal-in-plane}) and two segments placed on the smaller axis of theirs, as shown in Figure \ref{fig:billiard1}.
\begin{figure}[h]
\input{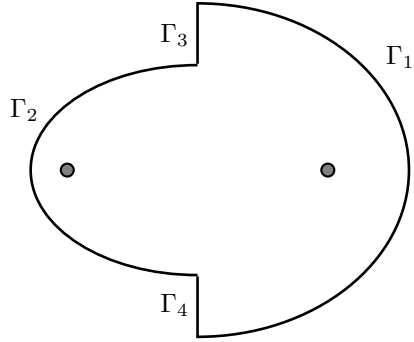}
\caption{Domain bounded by two confocal ellipses and two
segments on the $y$-axis.}\label{fig:billiard1}
\end{figure}
More precisely, we fix parameters $\beta_1$, $\beta_2$ such that
$\beta_1<\beta_2<b$,
and take the border of $\mathcal{D}_0$ to be:
\begin{gather*}
\partial\mathcal{D}_0=\Gamma_1\cup\Gamma_2\cup\Gamma_3\cup\Gamma_4,
\\
\Gamma_1=\{(x,y)\in\mathcal{C}_{\beta_1}\mid x\ge0\},\\
\Gamma_2=\{(x,y)\in\mathcal{C}_{\beta_2}\mid x\le0\},\\
\Gamma_3=\{(0,y)\mid\sqrt{b-\beta_2}\le y\le\sqrt{b-\beta_1}\},\\
\Gamma_4=\{(0,y)\mid-\sqrt{b-\beta_1}\le y\le-\sqrt{b-\beta_2}\}.
\end{gather*}
Notice that segments $\Gamma_3$, $\Gamma_4$ are lying on the the degenerate conic 
$\mathcal{C}_a$ of family (\ref{eq:confocal-in-plane}).

By Chasles'theorem \cite{Chasles}, as we have already mentioned in Section \ref{sec:billiards.confocal}, each line in the plane is touching exactly one conic from a given confocal family -- moreover, this conic remains the same after the reflection on any conic from the family.
Thus, each billiard trajectory in a domain bounded by arcs of several confocal conics has a caustic from the confocal family.

Consider billiard trajectories within domain $\mathcal{D}_0$ whose caustic is an ellipse 
$\mathcal{C}_{\alpha_0}$
completely placed inside the billiard table, i.e.~ $\beta_2<\alpha_0<b$.
An example of such a trajectory is shown in Figure \ref{fig:caustic1}.
\begin{figure}[h]
\input{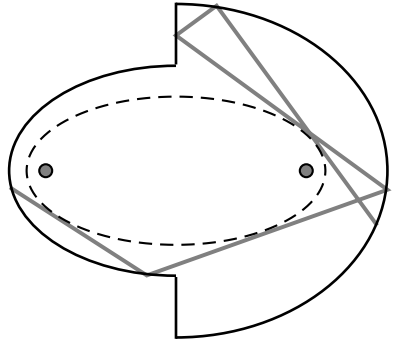}
\caption{A billiard trajectory in $\mathcal{D}_0$ with an ellipse as caustic.}\label{fig:caustic1}
\end{figure}

Such billiard trajectories fill out the ring $\mathcal R$ placed
between the billiard border and the caustic, see Figure
\ref{fig:ring1}.

\begin{figure}[h]
\input{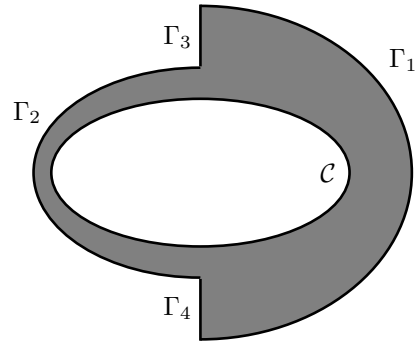}
\caption{Ring $\mathcal R$.}\label{fig:ring1}
\end{figure}

Let us examine the leaf of the phase space composed by these
trajectories.
This leaf is naturally decomposed into four rings
equal to $\mathcal R$, which are glued with each other along the
border segments. Let us describe this in detail:

\begin{itemize}
\item[$\mathcal R_1$]
This ring contains the points in the phase space that correspond
the billiard particle moving away from the caustic and the
clockwise direction around the ellipses center.

\item[$\mathcal R_2$]
Corresponds to the motion away from the caustic in the
counterclockwise direction.

\item[$\mathcal R_3$]
Corresponds to the motion towards the caustic in the
counterclockwise direction.

\item[$\mathcal R_4$]
Corresponds to the motion towards the caustic in the clockwise
direction.
\end{itemize}

Let us notice that the reflection off the two ellipse arcs
contained in the billiard boundary changes the direction of the
particle motion with respect to the caustic, but preserves the
direction of the motion around the foci. The same holds for
passing though tangency points with the caustic. On the other
hand, reflection on the axis changes the direction of motion
around the foci, but preserves the direction with respect to the
caustic. Thus, the four rings are connected to each other
according to the following scheme:
$$
 \xymatrix{
 &
\mathcal{R}_2 \ar@{-}[ld]_{\Gamma_3\Gamma_4} \ar@{-}[rd]^{\Gamma_1\Gamma_2\mathcal{C}}
 &
 \\
\mathcal{R}_1\ar@{-}[rd]_{\Gamma_1\Gamma_2\mathcal{C}}
 &
 &
\mathcal{R}_3
 \\
 &
\mathcal{R}_4 \ar@{-}[ru]_{\Gamma_3\Gamma_4} 
}
$$

Let us represent all the rings in Figures
\ref{fig:ring2} and \ref{fig:ring3}.

\begin{figure}[h]
\input{0-figures/ring2}
\caption{Rings $\mathcal R_1$, $\mathcal R_2$, $\mathcal R_3$,
$\mathcal R_4$.}\label{fig:ring2}
\end{figure}

\begin{figure}[h]
\input{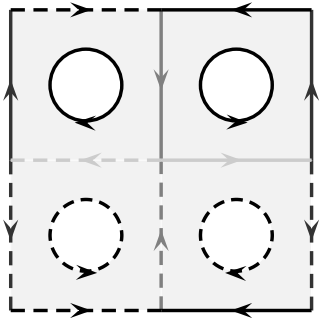}
\caption{Gluing rings $\mathcal R_1$, $\mathcal R_2$, $\mathcal
R_3$, $\mathcal R_4$.}\label{fig:ring3}
\end{figure}

Now, we have the following

\begin{proposition}\label{prop:genus3}
All billiard trajectories within domain $\mathcal{D}_0$ with a fixed elliptical caustic form an orientable surface of genus $3$.
\end{proposition}

In contrast, notice that the leaves for the billiard within an ellipse are tori, see Figures \ref{fig:ring-phase-ellipse} and \ref{fig:ring-phase-hyperbola} from Section \ref{sec:billiards}.


\subsection{Topological estimates}
\label{sec:topological}
Let $\mathcal{D}$ be a bounded domain in the plane such that its boundary $\Gamma=\partial\mathcal{D}$ is the union of finitely many arcs of confocal conics from the family (\ref{eq:confocal-in-plane}).

We consider the billiard system within $\mathcal{D}$.
Any trajectory of this billiard will have a caustic -- a conic from (\ref{eq:confocal-in-plane}) touching all lines containing segments of the trajectory.
Let us fix $\mathcal{C}_{\lambda_0}$ as caustic.

Notice that all tangent lines of a conic fill out an infinite domain in the plane: if the conic is an ellipse, the domain is its exterior; for a hyperbola, it is the part of the plane between its branches.

Denote by $\mathcal{D}_{\lambda_0}$ the intersection of $\mathcal{D}$ with the domain containing tangent lines of caustic $\mathcal{C}_{\lambda_0}$.
All billiard trajectories with caustic $\mathcal{C}_{\lambda_0}$ are placed in $\mathcal{D}_{\lambda_0}$.
$\mathcal{D}_{\lambda_0}$ is a bounded set whose boundary $\Gamma_{\lambda_0}=\partial\mathcal{D}_{\lambda_0}$
is the union of finitely many arcs of conics from (\ref{eq:confocal-in-plane}).
We assume that $\mathcal{D}_{\lambda_0}$ is connected as well, otherwise we consider its connected component.

All billiard trajectories in domain $\mathcal{D}$ with the caustic $\mathcal{C}_{\lambda_0}$ will correspond to a certain compact leaf $\mathcal{M}(\lambda_0)$ in the phase space.
$\mathcal{M}_{\lambda_0}$ is obtained by gluing four copies of $\mathcal{D}_{\lambda_0}$ along the corresponding arcs of the boundary $\Gamma_{\lambda_0}=\partial\mathcal{D}_{\lambda_0}$, similarly as it is explained in Section \ref{sec:few.confocal}.

On $\mathcal{M}_{\lambda_0}$, singular points of the billiard flow correspond to vertices of reflex angles on the boundary of $\mathcal{D}_{\lambda_0}$.
Since confocal conics are orthogonal to each other at points of intersection, only two types of such angles may appear: full angles and angles of $270^{\circ}$.
A vertex of a full angle is the projection of two singular points in the phase space.
Each of the singular points has four separatrices.
On the other hand, a vertex of a $270^{\circ}$ is a projection of only one singular point having six separatrices.

Using the Euler-Poincar\'e formula, as in \cite{Viana}, we get the following estimate for the total number $N=N(\mathcal{M}_{\lambda_0})$ of saddle-connections:

\begin{proposition}\label{prop:N(M)}
The total number $N=N(\mathcal{M}_{\lambda_0})$ of saddle-connections is bounded from above:
$$
N(\mathcal{M}_{\lambda_0})\le\frac{1}{2}\sum_{i=1}^ks_i=k-\chi(\mathcal{M}_{\lambda_0}),
$$
where $k$ is the number of singular points of the flow on $\mathcal{M}_{\lambda_0}$, and $s_1$, \dots, $s_k$ numbers of separatrices at each singular point.  
\end{proposition}

As a corollary, we get the following

\begin{proposition}
Consider billiard within $\mathcal{D}$ with $\mathcal{C}_{\lambda_0}$ as a caustic.
If the corresponding subdomain $\mathcal{D}_{\lambda_0}$ has $\tilde{k}$ reflex angles on its boundary $\Gamma_{\lambda_0}$ then: 
\begin{itemize}
\item $N\le 3\tilde{k}$;
\item $g(\mathcal{M}_{\lambda_0})=\tilde{k}+1$.
\end{itemize}
\end{proposition}

Notice that the genus of the surface $\mathcal{M}_{\lambda_0}$ depends only on the number of reflex angles on the boundary of $\mathcal{D}_{\lambda_0}$ and not of their types.
Also, $\tilde{k}\le k$.


\begin{example}
\begin{itemize}
\item
If there are no reflex angles on the boundary, i.e.~ $k=0$, then $\mathcal{M}_{\lambda_0}$ is a torus: $g=1$, $N=0$;

\item
if there is only one reflex angle on the boundary, independently if it is a $270^{\circ}$ angle or a full angle, we have that $g=2$.
\end{itemize}
\end{example}

We finish this section by formulating of an analogue of the Liuoville-Arnold theorem for pseudo-integrable billiard systems.
It is a consequence of the Maier theorem from the theory of measured foliations (see \cites{Maier1943,Viana}).
In our case, the measured foliation is defined by the kernel of an exact one-form $\beta_{\lambda}:=d K_{\lambda}$, where functions $K_{\lambda}$ are defined in Section \ref{sec:billiards.confocal}.

\begin{theorem}\label{th:maier}
There exist paiwise disjoint open domains $D_1$, \dots, $D_n$ on $\mathcal{M}_{\lambda_0}$, each of them being invariant under the billiard flow, such that their closures cover $\mathcal{M}_{\lambda_0}$ and for each $j\in\{1,\dots,n\}$:
\begin{itemize}
 \item
either $D_j$ consists of periodic billiard trajectories and is homeomorphic to a cylinder;
 \item
or $D_j$ consists of non-periodic trajectories all of which are dense in $D_j$.
\end{itemize}
The boundary of each $D_j$ consists of saddle-connections.
\end{theorem}

We see, that in contrast to completely integrable Hamiltonian systems, compact leaves of our billiards could be of a genus greater than $1$.
Moreover, one leaf could contain regions with periodic trajectories with different periods for different regions and simultaneousely could contain regions with non-periodic motion.
Because of that, we call such systems \emph{pseudo-integrable}, taking into account the fact that they possess two independent commuting first integrals, as it has been shown in Section \ref{sec:billiards.confocal}.

\subsection{Poncelet theorem and Cayley-type conditions}
\label{sec:poncelet.nosonja}
For billiards within confocal conics without reflex angles on the boundary, it is well known that the famous Poncelet porism holds (see \cites{DragRadn2004,DragRadn2011book}):
\begin{itemize}
\item[(A)] 
if there is a periodic billiard trajectory with one initial point of the boundary, then there are infinitely many such periodic trajectories with the same period, sharing the same caustic; 
\item[(B)]
even more is true, if there is one periodic trajectory, then all trajectories sharing the same caustic are periodic with the same period.
\end{itemize}

However, when reflex angles exist, which is the case studied in the present paper, one can say that (A) is still generally true.
However, (B) is not true any more.
In other words, \emph{the Poncelet porism is true locally, but not globally}.

\begin{theorem}\label{th:poncelet-nosonja}
There exist subsets $\delta_1$, \dots, $\delta_n$ of the boundary $\Gamma_{\lambda_0}$, with the following properties:
\begin{itemize}
 \item
$\delta_1$, \dots, $\delta_n$ are invariant under the billiard map; 

 \item 
$\delta_1$, \dots, $\delta_n$ are pairwise disjoint;
 
 \item
each $\delta_i$ is a finite union of $d_i$ open subarcs of $\Gamma_{\lambda_0}$: 
$$
\delta_i=\bigcup_{j=1}^{d_i}\ell_{j}^i;
$$

 \item
closure of $\delta_1\cup\dots\cup \delta_N$ is $\Gamma$,
\end{itemize}
such that they satisfy:
\begin{itemize}
 \item
if one billiard trajectory with bouncing points within $\delta_i$ is periodic, then all such trajectories are periodic with the same period $n_i$.
Moreover, $n_i$ is a multiple of $d_i$ and every such a trajectory bounces the same number $\dfrac{n_i}{d_i}$ of times off each arc $\ell_{j}^i$;

\item
if billiard trajectories having vertices in $\delta_i$ are non-periodic, then the bouncing points of each trajectory are dense in $\delta_i$.
\end{itemize}
The boundary of each $\delta_i$ consists of bouncing points of saddle-connections.
\end{theorem}

This theorem is a consequence of Theorem \ref{th:maier} from the previous section.
The proof follows from the fact that each of the domains $D_i$ intersects the boundary $\Gamma_{\lambda_0}$ and forms $\delta_i=\Gamma_{\lambda_0}\cap D_i$.


In \cite{DragRadn2004} conditions of Cayley's type for periodicity of billiards within several confocal quadrics in the Euclidean space of an arbitrary dimension were derived, see also \cite{DragRadn2006jms} for detailed examples.

We analyzed there billiards within domains bounded by arcs of several confocal quadrics and \emph{the billiad ordered game} within a few confocal ellipsoids.
Unlike in the present article, domains considered in \cites{DragRadn2004,DragRadn2006jms} did not contain reflex angles at the boundary.
However, the technique used there to describe periodic trajectories can be directly transferred to the present problems.

Before stating the Cayley-type conditions, recall that a point is being reflected off conic $\mathcal{C}_{\lambda_0}$ \emph{from outside} if the corresponding Jacobi elliptic coordinate achieves a local maximum at the reflection point, and \emph{from inside} if there the coordinate achieves a local minimum (see \cite{DragRadn2004}).

\begin{theorem}\label{th:cayley-nosonja}
Consider domain $\mathcal{D}$ bounded by arcs of $k$ ellipses $\mathcal{C}_{\beta_1}$, \dots, $\mathcal{C}_{\beta_k}$,  $l$ hyperbolas $\mathcal{C}_{\gamma_1}$, \dots, $\mathcal{C}_{\gamma_l}$, and several segments belonging to degenerate conics from the confocal family (\ref{eq:confocal-in-plane}):
$$
\beta_1,\dots,\beta_k\in(-\infty,b),\ k\ge1,\ 
\gamma_1,\dots,\gamma_l\in(b,a),\ l\ge0.
$$
Let $\mathcal{C}_{\alpha_0}$ be an ellipse contained within all ellipses $\mathcal{C}_{\beta_1}$, \dots, $\mathcal{C}_{\beta_k}$: $b>\alpha_0>\beta_i$ for all $i\in\{1,\dots,k\}$.
A necessary condition for the existence of a billiard trajectory within $\mathcal{D}$ with $\mathcal{C}_{\alpha_0}$ as a caustic which becomes closed after:
\begin{itemize}
\item 
$n_i'$ reflections from inside and $n_i''$ reflections from outside off $\mathcal{C}_{\beta_i}$, $1\le i\le k$;

\item
$m_j'$ reflections from inside and $m_j''$ reflections from outside off $\mathcal{C}_{\gamma_j}$, $1\le j\le l$;
 
\item
total number of $p$ intersections with the $x$-axis and reflections off the segments contained in the $x$-axis;

\item
total number of $q$ intersections with the $y$-axis and reflections off the segments contained in the $y$-axis;
\end{itemize}
is:
\begin{gather*}
\sum_{i=1}^k(n_i'-n_i'')(\mathcal{A}(P_{\beta_i})-\mathcal{A}(P_{\alpha_0})) +
\sum_{j=1}^l(m_j'-m_j'')\mathcal{A}(P_{\gamma_j}) +
p\mathcal{A}(P_a)-q\mathcal{A}(P_b)=0,
\\
m_j'-m_j''+p-q=0.
\end{gather*}
Here $\mathcal{A}$ is the Abel-Jacobi map of the ellitic curve:
$$
\Gamma\ :\ s^2=\mathcal{P}(t):=(a-t)(b-t)(\alpha_0-t),
$$
and $P_{\delta}$ denotes point $(\delta,\sqrt{\mathcal{P}(\delta)})$ on $\Gamma$.
\end{theorem}

\begin{proof}
Following Jacobi \cite{JacobiGW} and Darboux \cite{Darboux1870}, similarly as in \cite{DragRadn2004}, we consider sums
$$
\int\frac{d\lambda_1}{\sqrt{\mathcal{P}(\lambda_1)}}+\int\frac{d\lambda_2}{\sqrt{\mathcal{P}(\lambda_2)}}
\ \text{and}\ 
\int\frac{\lambda_1d\lambda_1}{\sqrt{\mathcal{P}(\lambda_1)}}+\int\frac{\lambda_2d\lambda_2}{\sqrt{\mathcal{P}(\lambda_2)}}
$$
over billiard trajectory $A_1\dots A_N$.
Here $(\lambda_1,\lambda_2)$ are Jacobi elliptic coordinates, $\lambda_1<\lambda_2$.
The second integral is equal to the length of the trajectory, while the first one is zero.

Notice that, along a trajectory, $\lambda_1$ achieves local extrema at points of reflection off ellipses and touching points with the caustic, and $\lambda_2$ at points of reflection off hyperbolas and intersection points with the coordinate axes, we obtain that $A_1=A_N$ is equivalent to the condition stated.
\end{proof}

We illustrate this theorem on the example when the billiard table is $\mathcal{D}_0$, as defined in Section \ref{sec:few.confocal}.

\begin{example}
A necessary condition for the existence of a billiard trajectory within $\mathcal{D}_0$ with $\mathcal{C}_{\alpha_0}$ as a caustic, such that it becomes closed after $n_1$ reflections off $\mathcal{C}_{\beta_1}$ and $n_2$ reflections off $\mathcal{C}_{\beta_2}$ is:
$$
n_1\mathcal{A}(P_{\beta_1})+n_2\mathcal{A}(P_{\beta_2})=(n_1+n_2)\mathcal{A}(P_{\alpha_0}).
$$
Notice that in this case number $p$ and $q$ are always even and equal to each other.
Since $2\mathcal{A}(P_a)=2\mathcal{A}(P_b)$, the corresponding summands are cancelled out.
\end{example}

\subsection{Interval exchange transformation}
\label{sec:interval}
In this section, we are going to establish a connection of the billiard dynamics within domain $\mathcal{D}_0$ defined in Section \ref{sec:few.confocal} with the inteval exchange transformation.

\subsubsection*{Interval exchange maps}
Let $I\subset\mathbf{R}$ be an interval, and $\{I_{\alpha}\mid\alpha\in\mathcal{A}\}$ its finite partition into subintervals.
Here $\mathcal{A}$ is a finite set of at least two elements.
We consider all intervals to be closed on the left and open on the right.

\emph{An interval exchange map} is a bijection of $I$ into itself, such that its restriction on each $I_{\alpha}$ is a translation.
Such a map $f$ is determined by the following data:
\begin{itemize}
 \item
a pair $(\pi_0,\pi_1)$ of bijections $\mathcal{A}\to\{1,\dots,d\}$ describing the order of the subintervals $\{I_{\alpha}\}$ in $I$ and $\{f(I_{\alpha})\}$ in $f(I)=I$.
We denote:
$$
\pi=\left(
\begin{array}{cccc}
\pi_0^{-1}(1) & \pi_0^{-1}(2) & \dots & \pi_0^{-1}(d)\\
\pi_1^{-1}(1) & \pi_1^{-1}(2) & \dots & \pi_1^{-1}(d)
\end{array}
\right).
$$

\item
a vector $\lambda=(\lambda_{\alpha})_{\alpha\in\mathcal{A}}$ of the lengths of $I_{\alpha}$.
\end{itemize}

\subsubsection*{Billiard dynamics}\label{sec:dynamics}
To each billiard trajectory, we join the sequence:
$$
\{(X_n,s_n)\},
\quad
X_n\in\mathcal{C}_{\alpha_0},
\quad
s_n\in\{+,-\}
$$
where $X_n$ are joint points of the trajectory with the caustic,
while $s_n=+$ if at $X_n$ the trajectory is winding counterclockwise and $s_n=-$ if it is winding clockwise about the caustic.

Introduce metric $\mu$ on the caustic $\mathcal{C}_{\alpha_0}$ as in Proposition \ref{prop:rotation}.
Then, we parametrize $\mathcal{C}_{\alpha_0}$ by parameters:
$$
p\ :\ \mathcal{C}_{\alpha_0}\to[0,1),
\quad
q\ :\ \mathcal{C}_{\alpha_0}\to[-1,0),
$$
which are natural with respect to $\mu$ such that $p$ is oriented counterclockwise and $q$ clockwise along $\mathcal{C}_{\alpha_0}$, and the values $p=0$ and $q=-1$ correspond to points $P_0$, $Q_0$ respectively, as shown in Figure \ref{fig:param}.
\begin{figure}[h]
\input{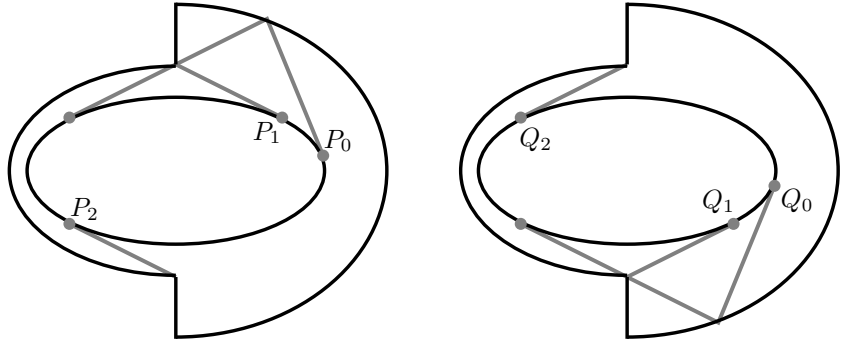}
\caption{Parametrizations of the caustic.}\label{fig:param}
\end{figure}

Consider one segment of a billiard trajectory, and let $X\in\mathcal{C}_{\alpha_0}$ be its touching point with the caustic.
Suppose that the particle is moving counterclockwise on that segment.
From Figure \ref{fig:param}, we conclude:
\begin{itemize}
\item
if $X$ is between points $P_1$ and $P_2$ then the particle is going to hit the arc
$\mathcal{C}_{\lambda_2}$;

\item
if $X$ is between $P_2$ and $P_0$, the particle is going to hit the arc
$\mathcal{C}_{\lambda_1}$;

\item
for $X$ between $P_0$ and $P_1$, the particle is going to hit $\mathcal{C}_{\lambda_1}$ and the upper segment before the next contact with the caustic and the direction of motion is changed to clockwise.
\end{itemize}
Similarly, if the particle is moving in clockwise direction, we have:
\begin{itemize}
\item
if $X$ is between points $Q_1$ and $Q_2$ then the particle is going to hit the arc
$\mathcal{C}_{\lambda_2}$;

\item
if $X$ is between $Q_2$ and $Q_0$, the particle is going to hit the arc
$\mathcal{C}_{\lambda_1}$;

\item
for $X$ between $Q_0$ and $Q_1$, the particle is going to hit $\mathcal{C}_{\lambda_1}$ and the lower segment before the next contact with the caustic and the direction of motion is changed to counterclockwise.
\end{itemize}

To see the billiard dynamics as an inteval exchange transformation, we make the following identification:
$$
(X,+)\sim p(X),
\quad
(X,-)\sim q(X).
$$
In other words:
\begin{itemize}
\item
we identify the joint point $X$ of a given trajectory with the caustic with 
$p(X)\in[0,1)$ if the particle is moving in the counterclockwise direction on the corresponding segment;

\item
for the motion in the clockwise direction, we identify $X$ with $q(X)\in[-1,0)$.
\end{itemize}

Denote the rotation numbers $r_1=\rho(\lambda_1)$, $r_2=\rho(\lambda_2)$ (see Proposition \ref{prop:rotation}).

The parametrizations values for points denoted in Figure \ref{fig:param} are:
\begin{gather*}
p(P_0)=0,
\quad
p(P_1)=r_1-r_2,
\quad
p(P_2)=r_1-r_2+\dfrac12,
\\
q(Q_0)=-1,
\quad
q(Q_1)=r_1-r_2-1,
\quad
q(Q_2)=r_1-r_2-\dfrac12.
\end{gather*}

Now, we distinguish three cases depending on the position of point $P_0$ with respect to the $x$-axis (see Figure \ref{fig:param}), i.e.~ on the sign of 
$\dfrac14+\dfrac{r_2}2-r_1$.

\subsubsection*{$P_0$ is on the $x$-axis: $\dfrac14+\dfrac{r_2}2-r_1=0$}
The interval exchange map is:
$$
\xi\mapsto
\begin{cases}
\xi+r_1+\frac32, &\xi\in[-1,-\frac12-r_1)\\
\xi+r_2, &\xi\in[-\frac12-r_1,-r_1)\\
\xi+r_1-1, &\xi\in[-r_1,0)\\
\xi+r_1-\frac12, & \xi\in[0,\frac12-r_1)\\
\xi+r_2, &\xi\in[\frac12-r_1,1-r_1)\\
\xi+r_1-1, &\xi\in[1-r_1,1),
\end{cases}
$$
as shown in Figure \ref{fig:interval0}.
\begin{figure}[h]
\input{0-figures/interval0}
\caption{Interval exchange transformation for the case $\frac14+\frac{r_2}2-r_1=0$.}\label{fig:interval0}
\end{figure}

To the map, pair $(\pi,\lambda)$ is joined:
\begin{gather*}
\pi=\left(
\begin{array}{cccccc}
A & B & C & D & E & F \\
C & B & D & F & E & A 
\end{array}
\right),
\\
\lambda=\left(\frac12-r_1,\ \frac12,\ r_1,\ \frac12-r_1,\ \frac12,\ r_1\right).
\end{gather*}

\subsubsection*{$P_0$ is above the $x$-axis: $\dfrac14+\dfrac{r_2}2-r_1>0$}
The interval exchange map in this case is shown in Figure \ref{fig:interval+} and given by:
$$
\xi\mapsto
\begin{cases}
\xi+r_1+\frac32, &\xi\in[-1,r_1-r_2-1)\\
\xi+r_2, &\xi\in[r_1-r_2-1,r_1-r_2-\frac12)\\
\xi+r_1, &\xi\in[r_1-r_2-\frac12,-r_1)\\
\xi+r_1-1, &\xi\in[-r_1,0)\\

\xi+r_1-\frac12, &\xi\in[0,r_1-r_2)\\
\xi+r_2, &\xi\in[r_1-r_2,r_1-r_2+\frac12)\\
\xi+r_1, &\xi\in[r_1-r_2+\frac12,1-r_1)\\
\xi+r_1-1, &\xi\in[1-r_1,1).
\end{cases}
$$
\begin{figure}[h]
\input{0-figures/interval+}
\caption{Interval exchange transformation for the case $\frac14+\frac{r_2}2-r_1>0$.}\label{fig:interval+}
\end{figure}

The map can be desribed by the pair $(\pi,\lambda)$:
\begin{gather*}
\pi=\left(
\begin{array}{cccccccc}
A & B & C & D & E & F & G & H\\
D & B & E & C & H & F & A & G 
\end{array}
\right),
\\
\lambda=\left(r_1-r_2,\ \frac12,\ r_2-2r_1+\frac12,\ r_1,\ r_1-r_2,\ \frac12,\ r_2-2r_1+\frac12,\ r_1\right).
\end{gather*}

\subsubsection*{$P_0$ is below the $x$-axis: $\dfrac14+\dfrac{r_2}2-r_1<0$}
The interval exchange map corresponding to the billiard dynamics is:
$$
\xi\mapsto
\begin{cases}
\xi+r_1+\frac32, &\xi\in[-1,-\frac12-r_1)\\
\xi+r_1-\frac12, &\xi\in[-\frac12-r_1,r_1-r_2-1)\\
\xi+r_2, &\xi\in[r_1-r_2-1,r_1-r_2-\frac12)\\
\xi+r_1-1, &\xi\in[r_1-r_2-\frac12,0)\\

\xi+r_1-\frac12, &\xi\in[0,\frac12-r_1)\\
\xi+r_1-\frac32, &\xi\in[\frac12-r_1,r_1-r_2)\\
\xi+r_2, &\xi\in[r_1-r_2,r_1-r_2+\frac12)\\
\xi+r_1-1, &\xi\in[r_1-r_2+\frac12,1),
\end{cases}
$$
see Figure \ref{fig:interval-}.
\begin{figure}[h]
\input{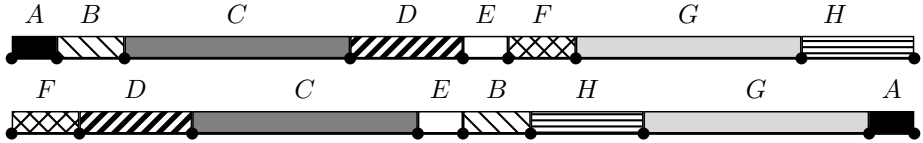}
\caption{Interval exchange transformation for the case $\frac14+\frac{r_2}2-r_1<0$.}\label{fig:interval-}
\end{figure}

To the map, pair $(\pi,\lambda)$ is joined:
\begin{gather*}
\pi=\left(
\begin{array}{cccccccc}
A & B & C & D & E & F & G & H\\
F & D & C & E & B & H & G & A 
\end{array}
\right),
\\
\lambda=\left(\frac12-r_1,\ 2r_1-r_2-\frac12,\ \frac12, r_2+\frac12-r_1,\ \frac12-r_1,\ 2r_1-r_2-\frac12,\ \frac12, r_2+\frac12-r_1\right).
\end{gather*}

Notice that in all three cases the interval exchange transformations depend only on the rotation numbers $r_1$, $r_2$.
Thus, we got

\begin{theorem}\label{th:nezavisnost}
The billiard dynamics inside the domain $\mathcal{D}_0$ with ellipse $\mathcal{C}_{\alpha_0}$ as the caustic, does not depend on the parameters $a$, $b$ of the confocal family but only on the rotation numbers $r_1$, $r_2$.
\end{theorem}

\subsubsection*{Domain bounded by elipses with rotation numbers $\dfrac{5-\sqrt5}{10}$ and $\dfrac{\sqrt5}{10}$}
As an example, let us analyze the billiard dynamics in a domain bounded by ellipses with rotation numbers $\dfrac{5-\sqrt5}{10}$ and $\dfrac{\sqrt5}{10}$.
By theorem \ref{th:nezavisnost} it is enough to consider the case when the confocal family is degenerate, i.e.~ consists of concentric circles.

In this example, there exist six saddle-connections, represented in Figure 
\ref{fig:c5-sadd}.
Each polygonal line shown on the figure corresponds to two trajectories in the phase space, depending on direction of the motion. 
\begin{figure}[h]
\input{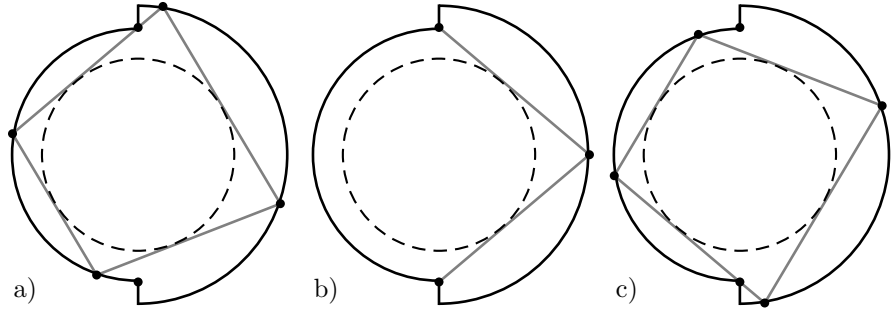}
\caption{Saddle-connections corresponding to circles with rotation numbers $\dfrac{5-\sqrt5}{10}$ and $\dfrac{\sqrt5}{10}$.}\label{fig:c5-sadd}
\end{figure}

Vertices of the saddle-connections divide the billiard border into eleven parts, see Figure \ref{fig:c5-parts}.
\begin{figure}[h]
\input{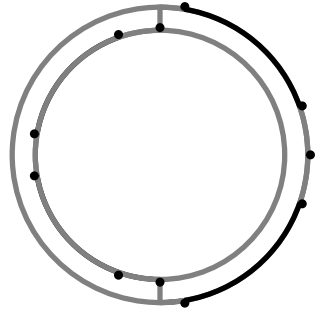}
\caption{Parts of the boundary corresponding to circles with rotation numbers $\dfrac{5-\sqrt5}{10}$ and $\dfrac{\sqrt5}{10}$.}\label{fig:c5-parts}
\end{figure}

All trajectories in this billiard domain corresponding to the fixed caustic are periodic:
\begin{itemize}
\item
either all bouncing points of a given trajectory are in gray parts -- in this case the billiard particle hits twice each gray part until the trajectory becomes closed and the trajectory is $14$-periodic,
see Figures \ref{fig:c5-periodic}a and \ref{fig:c5-periodic}c.
Notice that such a trajectory bounces six times on each of the circles and once on each of the segments; 

\item 
or all bouncing points are in black parts -- the particle will hit each part once until closure and the trajectory is $4$-periodic, see Figure \ref{fig:c5-periodic}b.
Such a trajectory reflects twice on each of the circular arcs.
\end{itemize}
\begin{figure}[h]
\input{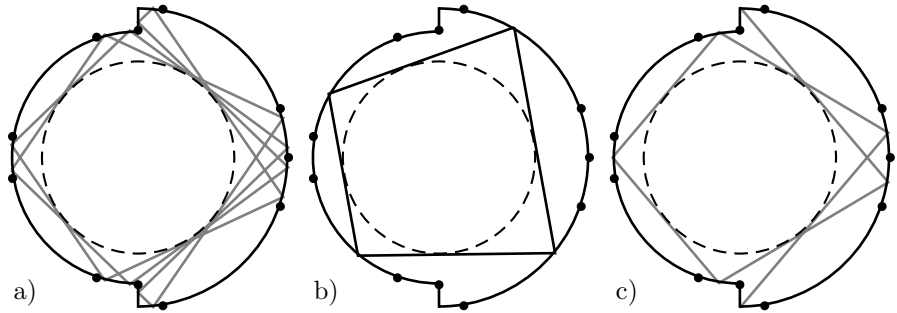}
\caption{Periodic trajectories corresponding to circles with rotation numbers $\dfrac{5-\sqrt5}{10}$ and $\dfrac{\sqrt5}{10}$.}\label{fig:c5-periodic}
\end{figure}

The corresponding level set in the phase space is divided by the saddle-connections into three parts:
\begin{itemize}
\item
the part containing all $14$-periodic trajectories: this part is bounded by four saddle-connections whose projections on the configuration space is shown on 
Figures \ref{fig:c5-sadd}a and \ref{fig:c5-sadd}c.
The saddle-connections corresponding to Figure \ref{fig:c5-sadd}b are lying inside this part;

\item
two parts containing all $4$-periodic trajectories winding about the caustic in the clockwise and counterclockwise direction: these parts are bounded by saddle-connections winding in the same direction whose projections on the configuration space is shown in Figures \ref{fig:c5-sadd}a and \ref{fig:c5-sadd}c.
\end{itemize}

\subsection{Keane condition and minimality}
\label{sec:keane}
An interval exchange transformation is called \emph{minimal} if every orbit is dense in the whole domain.
When considering pseudo-billiards, minimal interval exchange transformations will correspond to the cases when all orbits are dense in the domain between the billiard border and the caustic.

Following \cite{Viana, ZorichFLAT}, we are going to formulate a sufficient condition for minimality.
Let $f$ be an interval exchange transformation of $I$, given by pair $(\pi,\lambda)$.
Denote by $p_{\alpha}$ the left endpoint of $I_{\alpha}$.
Then the transformation satisfies \emph{the Keane condition} if:
$$
f^m(p_{\alpha})\neq p_{\beta}
\ \text{for all}\ m\ge1,\ \alpha\in\mathcal{A},\ \beta\in\mathcal{A}\setminus\{\pi_0^{-1}(1)\}. 
$$
Obviously, none of the transformations from Section \ref{sec:interval} satisfies the Keane condition:
namely, the midpoint of the interval is the left endpoint of one of $I_{\alpha}$, and it is the image of another endpoint in the corresponding interval exchange map.

The goal of this section is to find an analogue of the Keane condition for interval exhange transformations appearing in the billiard dynamics.

\subsubsection*{Billiard-like transformations and modified Keane condition}
Analysis of the examples from Section \ref{sec:interval} motivates the following definitions.

\begin{definition}\label{def:billiard-like}
An interval exchange transformation $f$ of $I=[-1,1)$ is \emph{billiard-like} if the partition into subintervals satisfies the following:
\begin{itemize}
 \item
for each $\alpha$, $I_{\alpha}$ is contained either in $[-1,0)$ or $[0,1)$;

 \item
for each $\alpha$, $f(I_{\alpha})$ is contained either in $[-1,0)$ or $[0,1)$;

 \item
both $[-1,0)$ and $[0,1)$ contain at least two intervals of the partition.
\end{itemize}
\end{definition}

\begin{definition}\label{def:keane}
We will say that a billiard-like interval exchange transformation $f$ satisfies \emph{the modified Keane condition} if
$$
f^m(p_{\alpha})\neq p_{\beta}
\ \text{for all}\ m\ge1,\ \alpha\in\mathcal{A},
\ \text{and}\ \beta\in\mathcal{B}\ \text{such that}\ p_{\beta}\not\in\{-1,0\}. 
$$
\end{definition}

\begin{lemma}\label{lema:keane.periodic}
If a billiard-like interval exchange transformation satisfies the modified Keane condition, then the transformation has no periodic points.
\end{lemma}

We say that an interval exchange transformation is \emph{irreducible} if for no $k<|\mathcal{A}|$ the union
$$
I_{\alpha_{\pi_0^{-1}(1)}}\cup\dots\cup I_{\alpha_{\pi_0^{-1}(k)}}
$$
is invariant under the transformation.
The usual Keane condition implies irreducibility.
However, this is not the case for the modified Keane condition -- it may happen that the transformation falls apart into two irreducible transformations on $[-1,0)$ and $[0,1)$.
On the other hand, if for a transformation satisfying the modified Keane condition there is an interval $I_{\alpha}\subset[-1,0)$ such that $f(I_{\alpha})\subset[0,1)$, the irreducibility will also take place.

\begin{proposition}\label{prop:keane.min}
If an irreducible billiard-like interval exchange transformation $f$ satisfies the modified Keane condition, then $f$ is minimal.
\end{proposition}

\subsubsection*{An example}
Consider billiard trajectories within domain $\mathcal{D}_0$ with the caustic $\mathcal{C}_{\alpha_0}$, as described in Section \ref{sec:few.confocal}.
In addition, suppose the rotation numbers corresponding to ellipses $\mathcal{C}_{\lambda_1}$ and $\mathcal{C}_{\lambda_2}$ are:
$$
r_1=\frac5{11}+\frac1{22\pi},\quad r_2=\frac5{11}-\frac1{220\pi}.
$$

With given rotation numbers, the Cayley-type conditions from Theorem \ref{th:cayley-pseudo} can be rewritten in a simpler form.
Namely, a necessary condition for existence of a trajectory within $\mathcal{D}_0$ which becomes closed after $n$ reflections of $\mathcal{C}_{\lambda_1}$ and $m$ reflections off $\mathcal{C}_{\lambda_2}$ is:
$$
nr_1+mr_2\in\mathbf{Z}.
$$
In this case, this condition is satisfied for $n=1$ and $m=10$:
\begin{equation}\label{eq:r1+10r2}
r_1+10r_2=5.
\end{equation}

Since $\dfrac14+\dfrac{r_2}2-r_1>0$, the corresponding interval exhange transformation is given by:
\begin{gather*}
\Pi=\left(
\begin{array}{cccccccc}
A & B & C & D & E & F & G & H\\
D & B & E & C & H & F & A & G 
\end{array}
\right),
\\
\lambda=\left(\frac1{20\pi},\ \frac12,\ \frac1{22}-\frac{21}{220\pi},\ \frac5{11}+\frac1{22\pi},
\ \frac1{20\pi},\ \frac12,\ \frac1{22}-\frac{21}{220\pi},\ \frac5{11}+\frac1{22\pi}\right).
\end{gather*}

\begin{proposition}\label{prop:keane-example}
The transformation $(\Pi,\lambda)$ satisfies the modified Keane condition.
\end{proposition}

\begin{proof}
Suppose that $p$ and $p'$ are two endpoints of the intervals such that $p'\not\in\{-1,0\}$ and $f^k(p)=p'$ for some $k\ge1$.
Notice that:
$$
p=\alpha r_1+\beta r_2+\gamma\frac12,\quad
p'=\alpha' r_1+\beta' r_2+\gamma'\frac12,
$$
for some
$\alpha,\alpha'\in\{-1,0,1\}$, $\beta,\beta'\in\{-1,0\}$, $\gamma,\gamma'\in\{-2,-1,0,1,2\}$.

We have:
$$
p'=f^k(p)=p+k_1r_1+k_2r_2+k_3\frac12,
$$
for some integers $k_1$, $k_2$, $k_3$ such that $k_1+k_2=k$, $k_1\ge0$, $k_2\ge0$.
Thus:
\begin{equation}\label{eq:k1r1+k2r2}
(k_1+\alpha-\alpha')r_1+(k_2+\beta-\beta')r_2+(k_3+\gamma-\gamma')\frac12=0.
\end{equation}
Since $r_1$ and $r_2$ are irrational, equations (\ref{eq:r1+10r2}) and (\ref{eq:k1r1+k2r2}) must be dependent:
\begin{equation}\label{eq:a}
a:=k_1+\alpha-\alpha'=\frac1{10}(k_2+\beta-\beta')=-\frac1{10}(k_3+\gamma-\gamma').
\end{equation}
For each $\xi\in B\cup F$, either $f(\xi)$ or $f^2(\xi)$ are not in $B\cup F$, thus
\begin{equation}\label{eq:relk1k2}
k_2\le2k_1+2.
\end{equation}
Combining (\ref{eq:relk1k2}) and (\ref{eq:a}) we get $8a\le7$.
Since $k_2$ is non-negative, (\ref{eq:relk1k2}) gives that $a=0$, which leads to $k=k_1+k_2\le3$.
By direct calculation we check that none of the partition interval endpoints is mapped into another one, different from $-1$ and $0$ by at most three iterations.
\end{proof}

In this example, although the Cayley-type conditon for periodicity is satisfied, not only that closed trajectories do not exist, but each of the trajectories densely fills the ring between the billiard border and the caustic.

\subsection{Unique ergodicity}
\label{sec:unique.erg}
In this section it will be shown that there is infinitely many billiard tables bounded by arcs of confocal conics, such that the corresponding flow will not be uniquely ergodic.

Consider the billiard table $\mathcal{D}_1$ whose boundary consists of ellipse $\mathcal{C}_{\beta_1}$ from (\ref{eq:confocal-in-plane}) and segment $\Gamma=\{(0,y)\mid\sqrt{b-\beta_2}\le y\le\sqrt{b-\beta_1}\}$, with $\beta_1<\beta_2<b$: $\partial\mathcal{D}_1=\mathcal{C}_{\beta_1}\cup\Gamma$, see Figure \ref{fig:billiard2}.
\begin{figure}[h]
\input{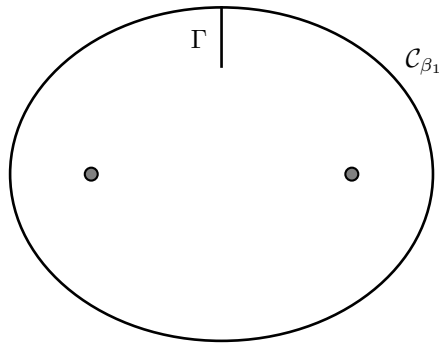}
\caption{Domain $\mathcal{D}_1$ within an ellipse having a ``wall'' on the $y$-axis.}\label{fig:billiard2}
\end{figure}

Fix parameter $\alpha_0$: such that $\beta_2<\alpha_0<b$, and take $\mathcal{C}_{\alpha_0}$ to be a caustic.
A corresponding billiard trajectory is shown in Figure \ref{fig:caustic2}.
\begin{figure}[h]
\input{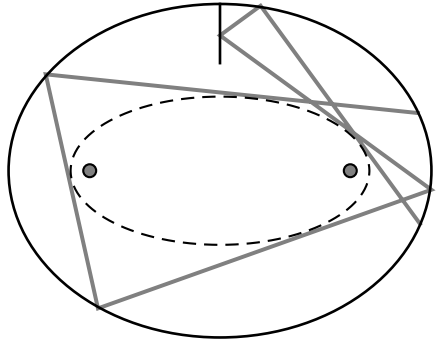}
\caption{A billiard trajectory in $\mathcal{D}_1$ with an ellipse as caustic.}\label{fig:caustic2}
\end{figure}

\begin{proposition}\label{prop:iet}
Billiard flow within domain $\mathcal{D}_1$ and caustic $\mathcal{C}_{\alpha_0}$ is equivalent to the following exchange transformation of the interval $[-1,1)$:
\begin{equation}\label{eq:iet}
 \xi\mapsto
\begin{cases}
\xi+r_1, &\xi\in[-1,-r_1)\\
\xi+r_1-1, &\xi\in[-r_1,r_2-r_1)\\
\xi+r_1, & \xi\in[r_2-r_1,1-r_1)\\
\xi+r_1-1, &\xi\in[1-r_1,r_2-r_1+1)\\
\xi+r_1-2, &\xi\in[r_2-r_1+1,1),
\end{cases}
\end{equation}
with $r_1=\rho(\beta_1)$, $r_2=\rho(\beta_2)$, and $\rho$ is the corresponding rotation function, see Proposition \ref{prop:rotation}.
\end{proposition}

\begin{proof}
The billiard flow is equivalent to the discrete dynamics of touching points of the trajectory with the caustic, with the direction of motion taken into account.

Introduce metric $\mu$ on the caustic $\mathcal{C}_{\alpha_0}$ as in Proposition \ref{prop:rotation}.
Then, we parametrize $\mathcal{C}_{\alpha_0}$ by parameters:
$$
p\ :\ \mathcal{C}_{\alpha_0}\to[0,1),
\quad
q\ :\ \mathcal{C}_{\alpha_0}\to[-1,0),
$$
which are natural with respect to $\mu$ such that $p$ is oriented counterclockwise and $q$ clockwise along $\mathcal{C}_{\alpha_0}$, and the values $p=0$ and $q=-1$ correspond to touching points, contained in the right half-plane and left-half plane respectively, of tangential lines from $(0,\sqrt{b-\beta_2})$.
Having in mind that reflection on the ``wall'' $\Gamma$ changes the orientation of motion, we obtain (\ref{eq:iet}).
\end{proof}

Map (\ref{eq:iet}) is represented by  pair $(\pi,\lambda)$:
\begin{gather*}
\pi=\left(
\begin{array}{ccccc}
A & B & C & D & E \\
B & E & A & D & C 
\end{array}
\right),
\\
\lambda=(1-r_1,\ r_2,\ 1-r_2,\ r_2,\ r_1-r_2),
\end{gather*}
see also Figure \ref{fig:interval1}.
\begin{figure}[h]
\input{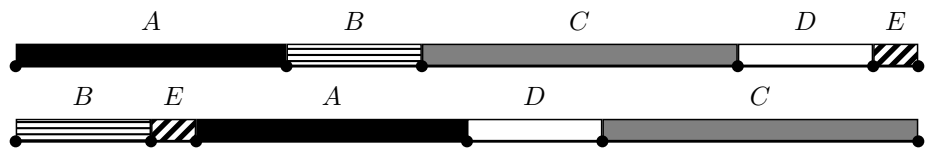}
\caption{Transformation corresponding to the billiard within $\mathcal{D}_1$.}\label{fig:interval1}
\end{figure}

\begin{theorem}\label{th:ergodic}
There are billiard tables $\mathcal{D}_1$ and caustics, such that the corresponding billiard flows are minimal and not uniquely ergodic.
\end{theorem}

\begin{proof}
The transformation (\ref{eq:iet}) corresponds to the Veech example of minimal and not uniquely ergodic systems \cite{Veech1969}, see also \cite{MasurTab2002}.
Namely, choose $\alpha_0$ and $\beta_1$ such that $r_1=\rho(\beta_1)$ is an irrational number with unbounded partial quotients.
Then there are irrational numbers $r$, such that for $r_2=\rho(\beta_2)=r_1-r$, the measure $\mu$ on $\mathcal{C}_{\alpha_0}$ is not ergodic, thus not uniquely ergodic.
\end{proof}


\end{document}